\documentclass[12pt]{amsart}

\usepackage[letterpaper,margin=3cm]{geometry} 
\usepackage[onehalfspacing]{setspace}
\allowdisplaybreaks

\usepackage{xcolor}

\usepackage{hyperref}

\usepackage[english]{babel}
\usepackage[T1]{fontenc}

\usepackage{graphicx}
\usepackage{grffile}    
\usepackage{graphics}
\usepackage{graphpap}  
\usepackage{tikz}
\usetikzlibrary{arrows,automata,positioning,calc,shapes,decorations.pathreplacing,
decorations.markings,shapes.misc,petri,topaths,plotmarks}
\usepackage{tikz-3dplot}
\usepackage{pgfplots}
\pgfplotsset{compat=newest}
\newlength\figureheight
\newlength\figurewidth
\pgfplotsset{
  tick label style={font=\scriptsize},
  label style={font=\footnotesize},
  legend style={font=\footnotesize},
  every axis plot/.append style={very thick}
}

\usepackage{amsmath,amssymb,amsfonts,amsthm}
\usepackage{mathtools}
\usepackage{mathrsfs}
\usepackage{stmaryrd}
\usepackage{dsfont}
\usepackage{bm} 
\usepackage{nicefrac}
\usepackage{paralist}
\usepackage[inline]{enumitem}
\usepackage{multirow}
\usepackage{booktabs}
\usepackage{url}
\usepackage{mathrsfs}
\usepackage{spverbatim}

\usepackage[authoryear]{natbib}

\usepackage{xcolor}
\usepackage{comment}
\usepackage{mdframed}
\usepackage{xmpmulti}
\usepackage{grffile}

\newcommand{\bbN}{\mathbb{N}}

\newcommand{\bbR}{\mathbb{R}}

\newcommand{\cD}{\mathcal{D}}

\newcommand{\cH}{\mathcal{H}}
\newcommand{\cI}{\mathcal{I}}

\newcommand{\cL}{\mathcal{L}}

\newcommand{\cN}{\mathcal{N}}

\newcommand{\normiii}[2]{\vert\kern-0.25ex\vert\kern-0.25ex\vert #1 \vert\kern-0.25ex \vert\kern-0.25ex\vert_{ #2 }}
\newcommand{\Normiii}[2]{\left\vert\kern-0.25ex\left\vert\kern-0.25ex\left\vert #1 \right\vert\kern-0.25ex\right\vert\kern-0.25ex\right\vert_{ #2 }}

\newcommand{\meshwidth}{\delta}

\newcommand{\mv}[1]{{\boldsymbol{\mathrm{#1}}}}

\newcommand{\eps}{\varepsilon}

\DeclareMathOperator{\diag}{diag}

\newcommand{\indep}{\perp \!\!\! \perp}

\newcommand{\HR}{H\"usler--Reiss}
\newcommand{\BR}{Brown--Resnick}
\newcommand{\WM}{Whittle--Mat\'ern}

\newtheorem{theorem}{Theorem}[section]
\newtheorem{lemma}[theorem]{Lemma}
\newtheorem{proposition}[theorem]{Proposition}

\newtheorem{remark}[theorem]{Remark}
\newtheorem{definition}[theorem]{Definition}

\newtheorem{example}[theorem]{Example}

\usetikzlibrary{calc,decorations.pathreplacing}

\title[Intrinsic Whittle--Mat\'ern fields and sparse spatial extremes]{Intrinsic Whittle--Mat\'ern fields and sparse \\ spatial extremes}

\author{David Bolin}
\author{Peter Braunsteins}
\author{Sebastian Engelke}
\author{Rapha\"el Huser}

\address{Statistics Program, King Abdullah University of Science and Technology, 23955-6900 Thuwal, Saudi Arabia}
\email{david.bolin@kaust.edu.sa}

\address{School of Mathematics and Statistics, University of New South Wales, Sydney NSW 2052, Australia}
\email{p.braunsteins@unsw.edu.au}

\address{Research Institute for Statistics and Information Science, University of Geneva, 1205 Geneva, Switzerland}
\email{sebastian.engelke@unige.ch}

\address{Statistics Program, King Abdullah University of Science and Technology, 23955-6900 Thuwal, Saudi Arabia}
\email{raphael.huser@kaust.edu.sa}

\date{} 


\begin{document}

\begin{abstract}
Intrinsic Gaussian fields are used in many areas of statistics as models for spatial or spatio-temporal dependence, or as priors for latent variables. However, there are two major gaps in the literature: first, the number and flexibility of existing intrinsic models are very limited; second, theory, fast inference, and software are currently underdeveloped for intrinsic fields. We tackle these challenges by introducing the new flexible class of intrinsic \WM{} Gaussian random fields obtained as the solution to a stochastic partial differential equation (SPDE). Exploiting sparsity resulting from finite-element approximations, we develop fast estimation and simulation methods for these models. We demonstrate the benefits of this intrinsic SPDE approach for the important task of kriging under extrapolation settings. Leveraging the connection of intrinsic fields to spatial extreme value processes, we translate our theory to an SPDE approach for \BR{} processes for sparse modeling of spatial extreme events. This new paradigm paves the way for efficient inference in unprecedented dimensions. To demonstrate the wide applicability of our new methodology, we apply it in two very different areas: a longitudinal study of renal function data, and the modeling of marine heat waves using high-resolution sea surface temperature data.
\end{abstract}

\maketitle

\section{Introduction}

Intrinsic Gaussian random fields are important models in areas ranging from geoenvironmental science \citep{lombardo2020spacetime} to disease mapping \citep{wakefield2007disease,Konstantinoudis2022regional} and brain imaging \citep{penny2005bayesian}.
They are improper in the sense that their finite-dimensional distributions are characterized by precision matrices that are not full rank. Here we focus intrinsic fields with a precision matrix characterized by zero row sums, making their densities invariant to the addition of a constant function.
These are a natural choice as a prior distribution in Bayesian statistics, since they impose no prior information on the overall average of the field but bias towards smoothness that captures spatial information. However, several factors prevent intrinsic fields from being further used in practice: (1) there are few existing parametric models for intrinsic fields and those that exist impose restrictive assumptions that are rarely desirable in practice; (2) theory, computationally efficient inference and software implementations developed for proper fields do not exist for intrinsic ones; and (3) the full potential of intrinsic fields in other areas of statistics has not yet been exploited.

For non-intrinsic, proper Gaussian random fields, the most commonly used model in applications is the class of Gaussian Whittle--Mat\'ern fields. These can be viewed as solutions $u$ to the SPDE \citep{whittle63}
\begin{align}\label{eq:maternspde}
	(\kappa^2-\Delta)^{\alpha/2} (\tau u) = \mathcal{W} \quad \text{ on $\mathbb{R}^d$},
\end{align}
where $\mathcal{W}$ denotes Gaussian white noise, and $\Delta = \sum_{k=1}^d {{\rm d}^2/{\rm d} s_k^2}$ is the Laplacian.
The parameters $\tau,\kappa>0$, $\alpha>d/2$ control the most important aspects of the Gaussian field, namely the variance, the practical correlation range, and the smoothness, respectively. \citet{Lin11} proposed the ``SPDE approach'' to spatial statistics that constructs a computationally efficient Gaussian Markov random field (GMRF) approximation of these fields with $\alpha\in\mathbb{N}$ by solving~\eqref{eq:maternspde} on a suitable mesh through a finite element method (FEM). These approximations have sparse precision matrices, whose zero entries encode conditional independencies, which typically reduces the computational cost for inference from being cubic in the number of observation locations to scaling like $O(N^{3/2})$, where $N$ is the number of mesh nodes. The rational SPDE approach \citep{BK2020rational,xiong2022} removed the restriction to ${\alpha\in\mathbb{N}}$, enabling efficient SPDE approximations for Gaussian fields with general smoothness $\alpha>d/2$. This relaxation is crucial for the model's predictive performance \citep{stein99, kb-kriging}. The computational advantage combined with the fact that the  approach facilitates extensions of Whittle--Mat\'ern fields to non-Euclidean domains \citep{mejia2020bayesian,BSW2022}, non-stationary models \citep{fuglstad2015exploring,bakka2019nonstationary}, and non-Gaussian models \citep{bolin14,bolin2020multivariate} has made it highly successful \citep{lindgren2022spde}.

The most commonly applied intrinsic Gaussian field is the intrinsic conditional autoregressive \citep[iCAR;][]{besag1974spatial,besag1975statistical} model,
whose sparse precision is a key reason for its popularity; however, it is exceptionally rigid with just a single parameter $\tau$, and the entries of the off-diagonal entries of the precision matrix can take only two values, so that not all spatial information is taken into account.
A continuous-space model which accounts for spatial distances can be defined through the partial differential equation
\begin{align}\label{eq:intrinsic1}
	(-\Delta)^{\beta/2} (\tau u) = \mathcal{W} \quad \text{ on $\mathbb{R}^d$},
\end{align}	
where $\tau>0$. For $d=2$ and $\beta\rightarrow 1$, this model is referred to as the de Wijs process \citep{besag2005first}, and the standard iCAR model on an evenly spaced lattice in $\mathbb{R}^2$ can be seen as a discretization thereof \citep{Lin11}. The de Wijs process has a single parameter and its ``smoothness'' is equal to zero causing outcomes of the field to have no pointwise meaning. If $\beta\in (d/2,1+d/2)$, then \eqref{eq:intrinsic1} gives an intrinsic Gaussian field with the fractional variogram $\text{Var}\{u(s) - u(t)\} \propto \tau^{-2} \|\mv s-\mv t\|^{2\beta-d}$. This model is still restrictive in that both the short- and long-range growth of the variogram is determined by a single parameter $\beta$. Moreover the rational SPDE approach has not been extended to intrinsic fields such as the fractional model and thus the computational cost for likelihood inference typically grows cubically in the number of locations.   

In this paper, we introduce the new class of intrinsic Whittle--Mat\'ern fields as the Gaussian random fields $u$ obtained as a solution to 
\begin{align}\label{SPDE}
    (-\Delta)^{\beta/2}(\kappa^2 - \Delta)^{\alpha/2} (\tau u) = \mathcal{W} \quad \text{ on $\mathbb{R}^d$},
\end{align}
with $\kappa,\tau>0$ and $\alpha+\beta>d/2$.
This model bridges the classical (non-intrinsic) Gaussian Whittle--Mat\'ern fields when $\beta=0$ with the fractional (intrinsic) Gaussian fields when ${\alpha=0}$. When $\beta>0$, this model provides a flexible class of intrinsic Gaussian fields, where $\alpha$ and $\beta$ control the behavior of the variogram at $0$ and at $\infty$, respectively, and $\kappa$ determines the rate at which the variogram transitions from its local to global behavior, while $\tau>0$ is a multiplicative constant. 
To allow for sparse computations, we establish a new rational SPDE approach for intrinsic fields. This yields a GMRF approximation of the solution to~\eqref{SPDE} for which we provide a rigorous convergence analysis. Our approach enables computationally efficient inference for the parameters of these fields, including estimation of the two fractional smoothness parameters. It also lays the foundation for the numerous extensions that have already been developed for proper Whittle--Mat\'ern fields. 

Prediction at unobserved locations based on nearby observations is  a highly important task in many applications and traditionally performed using kriging \citep{cressie2015statistics}.
While proper, non-intrinsic fields yield kriging estimates 
that revert to the overall mean, intrinsic models
such as our intrinsic \WM{} model allow for a much broader
behavior, particularly when extrapolating beyond the 
observation range.
We quantify this extrapolation behavior theoretically and
show that our new model outperforms proper fields
and the fractional model on clinical kidney function data.

As a key application of the new intrinsic \WM{} model, we consider the field of spatial extremes \citep[e.g.,][]{dav2012b,davison2019spatial,huser2022advances}, which is concerned with the modeling of low-probability, high-impact events in space, such as heat waves \citep{hea2024}, storms \citep{defondeville2022functional} or floods \citep{asadi2015extremes, zhong2025spatial}. 
The most popular model class are Brown--Resnick random fields
\citep{Kab09}, the analogue of Gaussian processes in the field of spatial extremes parameterized by a variogram $\gamma$. 
Current statistical models often use the fractional variograms~\eqref{eq:intrinsic1} and have important limitations:
they lack flexibility in terms of local and global behavior;
they are not easily generalizable to non-stationary or non-Euclidean settings;
and, most importantly, they do not allow for sparse likelihood estimation.
In fact, sparse modeling in extremes
has so far been limited to 
extremal graphical models in purely multivariate settings \citep{Eng20},
with focus on structure estimation rather than efficient likelihood inference \citep{engelke2022structure, wan2023graphical, hentschel2024statistical, engelke2025learning}. When observing a \BR{} process at a finite number of locations, 
the distribution does generally yield dense, fully connected extremal graphical 
models, thus hindering sparse computations; see Appendix~\ref{sec:Prop}. 
While many alternative methods for efficient and scalable inference have been developed over the past two decades \citep{wadsworth2014efficient,Eng15,dombry2017bayesian,defondeville2018high,huser2019full,huser2024vecchia}, the lack of sparse models for spatial extremes has hampered their application in high dimensional settings.

We introduce an SPDE approach that lays the foundation for sparse modeling of spatial extremes and that addresses the current limitations in the field discussed above. It builds on the class of \BR{} processes parameterized by 
variograms of the new intrinsic Whittle--Mat\'ern fields~\eqref{SPDE}. We derive (provably consistent) FEM approximations of these  processes that have sparse representations in terms of extremal graphical models.
This enables likelihood computations at a complexity that is
similar to that of traditional SPDE-based Gaussian processes, thus paving the way for efficient inference in unprecedented dimensions. Importantly, this new spatial extremes model inherits the high flexibility 
in capturing both local and global dependence behaviors from the intrinsic Whittle--Mat\'ern field. 
Finally, our general modeling framework naturally opens the door to modeling extremes in unexplored complex settings, including extremes over non-Euclidean domains and with a non-stationary dependence structures.

In Section~\ref{sec:model}, we introduce intrinsic Whittle--Mat\'ern fields, lay out their theoretical foundation and derive key properties. Section~\ref{sec:inference} introduces sparse FEM approximations for computationally efficient inference.
Based on these results, in Section~\ref{sec:kriging}
we discuss how kriging can be improved using our new intrinsic model class, and present the kidney data application. Section~\ref{sec:extremes} 
proposes the ``SPDE approach'' for
computationally efficient and flexible modeling of spatial extremes, and showcases its potential for modeling marine heat waves. Technical details and proofs are presented in 8 appendices. All methods introduced here are implemented in the \texttt{R} \citep{Rlanguage} package \texttt{rSPDE} \citep{rSPDE}, which has interfaces to the widely used \texttt{R} packages \texttt{INLA} \citep{inla} and \texttt{inlabru} \citep{inlabru} for Bayesian inference. Thus, the proposed models can be incorporated in general latent Gaussian models which can be fitted to data in a fully Bayesian framework. A tutorial can be found at  \url{https://davidbolin.github.io/rSPDE/articles/intrinsic.html}.

\section{Intrinsic Gaussian Whittle--Mat\'ern fields}\label{sec:model}

In this section we first review intrinsic GMRFs. 
We then define the new class of intrinsic Gaussian Whittle--Mat\'ern fields and lay out the mathematical foundation.  
Finally, we derive the corresponding variogram functions $\gamma$ and some key theoretical properties. 

\subsection{Intrinsic Gaussian Markov random fields}\label{sec:intrinsic}

A centered $k$-dimensional Gaussian random vector $\mathbf{W}$ can be parameterized by a symmetric positive definite \emph{precision matrix} $Q$.
The precision matrix $Q$ describes the full conditional distributions of $\mathbf{W}$. Writing $\mathbf{W}_{-i}$ for the vector $\mathbf{W}$ with the $i$th component removed, the conditional distribution of $W_i\mid\mathbf{W}_{-i}$ is indeed (univariate) Gaussian with mean and variance given by
\begin{equation}\label{eqn:Conde}
\begin{aligned}
\mathbb{E}(W_i \mid  \mathbf{W}_{-i}) &= - \frac{\sum_{j: j\neq i} Q_{ij} W_j}{Q_{ii}}, \qquad \text{Var}(W_i \mid \mathbf{W}_{-i}) &= Q_{ii}^{-1}.
\end{aligned}
\end{equation}
Importantly, $Q_{ij}=0$ if and only if  $W_i$ and $W_j$ are conditionally independent given all other variables, i.e., 
if ${W}_i \indep W_j \mid \mathbf{W}_{-ij}$,
with $\mathbf{W}_{-ij}$ the subvector of $\mathbf{W}$ with the $i$th and $j$th components removed. 
This conditional independence pattern can be represented by a graph $(\mathcal{V}, \mathcal{E})$ with vertex set $\mathcal{V}$ and edge set $\mathcal{E}$, where $\{i,j\} \in \mathcal{E}$ if $Q_{ij} \neq 0$.
In this case, $\mathbf{W}$ is referred to as a (proper) GMRF with respect to  the conditional independence graph $(\mathcal{V}, \mathcal{E})$.

Some popular GMRF models are defined in terms of the weights of the neighbors in the expression for $\mathbb{E}(W_i \mid\, \mathbf{W}_{-i})$. For example iCAR models are defined so that these weights sum to 1.
This implies a zero rowsum constraint on $Q$, $Q \mathbf{1}= \mathbf{0}$, which means that $Q$ is not positive definite, and therefore $\mv{W}$ is not a proper field. This motivates the definition of intrinsic GMRFs \citep[Definition 3.2]{Rue05}.

\begin{definition}\label{def:FOI}
Let $\Theta$ be a $k \times k$ symmetric positive semi-definite matrix with rank $k-1$ 
and $\Theta \mathbf{1}=\mathbf{0}$. Then $\mathbf{W}$ is a centered first-order intrinsic GMRF if its density is
\[
f(\mathbf{w}) = (2 \pi)^{-(k-1)/2} (|\Theta|^*)^{1/2} \mathrm{exp} \left( -\frac{1}{2} \mathbf{w}^\top \Theta \mathbf{w} \right),\qquad \mathbf{w}\in\mathbb{R}^k,
\]
with $|\Theta|^*$ the generalized determinant of $\Theta$ (the product of its non-zero eigenvalues). 
\end{definition}

We observe that the integral of $f(\mathbf w)$ over its domain is infinite since it is constant along the direction of the all-ones vector, that is, $f(\mathbf w + c \mathbf 1) = f(\mathbf w)>0$ for any $c\in\mathbb R$ and $\mathbf{w}\in \mathbb{R}^k$. 
Thus, $f(\mathbf{w})$ is not a density on $\mathbb R^k$, but it can be interpreted as a density on the subspace of $\mathbb R^k$ orthogonal to the null-space of $\Theta$ \citep{Rue05, bolin2021efficient}.
One can characterize the distribution of $\mathbf{W}$ through the joint distribution of its differences $W_i -W_j$, $i,j\in \{1,\dots, k\}$; see Appendix \ref{App:Intrinsic_intro}. The distribution of $\mathbf W$ can thus also be defined through the conditionally negative definite variogram matrix $\Gamma$ with entries
$\Gamma_{ij} = \text{Var}(W_i - W_j)$, for $i,j =1,\dots, k$ \citep[][Prop.~3.3]{hentschel2024statistical}.

Similarly to proper GMRFs, we say that a first-order intrinsic GMRF  $\mathbf{W}$ is an improper (since it does not have a density in the usual sense) GMRF with respect to the labeled graph $(\mathcal{V},\mathcal{E})$, if for all $i\neq j$, 
$\Theta_{ij} = 0$ if and only if  $\{i,j\} \notin \mathcal{E}$.

Having defined first-order intrinsic GMRFs on $\mathbb{R}^k$, we can naturally define the analogue for continuously indexed Gaussian random fields. 

\begin{definition}\label{def:FOI2}
A random field $u = (u(\mv s): \mv s \in \mathcal D)$ with domain $\mathcal D \subset \mathbb R^d$ is called a first-order intrinsic Gaussian random field if all finite-dimensional distributions are first-order intrinsic GMRFs.
\end{definition}

In the continuous domain case, it is common to characterize the field through the conditionally negative definite variogram function
\begin{align}\label{vario_def}
    \gamma(\mv s,\mv t) = \text{Var}\{u(\mv s) - u(\mv t)\}, \quad \mv s, \mv t \in \mathcal D.
\end{align}
If $\gamma(\mv s, \mv t)$ depends only on $\lVert \mv s - \mv t \rVert$ then the field is intrinsically stationary and isotropic. While intrinsic fields of order higher than one can be defined, from here on we use the term ``intrinsic field''  to refer to a first-order intrinsic field, for notational simplicity. 

\subsection{Formal definition of intrinsic  Whittle--Mat\'ern fields}\label{sec:intrinsic_matern}

In a spatial context, $\mathbf W$ corresponds to observations at $k$ locations $\mv s_1,\dots, \mv s_k \in \mathcal D \subset \mathbb R^d$, i.e., $\mathbf{W}=(u(\mv s_1), \dots,u(\mv s_k))$. When using pairwise distances between these locations to parameterize the precision matrix $\Theta$ (and thus the distribution of $\mathbf{W}$), a natural way to ensure $\Theta$ is a non-negative definite rank $k-1$ matrix with $\Theta \boldsymbol{1}=\boldsymbol{0}$ is to define a continuous intrinsic Gaussian random field $u$ on $\mathcal D$ as in Definition~\ref{def:FOI2} through an SPDE and let $\Theta$ be characterized by the corresponding variogram matrix. 

To give a formal definition of the intrinsic Whittle--Mat\'ern fields, we note that the intrinsic property of these fields comes from the fact that \eqref{SPDE} does not have a unique solution. Indeed, if $u_0$ solves the equation, then any field $u'_0 = u_0 + c$ is also a solution, where $c$ is any (possibly random) constant function on $\mathcal D$. Equation~\eqref{SPDE} therefore defines an equivalence class of fields all sharing the same variogram since the constant function $c$ cancels in the computation of the variogram in~\eqref{vario_def}. The reason for this is that $\Delta$ applied to any constant (or linear) function is zero, which means that these functions are in the kernel of the operator and satisfy $\Delta f = 0$. 

The model \eqref{SPDE} can be defined on $\mathcal{D} = \mathbb{R}^d$ through Fourier techniques as done for the proper SPDEs in \cite{Lin11}. For a bounded domain $\mathcal{D} \subset \mathbb{R}^d$, we use a method similar to that in \cite{BK2020rational}. Specifically, suppose that $\mathcal{D}$ is an open, bounded, convex polytope in $\bbR^d$, with closure $\overline{\mathcal{D}}$, and let $L_2(\cD)$ denote the Lebesgue space of square-integrable functions on $\cD$ with inner product $(u,v)_{L_2(\cD)} = \int_\cD u(s) v(s) {\rm d}s$. In this case, the domain of 
$\Delta$ is the standard Sobolev space $H^2(\mathcal{D})$, which contains twice (weakly) differentiable functions in $L_2(\cD)$. To obtain a first-order intrinsic field, only constant functions should be in the kernel of the operator; since linear functions also belong to the kernel, we apply homogeneous Neumann boundary conditions to set derivatives of the field to zero on the boundary of $\mathcal D$, thereby excluding linear non-constant functions from the kernel. 
We thus restrict the domain of $\Delta$ to 
$\cH_{\cN}^2(\cD) = \{f \in H^2(\cD) : \partial f/ \partial v  = 0\}$, 
where $v$ is the outward unit normal vector to~$\partial\cD$. 

It is well-known \citep{xiong2022} that the operator $-\Delta$ on the domain $\cH_{\cN}^2(\cD)$ admits an eigendecompostion with eigenfunctions $\{e_j\}_{j\in\mathbb{N}}$ that form an orthonormal basis of $L_2(\cD)$ and corresponding non-negative, non-decreasing eigenvalues $\{\lambda_j\}_{j\in\mathbb{N}}$, i.e., $-\Delta e_j = \lambda_j e_j$. In addition, $\lambda_1 = 0$ is the only zero eigenvalue
and it corresponds to the eigenfunction $e_1 \propto 1$, 
which reflects the fact that only constant functions are in the kernel of $-\Delta$. 
In certain situations these eigenvalues and eigenfunctions can be expressed explicitly. The following important example will be used later. 

\begin{example}\label{ex:cube}
Suppose that $\cD = [0,L]^d$ is a square in $\mathbb{R}^d$. Let $\mv{j} = (j_1,\ldots, j_d) \in \mathbb{N}_0^d$ be a multi-index and $\|\mv{j}\|^2 = j_1^2+\ldots + j_d^2 > 0$. Then, in addition to $e_\mv{0} \propto 1$ with eigenvalue $\lambda_\mv{0}=0$, the eigenfunctions and eigenvalues of $-\Delta$ are 
\begin{align*}
e_\mv{j}(s) &= \prod_{i=1}^d\frac{(\sqrt{2})^{1(j_i>0)}}{\sqrt{L}}\cos(j_i\pi s_i/L), \quad {\lambda}_{\mv{j}} = \|\mv{j}\|^2\pi^2/L^2,
\end{align*}
see \citet[][Chapter~VI.4]{courant1953methods}. Note that we can order the eigenvalues $\{\lambda_{\mv{j}} \}$ so they are non-decreasing with ties broken arbitrarily to obtain the sequence $\{ \lambda_j\}_{j \in \mathbb{N}}$ described above, with $\{e_j\}_{j \in \mathbb{N}}$ then the corresponding eigenfunctions. 
\end{example}

Using these eigenvalues and eigenfunctions, we can define the fractional operators $(-\Delta)^{\beta/2}$ and $(\kappa^2 - \Delta)^{\alpha/2}$ in \eqref{SPDE} by their spectral decompositions---see Appendix~\ref{app:fractional}---and define the intrinsic Whittle--Mat\'ern fields on bounded domains as follows.

\begin{definition}\label{def_WM}
    An intrinsic Whittle--Mat\'ern field with parameters $\beta \in (0,2)$ and  $\kappa, \alpha, \tau \geq 0$ is an intrinsic Gaussian random field $u = (u(\mv s): \mv s\in\mathcal D)$ solving
    \begin{align}\label{SPDE_def}
        (-\Delta)^{\beta/2}(\kappa^2 - \Delta)^{\alpha/2} (\tau u) = \mathcal{W}, \quad \text{on $\cD$},
    \end{align}
    where $\mathcal W$ is Gaussian white noise and $\Delta$ is the  Neumann Laplacian.
\end{definition}

Recall that the equivalence class of an intrinsic field $u$ consists of all fields obtained by adding constant functions. To obtain a unique solution to \eqref{SPDE_def}, we may, for example, restrict to random fields that integrate to zero and are thus defined on the space orthogonal to the kernel of $-\Delta$; see Appendix~\ref{App:Intrinsic_intro} for an explanation in the context of intrinsic GMRFs. These representative fields have the same variogram as any other solution to the equation. We introduce the Hilbert space $(H, \|\cdot \|_H)$ of mean-zero $L_2(\cD)$-functions, i.e., $H = \{f\in L_2(\cD) : \oint_{\cD} f(s) {\rm d}s = 0\}$ with $\|\cdot\|_H = \|\cdot\|_{L_2(\cD)}$, where the inner product is $(\cdot,\cdot)_H = (\cdot, \cdot)_{L_2(\cD)}$. Here $\oint_{\cD} f(\mv s) {\rm d}\mv{s} = |\cD|^{-1} \int_{\cD} f(\mv s) {\rm d}\mv{s}$ denotes the spatial average of the function $f$. 
We then define $\widetilde{\Delta} :  \cH_{\cN}^2(\mathcal{D})\cap H \subset H \to H$ as the Neumann Laplacian with the mean zero constraint. Restriction  to the space $H$ orthogonal to the kernel implies that the eigenfunctions and eigenvalues of $-\widetilde \Delta$ are
    $\{ e_{j+1} \}_{j \in \mathbb{N}}$ and $\{\lambda_{j+1}\}_{j \in \mathbb{N}}$,    
respectively, i.e., the same as those of $-\Delta$, except that $e_1$ and  $\lambda_1=0$ are removed.

Using these eigenvalues and eigenfunctions, we then define the fractional operator $(-\widetilde{\Delta})^{\beta/2}(\kappa^2 - \widetilde{\Delta})^{\alpha/2}$ by its spectral decomposition (see Appendix \ref{app:fractional}), and consider 
\begin{equation}\label{eq:fractional_intrinsic}
(-\widetilde{\Delta})^{\beta/2}(\kappa^2 - \widetilde{\Delta})^{\alpha/2} (\tau u) = \widetilde{\mathcal{W}} \quad \text{on $\cD$},
\end{equation}
where $\widetilde{\mathcal{W}}$ is Gaussian white noise on $H$, which can be represented as
$\widetilde{\mathcal{W}} = \sum_{j=1}^\infty \xi_j {e}_{j+1}$,   
where $\xi_j \sim N(0,1)$, $j\in\mathbb N$, are independent random variables.

With these modifications we now show that~\eqref{eq:fractional_intrinsic} has a unique solution in $L_2(\Omega, H)$, the space of $H$-valued Bochner measurable random variables with finite second moment, which defines a proper (rather than intrinsic) Gaussian random field.

\begin{proposition}\label{prop:uniqueness}
	Suppose that $\alpha\in \mathbb{R}$,  $\beta \geq 0$ and $\alpha + \beta > d/2$. Then \eqref{eq:fractional_intrinsic} has a unique solution $u\in L_2(\Omega, H)$, which is a centered square-integrable Gaussian random field satisfying the zero-mean constraint. 
\end{proposition}

The proof is provided in Appendix~\ref{app:fractional}. The smoothness of the sample paths is determined by $\alpha+\beta$ and two realizations for different smoothness behaviors are shown in Figure~\ref{Fig:simu}.
Importantly, the solution $u$ to \eqref{eq:fractional_intrinsic} is also a solution to \eqref{SPDE_def} and thus is the unique representative of the intrinsic field that integrates to zero. It therefore also has the same variogram, which we now investigate in more detail.

\begin{figure}
\centering
\vspace{-0.6cm}
\includegraphics[width=0.38\textwidth]{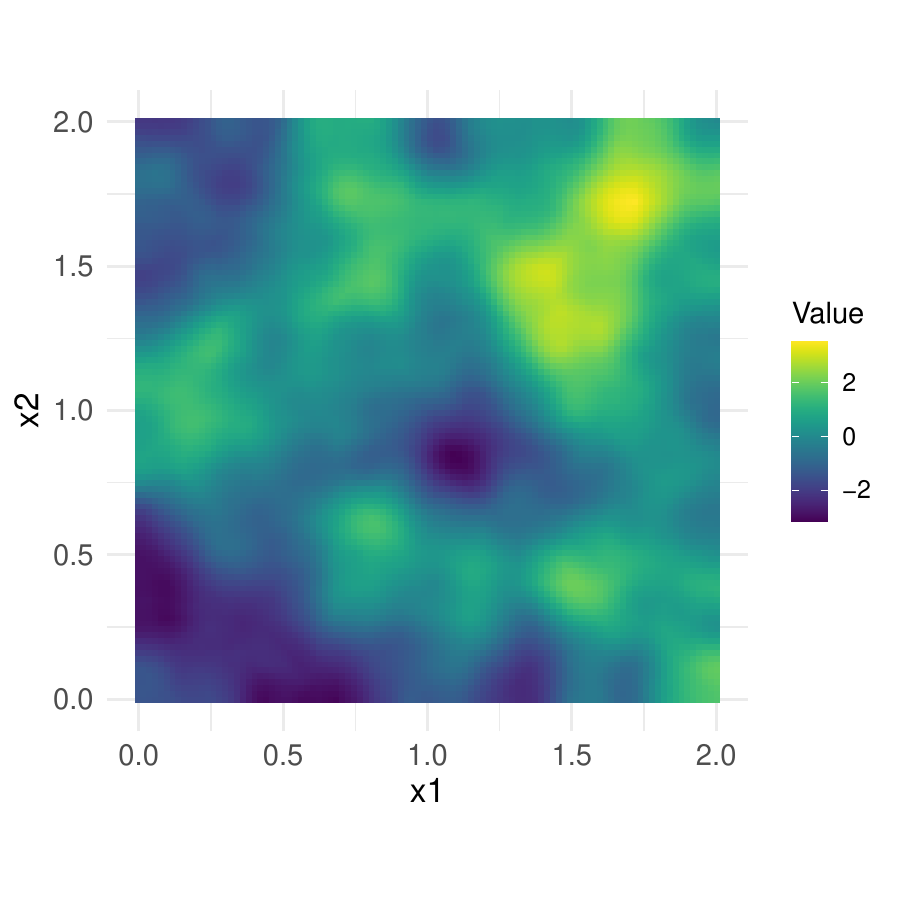}
\includegraphics[width=0.38\textwidth]{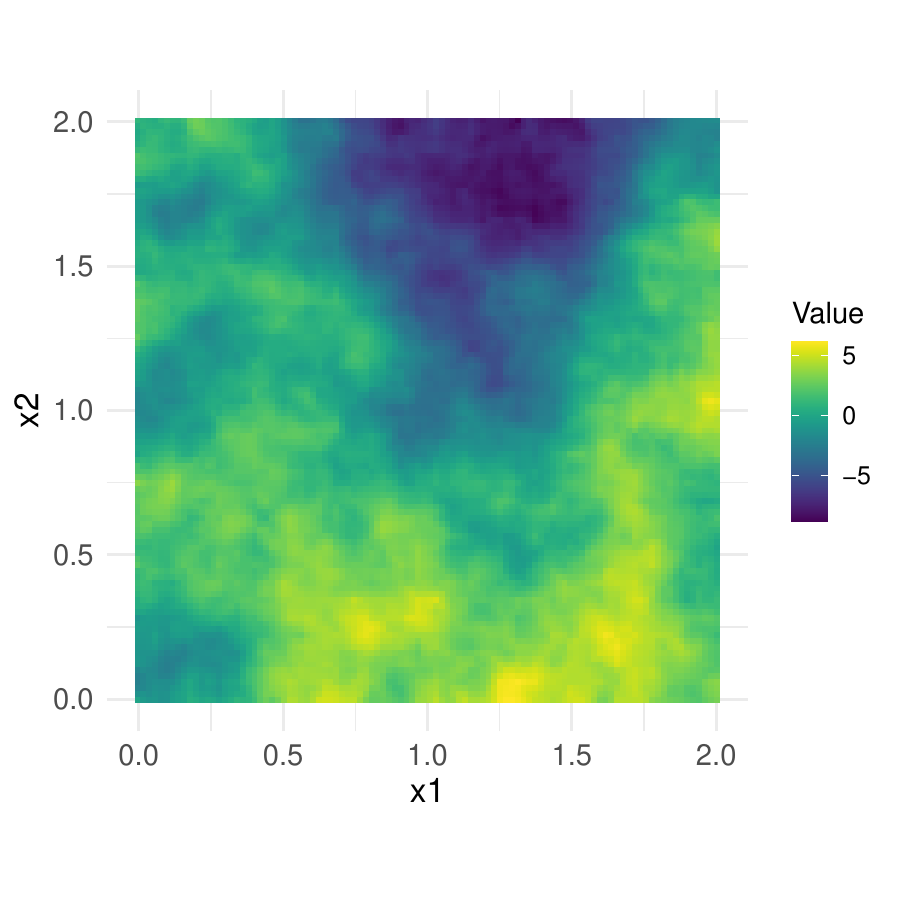}
\vspace{-0.9cm}
\caption{\label{Fig:simu} Simulations of the intrinsic \WM{} field with $\alpha=3$ and $\beta=1$ (left), and $\alpha=0.3$ and $\beta=1.5$ (right).}
\end{figure}

\subsection{Variogram}\label{sec:Properties}

With the proper Gaussian random field that solves~\eqref{eq:fractional_intrinsic}, we can now derive the variogram $\gamma$ of the corresponding intrinsic Whittle--Mat\'ern field that solves~\eqref{SPDE_def}.

If $\alpha+\beta>d/2$, the solution to \eqref{eq:fractional_intrinsic} is a Gaussian field with covariance operator $\mathcal{C} : H \rightarrow H$ given by $\mathcal{C} = \tau^{-2} (-\widetilde \Delta)^{-\beta}(\kappa^2- \widetilde \Delta)^{-\alpha}$.  
The operator $\mathcal{C}$ has eigenvalues 
$\{\hat \lambda_{j}\}_{j\in\mathbb N} =\{\tau^{-2}\lambda_{j+1}^{-\beta}(\kappa^2+\lambda_{j+1})^{-\alpha}\}_{j\in\mathbb N}$,
and eigenfunctions $\{\hat e_j\}=\{e_{j+1}\}_{j\in\mathbb N}$. 
Thus, the corresponding centered Gaussian field has the covariance function 
$\rho(\mv s,\mv t) = \sum_{j=1}^\infty \hat{\lambda}_j \hat{e}_j(\mv s)\hat{e}_j(\mv t)$, for $\mv s,\mv t \in \mathcal{D}$.
The variogram of the intrinsic Gaussian field is then 
\begin{equation}\label{eq:Vf}
\gamma(\mv s,\mv t) = \rho(\mv s,\mv s) + \rho(\mv t,\mv t) - 2\rho(\mv s,\mv t) = \sum_{j=1}^\infty \hat{\lambda}_j (\hat{e}_j(\mv s)-\hat{e}_j(\mv t))^2.
\end{equation}
When $\mathcal{D}=[0,L]^d$, \eqref{eq:Vf} combined with $\{\hat e_j\}$ and $\{\hat \lambda_{j}\}$ from Example~\ref{ex:cube} yields
$$
\gamma_L(\mv s,\mv t) = \tau^{-2}\sum_{\mv{j}\in\mathbb{N}_0^d} (\|\mv{j}\|^2\pi^2/L^2)^{-\beta}(\kappa^2 + \|\mv{j}\|^2\pi^2/L^2)^{-\alpha} (e_\mv{j}(\mv s) - e_\mv{j}(\mv t))^2.
$$ 
Due to the homogeneous Neumann boundary conditions, $\gamma_L$ is neither stationary nor isotropic, and 
$\gamma_L(\mv s,\mv t)$ depends not only on the distance $\| \mv s - \mv t\|$ but also their proximity to the boundary of the domain $\mathcal{D}=[0,L]^d$.
To obtain a stationary and isotropic variogram we take the limit as these boundaries move further away, i.e., as $L \to \infty$.

\begin{proposition}\label{lem:Vard2}
 Let $\gamma_L$ denote the variogram of \eqref{SPDE_def} on $\mathcal{D}=[-L/2,L/2]^d$ with $\beta < 1 + d/2$ and $\beta+\alpha > d/2$. Then, 
\begin{align*}
    \lim_{L \to \infty}\gamma_L(\mv s,\mv t) = \gamma(\lVert \mv s-\mv t \rVert) := 
 \frac{4}{\tau^{2}2^{d}\pi^{d/2}\Gamma(d/2)}  \int_0^\infty (1 - \Lambda_d(\lVert \mv s- \mv t \rVert r))  \frac{r^{d-1}}{r^{2\beta}(\kappa^2 + r^{2})^{\alpha}}  {\rm d}r 
\end{align*}
is a stationary and isotropic variogram on $\mathbb{R}^d$, where 
$\Lambda_d(z) = 2^{(d-2)/2}\Gamma\left(\frac{d}{2}\right)\frac{J_{(d-2)/2}(z)}{z^{(d-2)/2}}$, $z \geq 0$, 
and $J_m$ is the Bessel function of the first kind of order $m$.
\end{proposition}

The proof is given in Appendix \ref{App:Variogram}. Note that $\Lambda_d(z)$ simplifies for low dimensions. For example, $\Lambda_1(z) = \cos(z)$, 
$\Lambda_2(z) = J_0(z)$ and $\Lambda_3(z) = \sin(z)/z$. This can be used to compute explicit forms of the variogram for particular parameter combinations. A simple example is that we indeed obtain the fractional variogram when $\alpha=0$.

\begin{example}\label{ex:fractional}
    For $\alpha = 0$ and $d/2 < \beta < 1 + d/2$, we obtain the fractional variogram
$\gamma(h) = c_d^{-1}\tau^{-2}h^{2\beta - d}$, 
where $c_d$ is a dimension-dependent constant which for $d=1$ and $d=2$ is
$c_1= \Gamma(2\beta)\cos\{\pi(\beta+1)\}$ and  $c_2 = 2^{2\beta-1}\Gamma(\beta)^2\sin\{\pi(\beta+1)\}$, respectively. 
 \end{example}

Another example is the following non-fractional model which we will use later.
\begin{example}
	For $d=2$ and $\alpha=\beta=1$ we obtain
    \begin{align*}
\gamma(h) &= 
\frac{1}{\pi \kappa^2 \tau^2} \left\{ K_0(\kappa h) + \log \left( \frac{\kappa h}{2}\right) + \psi(1) \right\},
 \end{align*}
    where $K_0$ is the modified Bessel function of the second kind of order zero, and
    $\psi$ is the digamma function.
\end{example}

The variogram can be understood through its behavior for large and small $h$.
Below we write $f(x) \asymp g(x)$ if  $c_1,c_2 \in (0,\infty)$ exists such that $c_1 g(x) \leq f(x) \leq c_2g(x)$ for all $x$ sufficiently large, and $f(x) \sim g(x)$ if $\lim_x f(x)/g(x) = 1$.
\begin{proposition}\label{prop:LG} 
    If $\alpha+\beta>d/2$, then as $h \to 0$
    \begin{equation}\label{eq:LB}
    \gamma(h) \asymp h^{\min(2(\beta+\alpha)-d, 2)},
    \end{equation}
    and as $h \to \infty$ 
    \begin{equation}\label{eq:GB}
    \gamma(h) \asymp \begin{cases}
        1, \qquad &\text{ when } 0<\beta<d/2, \\
        \log(h), \qquad &\text{ when } \beta=d/2, \\
        h^{2\beta-d}, \qquad &\text{ when } d/2< \beta<d/2+1.
    \end{cases}
    \end{equation}
    In addition, when $d=2$ and $\beta=1$, 
    $\gamma(h) \sim (2 \pi \kappa^{2\alpha} \tau^2)^{-1} \log (h)$ as $h \to \infty$.
\end{proposition}

The proof is given in Appendix \ref{App:Variogram}. This result shows that the family of intrinsic Whittle--Mat\'ern variograms is very flexible in the sense that it allows different short- and long-range behaviors; this is in contrast to existing models such as the fractional variogram where both short- and long-range behaviors are controlled by the single parameter $\beta$.
The variogram $\gamma(h)$ from Proposition \ref{lem:Vard2} is illustrated in Figure \ref{Fig:SC1}. The left panel illustrates the effect of $\kappa$, which determines the rate at which the $\gamma$ transitions from its short-range ($h^2$ in this example) to its long-range behavior ($\log(h)$). When $\kappa$ is large the transition occurs quickly and the shape of $\gamma$ is determined by its long-range behavior, whereas when $\kappa$ is small the transition occurs slowly and the shape is determined by its short-range behavior. In the center panel the long-range behavior is fixed to linear, and we observe the effect of changing $\alpha$ on the short-range behavior. Finally, in the right panel the short-range behavior is held fixed (linear), and we observe the effect of changing $\beta$ on the long-range behavior.

\begin{figure}
\includegraphics[width=0.328\textwidth]{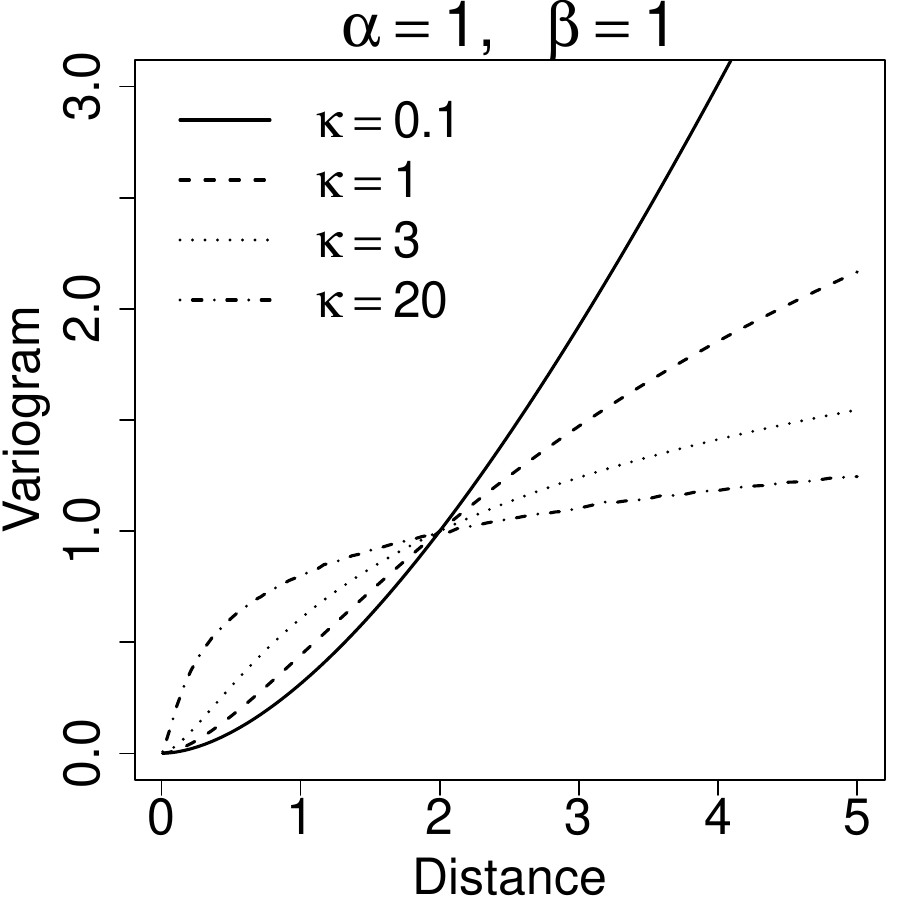}
\includegraphics[width=0.328\textwidth]{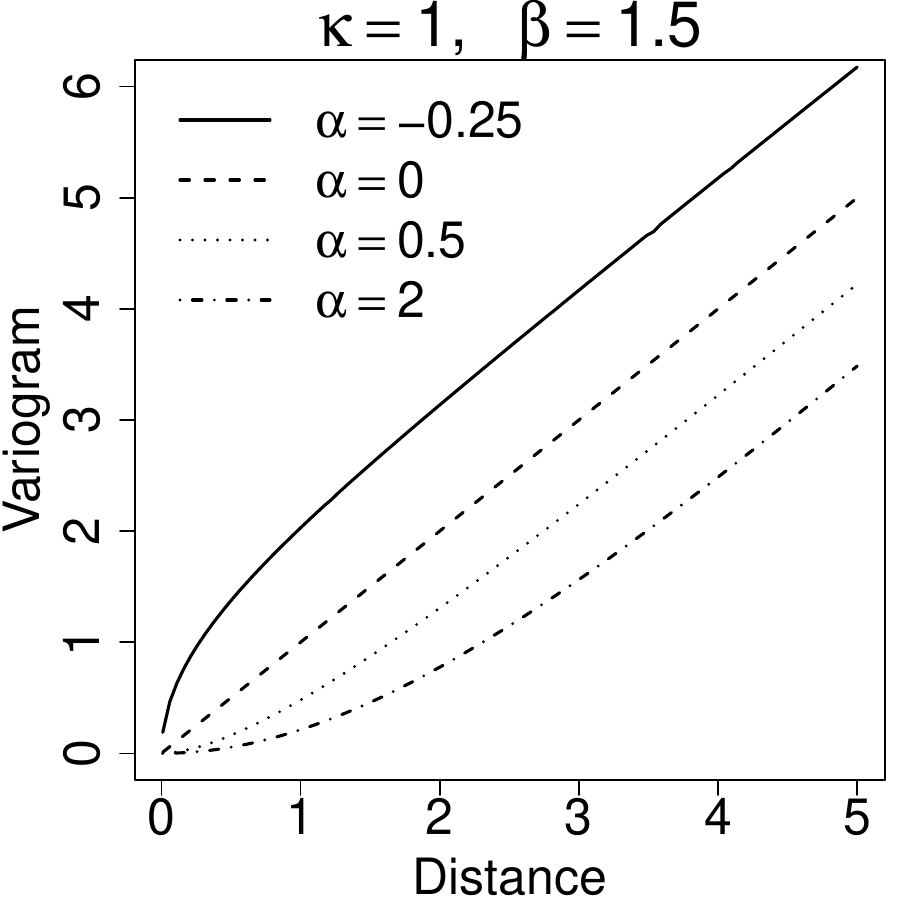}
\includegraphics[width=0.328\textwidth]{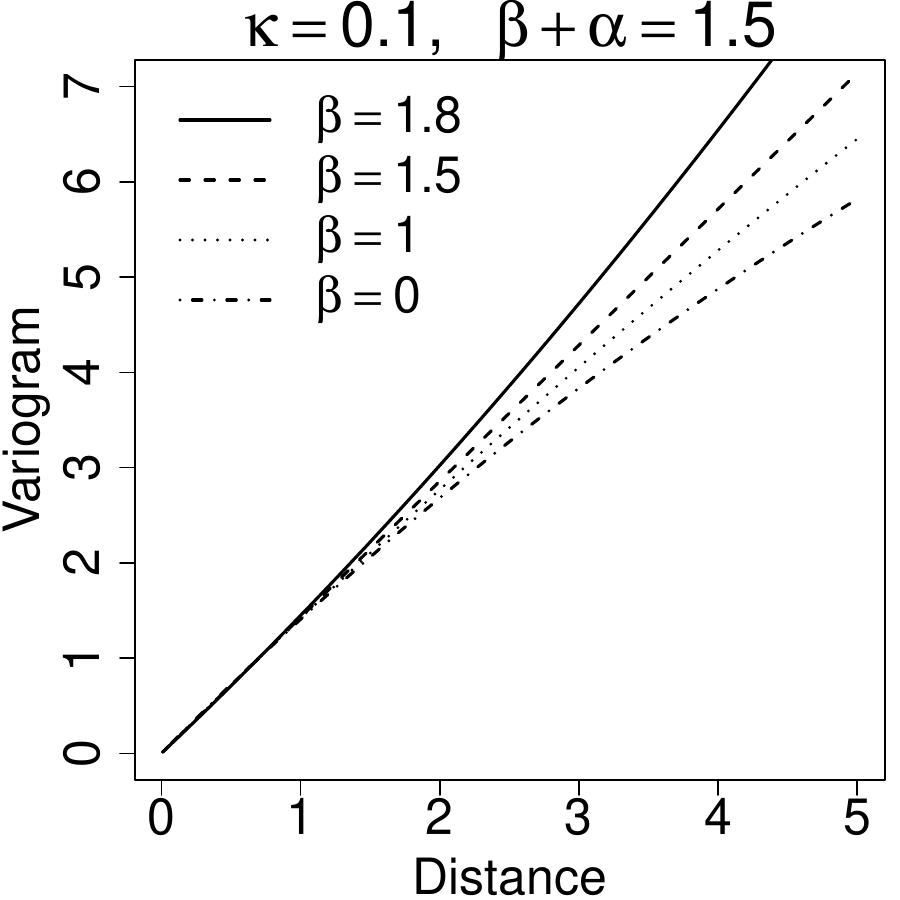}
\caption{\label{Fig:SC1} The variogram $\gamma(h)$ in Proposition \ref{lem:Vard2} when $d=2$. We choose $\tau$ such that $\gamma(2)=1$ (left), $\gamma(5)-\gamma(4)=1$ (center), and $\gamma(0.7)=1$ (right).}
\end{figure}

\section{Statistical inference}\label{sec:inference}

Statistical inference for the intrinsic Whittle--Mat\'ern field $u = (u(\mv s): \mv s\in\mathcal D)$ from Definition~\ref{def_WM} with variogram $\gamma$ relies on its finite-dimensional distribution at observation locations $s_1, \dots, s_k \in \mathcal{D}$. 
This distribution is characterized by the variogram matrix $\Gamma$ with elements  $\Gamma_{i,j} = \gamma(s_i, s_j)$, which has a corresponding precision matrix that is generally dense. 
Likelihood inference then becomes computationally prohibitive for large  $k$ as  the cost of likelihood evaluation is of order $k^3$.
Our goal is now to construct an approximation of the intrinsic Whittle--Mat\'ern fields which has a sparse precision matrix and therefore scales better to large $k$.

\subsection{Finite element and rational approximation}\label{FEM}

To construct a sparse approximation we use a triangulation $\mathcal{T}$ of the domain $\mathcal{D}$ with $N$ vertices $\mv{v}_1,\dots, \mv{v}_N$, and let $\{\varphi_j\}_{j=1,\dots, N}$ be the set of continuous and piecewise linear functions with $\varphi_j(\mv{v}_j) = 1$ at one vertex of the triangulation and 0 at all other vertices. The approximation is constructed through a FEM approximation, resulting in the following approximation of the continuous field
\begin{equation}\label{eqn:FEMp}
u(\mv{s}) \approx \widehat u(\mv{s}) := \sum_{j=1}^N W_j \varphi_j(\mv{s}), \quad \mv{s} \in \mathcal{D},
\end{equation}
which applies a linear interpolation between the values on vertices.

For integer $\alpha$ and $\beta$, an extension of the SPDE approach proposed by \cite{Lin11} to this more general SPDE model yields that the weights $\mathbf{W}=(W_1, \dots, W_N)$ are centered multivariate normal with a sparse precision matrix, 
\begin{align}\label{prec_discrete}
    {Q}_{\kappa,\alpha,\beta} = {Q}_{\kappa,0,\beta}{C}^{-1}{Q}_{\kappa,\alpha,0}.
\end{align}
Here, ${C}$ is the diagonal mass matrix with elements ${C}_{i,i} = (\varphi_i,1)_{L_2(\mathcal{D})}$, and defining ${G}$ as the stiffness matrix with elements ${G}_{i,j} = (\nabla\varphi_i,\nabla\varphi_j)_{L_2(\mathcal{D})}$, for $i,j=1,\dots, N$, we have that ${Q}_{\kappa,0,1} = {G}$ and ${Q}_{\kappa,0,\beta} = {G}{C}^{-1}{Q}_{\kappa,0,\beta-1}$  for ${\beta\in\{2,3,\ldots\}}$. Further, ${Q}_{\kappa,1,0} = \kappa^2{C} + {G}$ and 
${Q}_{\kappa,\alpha,0} = (\kappa^2{C} + {G}){C}^{-1}{Q}_{\kappa,\alpha-1,0}$ for ${\alpha\in\{2,3,\ldots\}}$. Thus, $\mv{W}$ is a GMRF if $\beta = 0$ and because $\mv{G}\mv{1} = 0$, it is a first order iGMRF if $\beta > 0$. The details of the derivation are given in Appendix \ref{app:fem}.

When $\alpha,\beta\geq 0$ such that $\alpha+\beta>d/2$, we extend the rational approximation method in \citet{xiong2022} to intrinsic fields. 
This results in an approximation where $\mathbf W = \mathbf W_1 +\dots + \mathbf W_M$ is a sum of first-order intrinsic GMRFs $\mathbf{W}_1,\ldots, \mathbf{W}_M$ for some $M\in\mathbb N$, with sparse precision matrices, which will facilitate computationally efficient inference. 
The details are provided in Appendix~\ref{app:fem}, but in short, we first perform the same FEM approximation as in the integer case. We can then represent the covariance operator of this approximation as 
$
\widetilde{L}_\meshwidth^{-\beta}L_\meshwidth^{-\alpha} = \widetilde{L}_\meshwidth^{-\lfloor \beta \rfloor}L_\meshwidth^{-\lfloor \alpha \rfloor} \widetilde{L}_\meshwidth^{-\{\beta\}} L_\meshwidth^{-\{\alpha\}},
$
where $\widetilde{L}_\meshwidth$ is the FEM approximation of $-\widetilde{\Delta}$,  $L_\meshwidth$ is the FEM approximation of $\kappa^2-\widetilde{\Delta}$ and $\{a \} = a - \lfloor a \rfloor$ denotes the fractional part of $a\in\mathbb R$.
The final approximation is then obtained by approximating the fractional operators $\widetilde{L}_\meshwidth^{-\{\beta\}}$ and $L_\meshwidth^{-\{ \alpha \}}$ as sums of $\widetilde{m}$ respectively $m$ non-fractional operators, where each term in the sums corresponds to a first-order intrinsic GMRF. 
This results in an approximation, 
$\widetilde{L}_\meshwidth^{-\beta}L_\meshwidth^{-\alpha} \approx 
\sum_{i=1}^{M} Q_{\meshwidth,i}^{-1}$,
where $M = 1 + m + \widetilde{m} + m\widetilde{m}$, and the matrix representation of each  $Q_{\meshwidth,i}$ can be computed solely based on the mass and stiffness matrices; see Appendix~\ref{app:fem}.

The accuracy of the variogram $\gamma_{\meshwidth,m,\widetilde{m}}$ corresponding to this approximation depends on the triangulation $\mathcal{T}$, and on the number of terms $m$ and $\widetilde m$ used in the approximation~\eqref{rational_app} of the functions $x^{\{ \alpha \}}$ and $x^{\{ \beta \}}$, respectively. 
To express this formally, we consider a family of triangulations $(\mathcal{T}_\meshwidth)_{\meshwidth \in (0,1)}$ indexed by the mesh width 
    $\meshwidth:=\max_{T \in \mathcal{T}_\meshwidth} \meshwidth_T$, 
where $\meshwidth_T = diam(T)$ is the diameter of $T$. To avoid the presence of long thin triangles, we assume that this family is quasi-uniform (that is, if $\rho_T$ denotes the radius of the largest ball inscribed in $T$, there exist constants $K_1,K_2>0$ such that $\rho_T \geq K_1 \meshwidth_T$ and $\meshwidth_T \geq K_2 h$ for all $T\in \mathcal{T}_\meshwidth$).
We can then prove the following explicit rate of convergence of the approximate variogram to the true variogram; see also Figure~\ref{Fig:fem} for a graphical illustration of this result, and of the quality of the FEM approximation.

\begin{theorem}\label{cov_fem_approx_rate_general}
    Let $\gamma$ denote the variogram of the solution $u$ to~\eqref{SPDE_def} with parameters $\alpha\in\mathbb R$, $\beta \geq 0$, $\kappa, \tau >0$ such that ${\alpha + \beta > d/2}$.  Then, for each $\varepsilon>0$, 
    \begin{equation*} \label{eq:variogram_error}
	\|\gamma - \gamma_{\meshwidth,m,\widetilde{m}}\|_{L_2(\cD \times \cD)}  \leq c \meshwidth^{r} + c\meshwidth^{-\frac{d}{2}}\left(1_{\beta \notin \mathbb{N}}  e^{-2\pi\sqrt{\{\beta\} \widetilde{m}}} + 1_{\alpha \notin \mathbb{N}} e^{-2\pi\sqrt{\{\alpha\} m}}\right),
\end{equation*}    
    where $r = \min\{2\alpha + 2\beta-\frac{d}{2} -\varepsilon,2\}$ and $c>0$ does not depend on $m,\tilde m$ and $\meshwidth$, but may depend on $\varepsilon,\alpha, \beta, \kappa$ and~$\mathcal{D}$. 
\end{theorem}
The proof is given in Appendix~\ref{app:fem}. The terms with $1_{\alpha \notin \mathbb{N}}$ and $1_{\beta \notin \mathbb{N}}$ in the bound are due to the rational approximations; they are not present if $\alpha$ and $\beta$ are integers. 
\begin{remark}\label{rem:mtm}
The  rate of convergence $\min\{2\alpha + 2\beta-d/2-\varepsilon,2\}$ is obtained by calibrating $m$ and $\widetilde{m}$ with $\meshwidth$, by choosing $m = g(\alpha)$ and $\widetilde{m} = g(\beta)$, where
\begin{align*}
g(\eta) &= \lceil(\min\{2\eta-d/2 - \varepsilon,2\} + d/2)^2 \frac{(\log{\meshwidth})^2}{4\pi^2\{\eta\}}\rceil. 
\end{align*}
\begin{figure}[t!]
\centering
\includegraphics[width=0.38\textwidth]{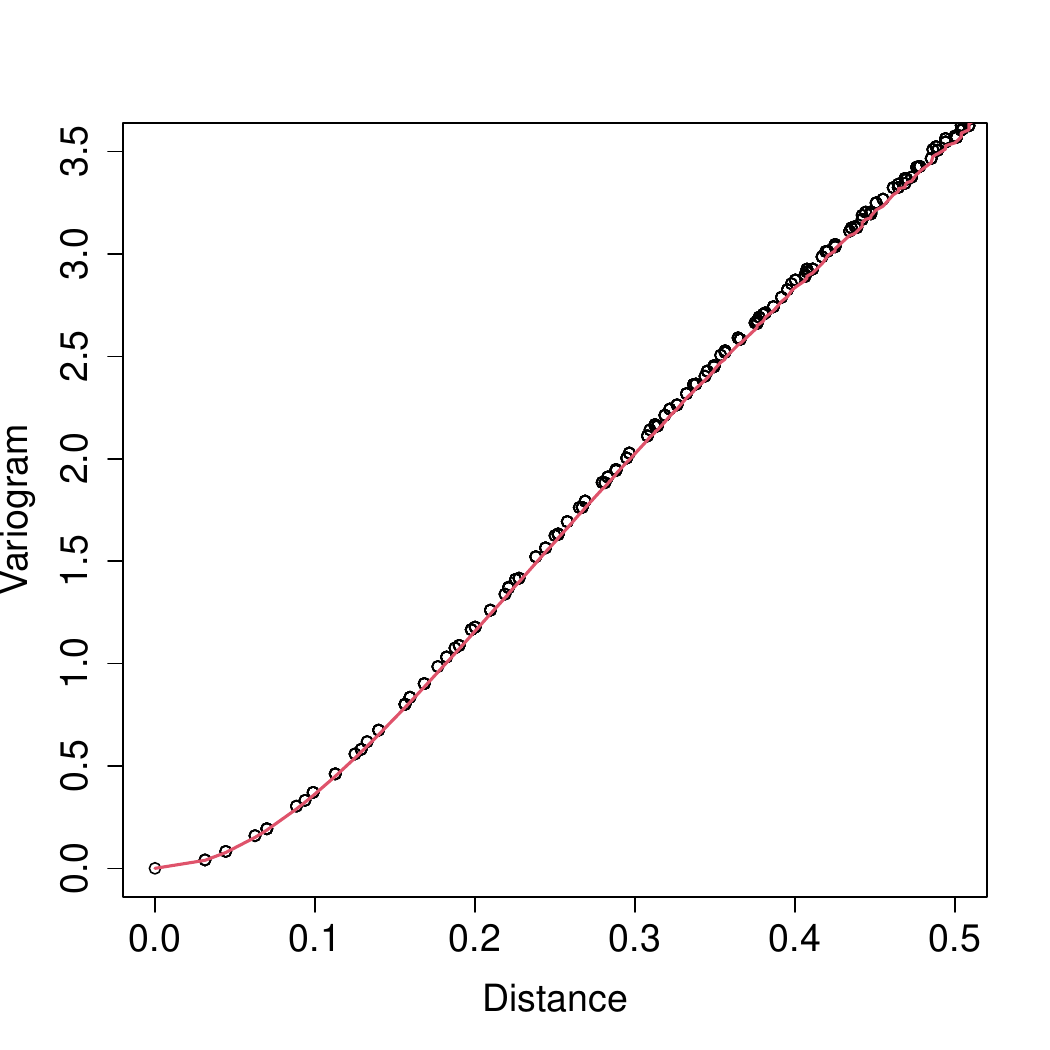}
\includegraphics[width=0.38\textwidth]{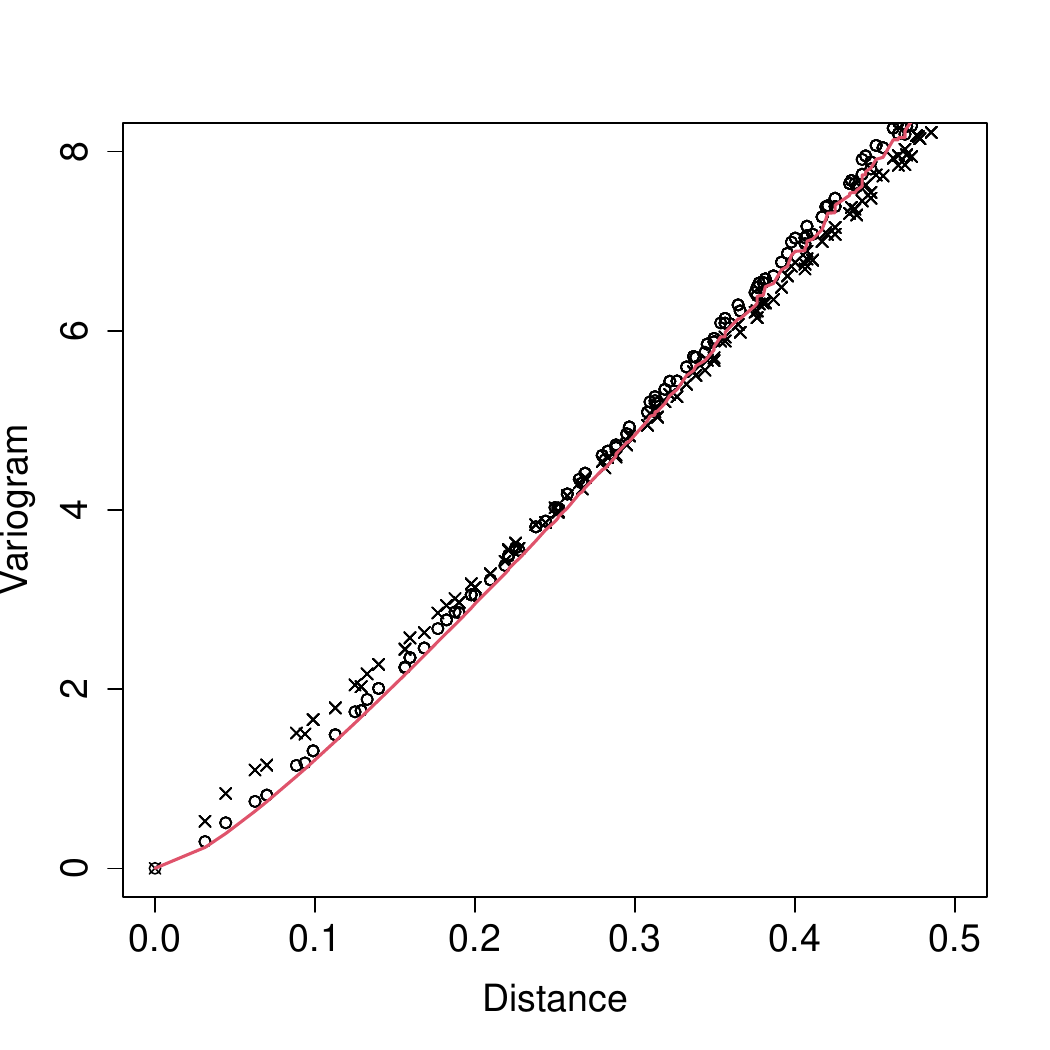}
\caption{\label{Fig:fem} FEM approximation (points) of an intrinsic \WM{} variogram with $\alpha=3$ and $\beta=1$ (left) and $\alpha=0.3$ and $\beta=1.5$ (right) when $d=2$ and $\kappa=15$. The true variogram is in solid red. The same mesh was used for the FEM approximation in both panels. The right panel displays a fractional approximation with $m=\tilde m=1$ (crosses) and $m=\tilde m=2$ (circles).}
\end{figure}
\end{remark}

\subsection{Likelihood estimation}\label{sec:LE}

To perform computationally efficient likelihood-based estimation, we need to be able to evaluate the likelihood of the model. We describe here how to evaluate the likelihood of the model based on noisy observations $\widetilde{\mathbf{w}} = (w_1,\ldots, w_k)^\top$ of the field $u$ at observation locations $\mv{s}_1, \dots, \mv{s}_k \in \mathcal{D}$ assuming that $w_i = u(\mv{s}_i) + \varepsilon_i$ with $\mv{\varepsilon} = (\varepsilon_1,\ldots, \varepsilon_k)^\top \sim N(\mv{0}, {Q}_\varepsilon^{-1})$ for some sparse precision matrix ${Q}_\varepsilon$. 
The methodology is similar to the rational SPDE approach for proper Gaussian fields \citep[][]{xiong2022} but is included for comparison with our methodology for extreme value models in Section~\ref{sec:extremes}. 

Using the FEM approximation of $u(\mv{s})$, we formulate the latent model in terms of the stochastic weights  
$\overline{\mv{W}} = (\mv{W}_1^\top, \ldots, \mv{W}_M^\top)^\top$,
which is an iGMRF of size $MN$ with precision matrix ${Q}_{\kappa,\alpha,\beta} = \text{blockdiag}({Q}_1, \ldots, {Q}_M)$. 
With the notation $\text{GMRF}(\boldsymbol{\mu}, Q)$ for a proper Gauss--Markov random field with mean $\boldsymbol{\mu}$ and precision $Q$, and $\text{iGMRF}(\boldsymbol{\mu}, Q)$ for an intrinsic analogue, we specify the hierarchical model as 
\begin{equation}\label{eqn:Hich}
\widetilde{\mathbf{W}}\mid \overline{\mathbf{W}} \sim \text{GMRF}(\overline{A} \overline{\mathbf{W}}, {Q}_\varepsilon), \qquad \overline{\mathbf{W}} \sim \text{iGMRF}(\mv{0}, {Q}_{\kappa,\alpha,\beta}),
\end{equation}
where 
$\overline{A} = (A, \ldots, A)$ is a concatenation of $M$ projection matrices with elements $A_{ij}=\varphi_j(\mv{s}_i)$, $i=1,\dots, k, j=1, \dots, N$, mapping the weights to the observation locations.  
 
A common choice for ${Q}_\varepsilon$ is $\frac{\sigma^2}{2} {I}$, where ${I}$ is the $k\times k$ identity matrix, which corresponds to adding a nugget effect to the variogram of the latent model
\begin{equation}\label{eqn:nugg}
\widetilde \Gamma_{ij} = \text{Var}(\widetilde{{W}_i}-\widetilde{{W}_j}) = \text{Var}\{(A\mv{W})_i-(A\mv{W})_j\}+\sigma^2, \quad i,j=1,\dots, k,  i\neq j.
\end{equation}
To compute the likelihood of $\widetilde{\mathbf{W}}$ first observe that for the general fractional case,
$\overline{\mathbf{W}} \mid \widetilde{\mv{W}}=\widetilde{\mv{w}} \sim \text{iGMRF}(\mv{\mu}_{\overline{\mv{w}}|\widetilde{\mv w}}, {Q}_{\overline{\mv{w}}|\widetilde{\mv w}})$,
where 
$\mv{\mu}_{\overline{\mv{w}}|\widetilde{\mv{w}}} = {Q}^{+}_{\mv{w}|\widetilde{\mv w}} A^\top {Q}_\varepsilon \widetilde{\mv{w}}$  and ${Q}_{\overline{\mv{w}}|\widetilde{\mv w}} = {Q}_{\kappa,\alpha,\beta} + A^\top {Q}_\varepsilon A$.
Here ${Q}^{+}_{\mv{w}|\widetilde{\mv w}}$ denotes the Moore--Penrose pseudoinverse of ${Q}_{\mv{w}|\widetilde{\mv w}}$ as the posterior precision is also rank deficient as it has $M-1$ zero eigenvalues. The only case where the posterior of $\overline{\mv{W}}$ is a proper GMRF is if $\alpha$ and $\beta$ both are integers and thus $M = 1$. By standard computations for latent GMRFs \citep[e.g.][]{xiong2022}, we obtain 
\begin{equation}\label{eqn:LCom}
\begin{aligned}
2 \log f_{\widetilde{\mv W}}(\widetilde{\mv{w}}) &= -(d-1) \log 2 \pi + \log | {Q}_{\kappa,\alpha,\beta}|^* + \log | {Q}_\varepsilon | - \log|{Q}_{\overline{\mv{w}}|\widetilde{\mv{w}}}|^{*}  \\
&\quad - \mv{\mu}^\top_{\mv w| \widetilde{\mv w}} {Q}_{\kappa,\alpha,\beta} \mv{\mu}_{\mv w| \widetilde{\mv w}} - (\widetilde{\mv{w}} - A \mv{\mu}_{\mv w| \widetilde{\mv w}})^\top {Q}_\varepsilon (\widetilde{\mv{w}} - A \mv{\mu}_{\mv{w}| \widetilde{\mv w}}),
\end{aligned}
\end{equation}
for the log-likelihood of the data.
The only non-standard terms in \eqref{eqn:LCom} are the generalized determinants and the calculation with the pseudoinverse in the posterior mean, which can be computed efficienty through sparse LDL factorization and back substitution and thus taking advantage of the sparsity of the matrices. 
Thus, all components in \eqref{eqn:LCom} can be computed solely based on sparse matrix operations, making likelihood evaluation feasible also for high-dimensional latent models.

\section{Kriging}\label{sec:kriging}

A common task in geostatistics is 
prediction of $u(\mv s_0)$ at an unobserved location $\mv s_0 \in \mathcal D$ based
the observed vector $\mv W= (u(\mv s_1),\dots, u(\mv s_k))$.
The most popular method to perform such predictions is
kriging due to its optimality properties \citep{cressie2015statistics}.
While proper, non-intrinsic fields yield kriging estimates 
that revert to the overall mean, in this section
we show that intrinsic models
such as our intrinsic \WM{} model allow for a much broader
behavior, particularly for extrapolation. 
We quantify this extrapolation behavior theoretically and
show that our new model outperforms proper fields
and the fractional model on a kidney disease data set.

\subsection{Optimal prediction and extrapolation behavior}
Let $\mv W_0 = (u(\mv s_0), u(\mv s_1), \dots, u(\mv s_k))$ be the joint vector including the unobserved location. Ordinary kriging computes the conditional mean $\hat{u}(s_0):={\rm E}(u(\mv s_0) \mid \mv W)$ and variance  $\hat\sigma^2(\mv s_0):={\rm Var}(u(\mv s_0) \mid \mv W)$, under the assumption that $\mv W_0 \sim \text{GMRF}(\mu \mv 1, Q)$, with $Q$ known and $\mu$ estimated from the observed values in $\mv W$. 
While $\hat u(\mv s_0)$ is often written using the variogram matrix of $\mv W_0$ (see \eqref{lam_krig}), we express it here using the precision matrix to highlight the difference between kriging with intrinsic and non-intrinsic fields. If $Q$ is a first-order intrinsic precision (i.e., $Q \equiv \Theta$ with $\Theta\mv1 =\mv 0$) then
\begin{equation}\label{krig2}
     \hat{u}(\mv s_0) = -\sum_{i=1}^k \frac{\Theta_{0i}}{\Theta_{00}} u(\mv s_i).
\end{equation}
Since $\mu \mv 1$ is in the null space of $\Theta$, $\hat u(\mv s_0)$ does not depend on $\mu$; hence \eqref{krig2} is the standard formula for the conditional expectation of a multivariate GMRF with mean~$\mv 0$. Since $\Theta \mv 1=\mv 0$ the weights $-\Theta_{0i}/\Theta_{00}$ sum to 1, so if $\Theta$ is sparse then $\hat u(\mv s_0)$ can be interpreted as a (weighted) local average. On the other hand, if $Q$ is a positive definite precision matrix, then \eqref{krig2} can again be applied, but with $\Theta=Q-Q\mv 1 \mv 1^\top Q/(\mv 1^\top Q \mv 1)$, which, as shown in Proposition \ref{thm:GR}, is the first order intrinsic precision matrix with the same variogram matrix as $\mv W_0$.
Note that, if $Q$ is sparse, then $\Theta$ is typically dense, so $\hat u(\mv s_0)$ can be interpreted as a (weighted) global average. Intuitively, this is because all observations in $\mv W$ contain information about $\mu$ and thus contribute to the computation of $\hat u(\mv s_0)$. We observe a similar phenomenon in Section \ref{sec:extkrig} on kriging for extremes.

\begin{figure}[t!]
\centering
\includegraphics[width=0.38\textwidth]{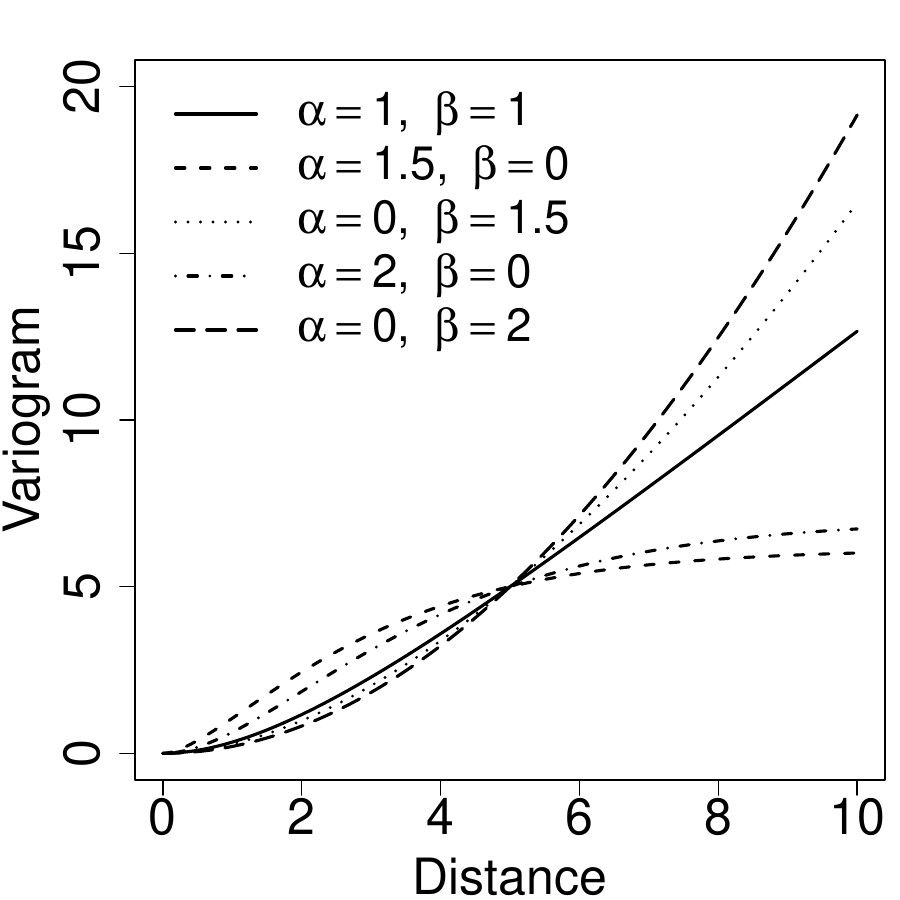}
\includegraphics[width=0.38\textwidth]{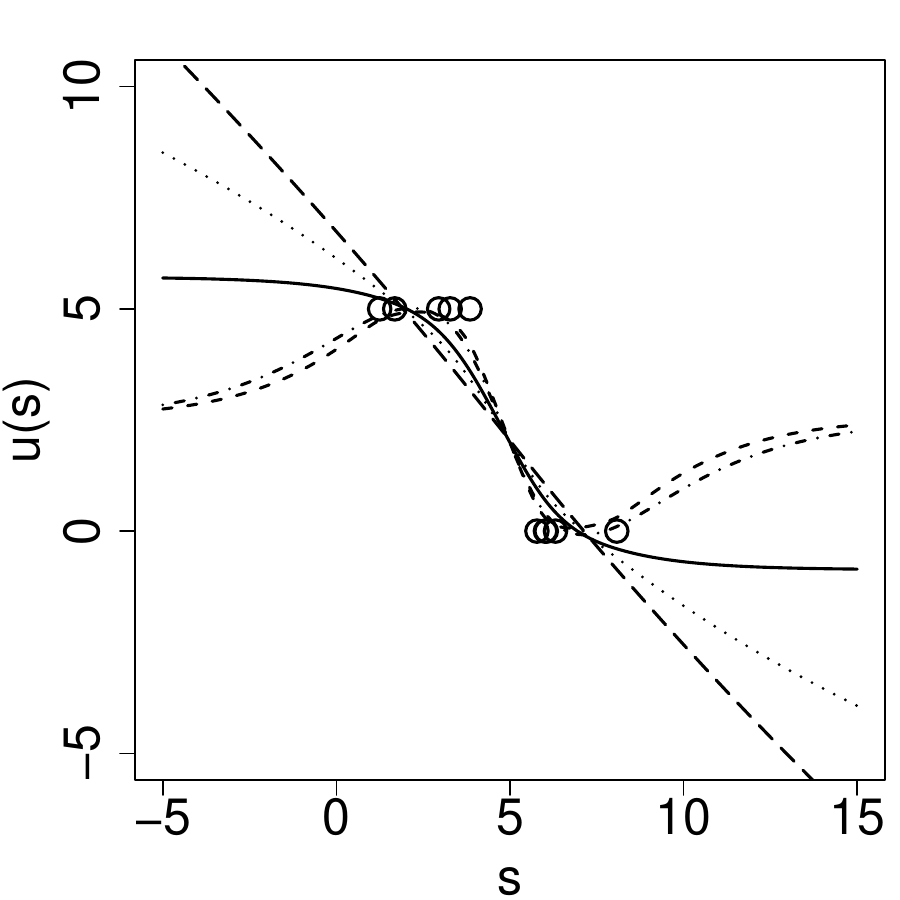}
\caption{\label{kriging_ex} Left: Variogram with $\kappa=0.5$ and $\tau$ such that $\gamma(3)=5$. Right: Kriging estimate $\hat{u}(s)$ with observations $Y_i=u(s)+\varepsilon_i$, where $\varepsilon_i \sim N(0,1)$.}
\end{figure}

The difference in $\hat u(\mv s_0)$ for different choices of $Q$ (i.e., different variograms) is often most pronounced when $\mv s_0$ lies outside the range of $\mv s_1, \dots \mv s_k$, i.e., when we extrapolate.
Figure~\ref{kriging_ex} displays variograms (left) and ordinary kriging estimates (right) for a toy example with $k=10$ observations. There are two proper Whittle--Mat\'ern fields ($\beta=0$) which have bounded variograms, an intrinsic Whittle--Mat\'ern field ($\beta=1$) which, by Proposition \ref{prop:LG}, has a  linearly growing variogram at large distances, and two fractional variograms that grow faster than linearly ($\beta=1.5,2$).
The next result shows how the growth of the variogram determines the extrapolation behavior of the kriging estimate. To state this result, let $\hat{\bar u} = \mv {1}^\top \Gamma^{-1} \mv{W}/(\mv 1^\top \Gamma^{-1} \mv 1)$, where $\Gamma=(\gamma( \lVert \mv s_i- \mv s_j \rVert))_{i,j=1,\dots,k}$. Lemma \ref{lem:ovmean} shows that $\hat{\bar u}$ is the conditional expectation of the average value of the field over an asymptotically large spherical region centered at the origin. In Figure \ref{kriging_ex}, $\hat{\bar u}\approx 2.5$ for each of the five variograms.

\begin{proposition}\label{prop:extrap}
Suppose that $\mv s_1, \dots, \mv s_k \in \mathbb{R}^d$, $\mv w = (u(\mv s_1), \dots, u(\mv{s}_k))$, $\mv v \in \mathbb{R}^d$ with $\lVert \mv v \rVert =1$, and $\ell(\cdot)$ is a slowly varying function. If $\gamma(h)=\ell(h)h^{b}$ with $0<b<2$, then
\[
\hat u(L \mv v) - \hat{\bar u}  \sim c(\mv s_1, \dots,\mv s_k, \mv w, \mv v) \ell{(L)}L^{b-1}, \quad \text{as } L \to \infty, 
\]
provided $c(\mv s_1, \dots,\mv s_k, \mv w, \mv v) \neq 0$; a formula for $c(\cdot)$ is in Appendix \ref{App:kriging}. Furthermore, if $\ell(h) \equiv \ell$ is constant or $b \leq 1$, then $\hat\sigma^2(L \mv v)={\rm Var}(\hat u(L \mv v)) \sim \ell(h) L^b$ as $L \to \infty$.
\end{proposition}

When $u$ is a proper field, $\hat u(L \mv v)$ converges to the overall mean $\hat \mu = \hat{\bar u}$ as $L \to \infty$. By contrast, Proposition \ref{prop:extrap} shows that intrinsic fields exhibit a range of extrapolation behavior. 
Indeed,  when $\alpha=\beta=1$, Proposition \ref{prop:LG} suggests that the variogram grows linearly ($b=2 (\beta -d/2)=1$); hence by Proposition \ref{prop:extrap}, we expect $\hat u(L \mv v)$ to remain a constant distance from $\hat{\bar{u}} \approx 2.5$. Similarly, when $\beta=1.5,2$ we have $b>1$, so $\hat u(L \mv v)$ diverges from $\hat{\bar u}$, allowing the kriging estimate to capture an overall trend.

\subsection{Application: clinical kidney function data}\label{sec:application.krig}
Our theoretical results suggest that intrinsic models often provide meaningful predictions when extrapolating (as opposed to interpolating) beyond the observation domain. We illustrate this by applying our new SPDE model \eqref{SPDE} to a large clinical dataset of repeated kidney function measurements for patients from the city of Salford in the United Kingdom (UK), who are at high risk of renal failure. This data set was already extensively studied by \citet{diggle2015realtime} and \citet{asar2020linear}, and we here reanalyze it with a view toward exploring the prediction performance of intrinsic and non-intrinsic models in the context of a linear mixed effects for longitudinal data.

For computational reasons, we consider a subset of $300$ randomly selected patients from this database with at least $30$ observations in time, resulting in $15887$ observations in total. For these patients, the estimated glomerular filtration rate (eGFR), a proxy for renal function, is measured repeatedly at irregular follow-up times ranging from $0$ to $6.64$ years. For illustration, the data for two selected patients are plotted in Figure~\ref{fig:CVpreds}.

\begin{figure}[t!]
\centering
\includegraphics[width=0.9\textwidth]{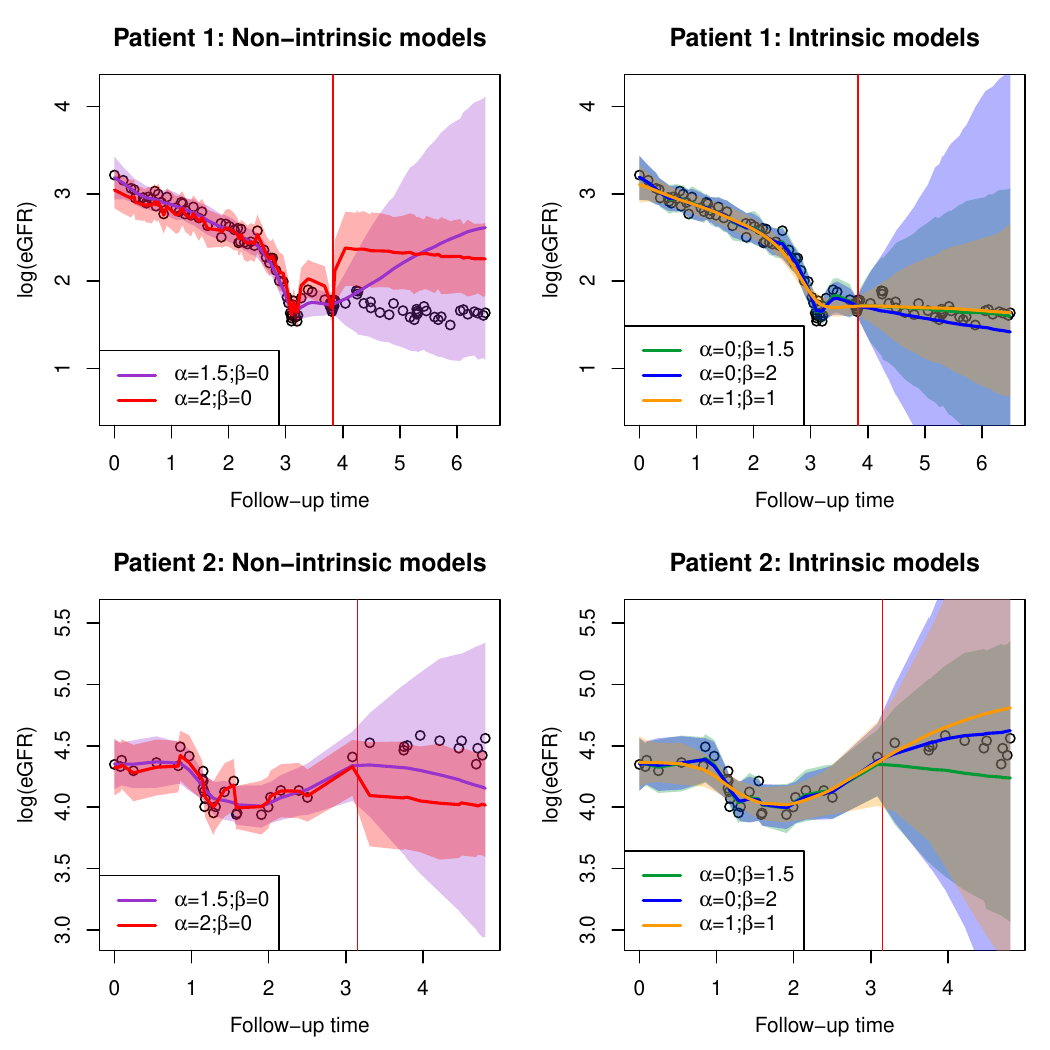}
\caption{\label{fig:CVpreds} Data (black dots) for two selected patients (top and bottom), and fitted/predicted values obtained by non-intrinsic (left) and intrinsic (right) models when holding out the patient's last $30\%$ of the data during training. Solid lines are posterior means and shaded areas are $90\%$ credible intervals. The vertical red lines separates data used for model fitting (left-hand side) and for validating predictions (right-hand side).}
\end{figure}

Following \citet{diggle2015realtime} and \citet{asar2020linear}, we assume that the logarithm of eGFR for patient $i$ and follow-up time $t_{ij}$, denoted by $Y_{ij}$, follows the model 
$Y_{ij}=\mv{x}_{ij}^\top\mv{\beta} + W_i(t_{ij})+Z_{ij}$, 
where $\mv{x}_{ij}$ is a vector of fixed covariates including sex, age at the start of the study, and a piecewise linear age correction, $\mv{\beta}$ is the corresponding vector of regression coefficients with a Gaussian prior, 
$W_i$ is a continuous-time Gaussian process capturing temporal trends and dependence for the $i$-th patient, and $Z_{ij}$ is Gaussian noise capturing measurement error. When $W_i$ is non-intrinsic, we specify its mean in terms of a patient-specific intercept $U_i$ modeled as a Gaussian random effect; this is not necessary when $W_i$ is first-order intrinsic as the model is invariant to the addition of constants, which removes one hyperparameter (the variance of $U_i$) to estimate. Conditionally on all model hyperparameters, which include the random effects' variances and dependence parameters (shared among all patients), the data $Y_{ij}$ are assumed to be independent across patients. For further details (e.g., prior and hyperparameter selection), see \citet{diggle2015realtime}. We focus on the $W_i$ process and the effect of its intrinsic/non-intrinsic specification on predictions extrapolating forward in time.

According to \citet{asar2020linear}, clinical care guidelines in the UK require patients whose kidney function decays faster than a relative rate of $5\%$ per year to be referred to specialist secondary care. To detect or predict such situations, one needs to estimate the slope of the estimated $W_i$ effect for each patient and (future) time point. Therefore, the random effect $W_i$ must be specified so that it has differentiable sample paths.  
To satisfy this clinically-motivated constraint, both \citet{diggle2015realtime} and \citet{asar2020linear} used an integrated random walk for $W_i$, but this model is limited in the smoothness and short- and long-range behaviors that it can capture. We instead focus on our new SPDE model~\eqref{SPDE}, and consider a range of submodels with differentiable sample paths. Precisely, we fit models with 
$(\alpha,\beta)=(0,1.5)$, $(1.5,0)$, $(0,2)$, $(2,0)$, and $(1,1)$ held fixed. 
The first two models have $\nu=\alpha+\beta-1/2=1$, thus reaching the lower boundary to get differentiable sample paths, while the last three models have slightly smoother sample paths with $\nu=1.5$. Furthermore, while the models with $\beta>0$ are intrinsic, the ones with $\beta=0$ are not. We do not add the integrated random walk model from \citet{asar2020linear} to our comparison because it is essentially the same (up to boundary conditions) as our model with $(\alpha,\beta) =(0,2)$. Estimating $\alpha$ and $\beta$ instead of fixing them as above did not lead to significant improvements in model fits and predictions under the differentiability constraint. To estimate the $300$ processes $W_i$ simultaneously we approximate the SPDEs with a common mesh discretizing the time interval $[0,6.64]$ into $500$ equispaced nodes, ensuring sufficient accuracy for each patient.

We perform an extensive cross-validation study to assess the ability of each model to predict observations forward in time. For each patient, we split the observations into 10 time intervals containing approximately the same number of observations, and for each $k=1,\ldots,9$, we fit the models using all data (from all patients) but holding out data from the intervals ${\{k+1,\ldots,10\}}$ for that particular patient, and we then predict data in the next ${(k+1)}$th interval for the same patient.  Figure~\ref{fig:CVpreds} shows predictions from all models for two patients, using the first $k=7$ time intervals for model fitting. The results shown for these two selected patients are symptomatic of the general behavior: while predictions using non-intrinsic models tend to have lower uncertainty, they are also often strongly biased. On the other hand, predictions using intrinsic models are often more accurate, but the prediction uncertainty increases quickly as $\beta$ increases. Therefore, for this particular dataset, intrinsic models with $\beta\approx1$, which have a meaningful extrapolation behavior (recall Figure~\ref{kriging_ex}) and display a reasonable bias-variance trade-off, often perform better. To validate this, we assess the models' extrapolation capability using three well-known proper scoring rules, namely the log-score, the continuous ranked probability score, and a scaled version of the latter (SCRPS), averaged over all predicted intervals for all patients (i.e., $9\times300$ folds in total); see \citet{gneiting2007strictly} and \citet{bolin2023local} for details on these proper scoring rules. The results are reported in Table~\ref{tab:results.application.kidney}. Interestingly, irrespective of the prediction metric, the new intrinsic model with $\beta=1$ indeed performs better than the other ones, in particular the non-intrinsic ones with $\beta=0$ and the intrinsic model with $\beta=2$, whose extrapolations either revert back to the global mean or diverge quickly, respectively.  
Overall, our new SPDE model with $\alpha=\beta=1$ is thus the best model here, and while these results depend on the precise behavior of the data at hand, our results provide support for using it in other domains of application.

\begin{table}[t!]
\caption{Out-of-sample metrics to evaluate the extrapolation ability of different intrinsic or non-intrinsic submodels of the SPDE~\eqref{SPDE}. Chosen metrics (log-score, CRPS, and SCRPS) are proper scoring rules. Scores are negatively oriented and the best models are highlighted in bold for each metric.\label{tab:results.application.kidney}}
\centering
\begin{tabular}{r|r|cccc}
Model & Type & Log-Score & CRPS & SCRPS \\ \hline
$\alpha=0,\beta=1.5$ & Intrinsic & -1.136 & 0.162 & 0.461 \\
$\alpha=1.5,\beta=0$ & Non-intrinsic & -1.111 & 0.167 & 0.468 \\
$\alpha=0,\beta=2$ & Intrinsic & -0.952 & 0.179 & 0.505 \\
$\alpha=2,\beta=0$ & Non-intrinsic & -1.152 & 0.172 &  0.471\\
$\alpha=1,\beta=1$ & Intrinsic & \textbf{-1.225} & \textbf{0.162} &  \textbf{0.444} 
\end{tabular}
\end{table}

\section{Sparse spatial extremes}\label{sec:extremes}

\subsection{Background}\label{sec:BR1}

There are two main approaches in statistics of extremes for modeling large values of a spatial process $X = \{X(\mv s): \mv s\in \mathcal D\}$: that based on block-maxima and that based on peaks-over-threshold \citep[e.g.,][]{davison2015statistics,huser2022advances}. We recall background on these two approaches, before developing sparse models for spatial extremes in Section~\ref{sec:WM.BR.processes}. To concentrate on models for the extremal dependence structure, we assume that the univariate marginal distributions of $X = \{X(\mv s) :\mv s\in\mathcal D\}$ has been standardized to a unit exponential distribution. 

\vspace{5pt}

\noindent\textbf{Block-maxima.} The block-maxima approach studies the limiting distribution of pointwise maxima of $n$ independent copies $X_1,\dots ,X_n$ of $X$.
In practice, the processes $X_i$, $i=1,\ldots,n$, may represent daily observations of a spatial variable, e.g., temperature, in which case taking $n=365$ may be interpreted as studying the process of annual temperature maxima at each point in space. 
When the $X_i$'s have unit exponential margins, the stochastic process $\eta = \{\eta(\mv s): \mv s\in \mathcal D\}$ is called max-stable if the convergence
\begin{align}\label{def_eta} \eta(\mv s) = \lim_{n\to\infty} \max_{i=1}^n \{X_i(\mv s) - \log(n)\} , \quad \mv s \in\mathcal D, 
\end{align}
holds in the sense of finite-dimensional distributions. The process $X$ is then in the max-domain of attraction of $\eta$, which has standard Gumbel margins and admits the stochastic representation 
\begin{align}\label{rep_eta}
    \eta(\mv s) = \max_{i=1}^\infty \{U_i + V_i(\mv s)\}, \quad \mv s\in \mathcal D,
\end{align}
where $\{U_i : i\in \mathbb N\}$ are the points of an independent Poisson point process on $\mathbb{R}$ with intensity $\exp(-y) {\rm d}y$, and the spectral functions $V_i$ are independent copies of a stochastic process $V = \{V(\mv s) : \mv s\in\mathcal D\}$ \citep{deh1984}. The max-stable field $\eta$ is characterized by the joint distribution of all increments $\{V(\mv s) -V(\mv t)\}$ of $V$. 
The process $V$ is therefore not unique (but must satisfy $\mathbb E[\exp\{V(\mv s)\}] = 1$ for any $\mv s\in\mathcal D$) and is naturally modeled by an intrinsic field. 

\vspace{5pt}

\noindent\textbf{Peaks-over-threshold.} The peaks-over-threshold approach studies individual extreme events that are defined with respect to a risk functional $r: \mathcal C(\mathcal D) \to \mathbb R$ as defined in \cite{dombry2015functional}, which satisfies $r(f + u) = u + r(f)$
for any $u\in\mathbb R$ and $f$ in the space $C(\mathcal D)$ of continuous functions. 
An $r$-Pareto process is defined as the limit in the sense of finite-dimensional distributions of the conditional exceedances
\begin{align}\label{def_pareto}
Y^{(r)}(\mv s) = \lim_{u \to \infty} \left( X(\mv s) - u  \mid r(X) > u \right), \quad \mv s\in \mathcal D.
\end{align}
Popular risk functionals include
    $r(f) = f(\mv{s}_0)$, and $r(f)= \log \int_{\mathcal D} \exp\{f(\mv s)\} \mathrm d \mv s$, 
which correspond to extremes at one location $\mv{s}_0 \in \mathcal D$ and extremes of the integrated process when the marginals are on the unit Pareto scale, respectively; if the domain $\mathcal D$ is a finite set, the integral is to be understood as a sum. 
The distribution of $Y^{(r)}$ is determined by the spectral function $V$.  
 Existences of the limits in~\eqref{def_eta} and~\eqref{def_pareto} are equivalent \citep{thibaud2015efficient} and the limiting processes are then called associated.

\vspace{5pt}

\noindent\textbf{Brown--Resnick processes.}
The most popular spatial extreme model is the class of Brown--Resnick processes since they are the only possible limits under suitable rescaling of pointwise maxima of Gaussian processes \citep{Kab09}.
\begin{definition}\label{def:BR}   
    The max-stable and $r$-Pareto Brown--Resnick process with variogram $\gamma$ is characterized by the spectral function 
    \begin{align}\label{spec_Gauss}
        V(\mv s) = W(\mv s) - \sigma^2(\mv s)/2, \quad \mv s\in \mathcal D,
    \end{align} 
    where $W = \{W(\mv s): \mv s\in\mathcal D\}$ is a zero-mean Gaussian processes with variance function $\sigma^2(\mv s)$ and variogram function $\gamma(\mv s,\mv t)$. 
\end{definition}
The distribution of Brown--Resnick processes only depends on the variogram $\gamma$ and not the covariance function of $W$ \citep{Kab09}. 
If the Gausssian process $W$ in~\eqref{spec_Gauss} is intrinsically stationary and isotropic, that is, $\gamma(\mv s,\mv t) = \gamma(\|\mv s-\mv t\|)$ only depends on the distance between $\mv s,\mv t\in \mathcal D$, then the corresponding max-stable Brown--Resnick process $\eta$ is stationary and isotropic. 
A popular summary for the strength of extremal dependence
is the extremal correlation coefficient \citep{col1999}, which for a stationary and isotropic Brown--Resnick processes takes the form 
\begin{align}\label{chi_coeff}
    \chi(h) = \lim_{x\to\infty}{\rm Pr}\{\eta(\mv t)>x\mid \eta(\mv s)>x\}=2 - 2\Phi(\sqrt{\gamma(h)}/2) \in [0,1],
\end{align} 
where $h = \|\mv s- \mv t\|$ and $\Phi$ is the standard normal distribution function. 

The finite-dimensional distributions of the \BR{} process at locations
${\mv s_1,\dots, \mv s_k \in \mathcal D}$  is called a \citet{hueslerReiss1989} distribution, characterized by the variogram matrix $\Gamma_{ij} = \gamma(\mv s_i,\mv s_j)$, $i,j=1,\dots ,k.$
Alternatively, the model can be parameterized by the $k \times k$, rank $k-1$, symmetric positive positive semi-definite {\HR{} precision matrix} $\Theta$ satisfying $\Theta \mathbf{1} = \mathbf{0}$ \citep{hentschel2024statistical}.
The density of the \HR{} model is characterized via its exponent measure density
\begin{equation}\label{eqn:ExpDN}
\lambda(\mathbf{y}) = \exp(-y_m) k^{-1/2} f(\mathbf{y} + \Gamma_{\cdot, m} / 2 ; \Theta),
\end{equation}
where $f(\cdot; \Theta)$ is the density of an intrinsic GMRF with precision $\Theta$ as in Definition~\ref{def:FOI} and $m\in\{1,\dots, k\}$ is an arbitrary index. Here we use Lemma \ref{lem:Eqv} to write $\lambda(\mv y)$ in terms of an intrinsic $k$-dimensional Gaussian density rather than a proper $(k-1)$-dimensional Gaussian density; see \citet[][Appendix C.1]{rottger2023total} for details on the exponent measure. 
For a risk functional $r:\mathbb R^k \to \mathbb R$, defined
on functions on the finite set $\{\mv s_1,\dots, \mv s_k\}$, the density of the $r$-Pareto \HR{} distribution is 
\begin{equation}\label{eqn:frd}
f^{(r)}(\mv{y} ; \Theta) = \frac{1}{C_{r,\Theta}}\lambda(\mathbf{y}; \Theta), \quad \mathbf y \in \mathcal L_r,
\end{equation}
where $\mathcal L_r = \{\mathbf{x}: r(\mathbf{x})>0\}$, and $C_{r,\Theta} = \int_{\mathcal L_r} \lambda(\mathbf{x}) \rm{d}\mathbf{x}$ is the normalizing constant.

Sparsity in $\Theta$ corresponds to a sparse extremal conditional independence structure according to the notion introduced by \cite{Eng20} (and denoted with the symbol $\indep_e$) 
for the $r$-Pareto distribution $\mathbf{Y}^{(r)}$, that is    
\begin{equation}\label{eqn:HRCI}
Y_i \indep_e Y_j \mid \mathbf{Y}_{\setminus \{i,j\} } \quad \Leftrightarrow \quad \Theta_{ij}=0.
\end{equation}
Here we drop the superscript since this property is independent of the choice of risk functional $r$. Similarly to Gaussian graphical models, an extremal graphical model on a graph $\mathcal{G}=(\mathcal{V},\mathcal{E})$ is defined through the pairwise Markov property, i.e., it satisfies extremal conditional independence~\eqref{eqn:HRCI} for any $\{i,j\}\notin \mathcal E$. 
The exponent measure density $\lambda$ and the density of the $r$-Pareto distribution $\mathbf{Y}^{(r)}$ of an extremal graphical model factorize on $\mathcal G$ into lower-dimensional terms, which enables efficient inference \citep{engelke2021sparse,engelke2024infinite}.

\subsection{\WM{} \BR{} processes}
\label{sec:WM.BR.processes}

Current models for spatial extremes are mostly based on Brown--Resnick processes with suitable parametric families for the variogram $\gamma$.
In the sequel we focus on stationary and isotropic models with $\gamma(\mv s,\mv t) = \gamma(\|\mv s-\mv t\|)$, $\mv s,\mv t\in \mathcal D$, and discuss possible extensions in Section~\ref{sec:discussion}.
To address limitations of existing parametric models, we first state properties that are desirable for a model for spatial extreme events.  
\begin{itemize}
    \item[(i)] The model should allow for the full range of the extremal correlation function $\chi(h)$ in~\eqref{chi_coeff},
and in particular for large distances the margins should become independent, i.e., $\chi(h)\to 0$, or equivalently, $\gamma(h) \to\infty$  as $h\to\infty$. 
    \item[(ii)] The model should be flexible with different parameter values resulting in different local and global dependence behaviors, and natural extensions to address practical challenges such as non-stationarity or non-Euclidean domains (e.g., spheres).
    \item[(iii)] Statistical inference for the model class should be efficient even for large datasets measured at many spatial locations. 
\end{itemize}

The most popular model in spatial extremes applications is the class of fractional variograms in Example~\ref{ex:fractional}; see Figure~\ref{Fig:extremal_correlation} (left panel) for two examples of the corresponding extremal correlation functions. While they do satisfy (i), they fail to satisfy (ii) because their local and global behavior is governed by the same parameter $\beta$. They also do not satisfy (iii) for likelihood estimation as they have dense \HR{} precision matrices and current methods cannot exploit sparse computations. 

\begin{figure}[t!]
    \centering
	\includegraphics[width=0.45\textwidth]{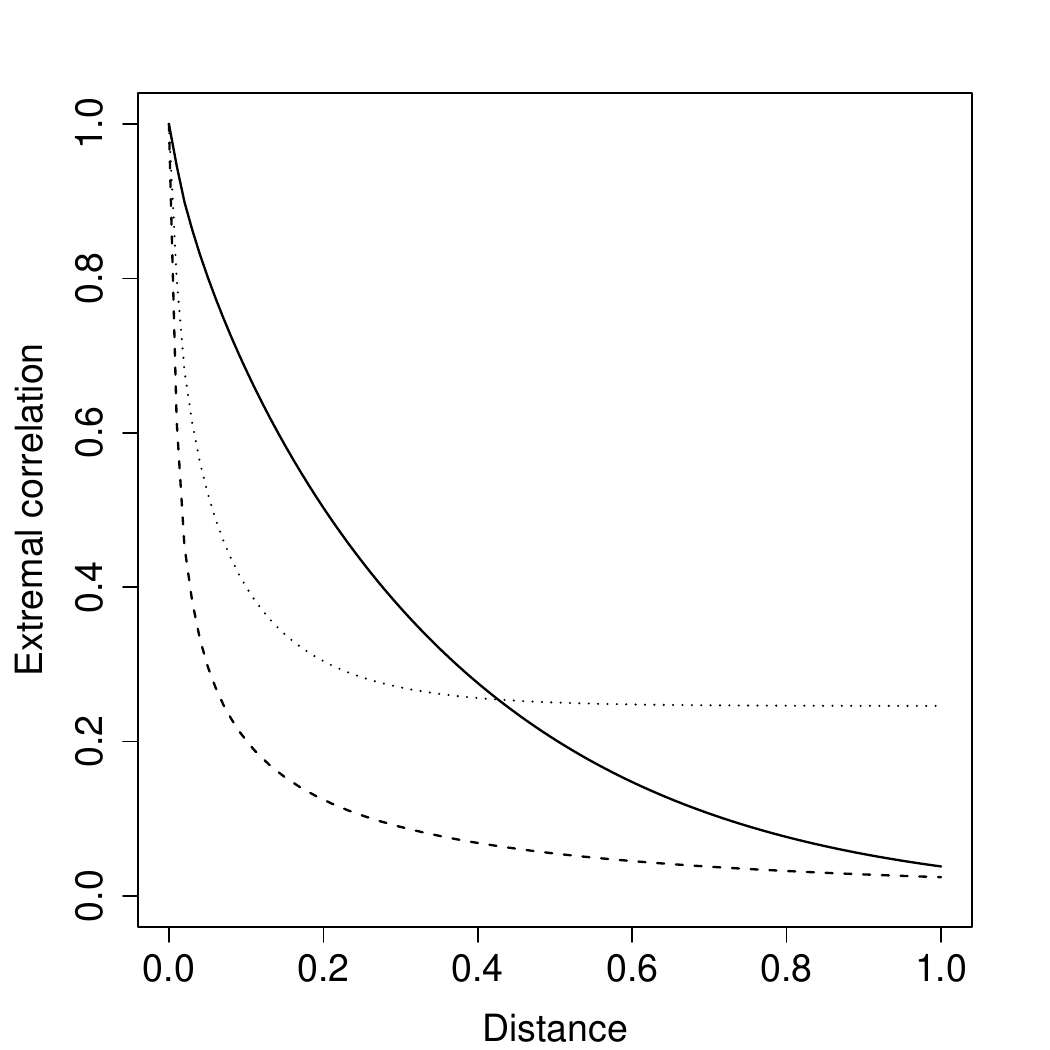}\hspace{1cm}
	\includegraphics[width=0.45\textwidth]{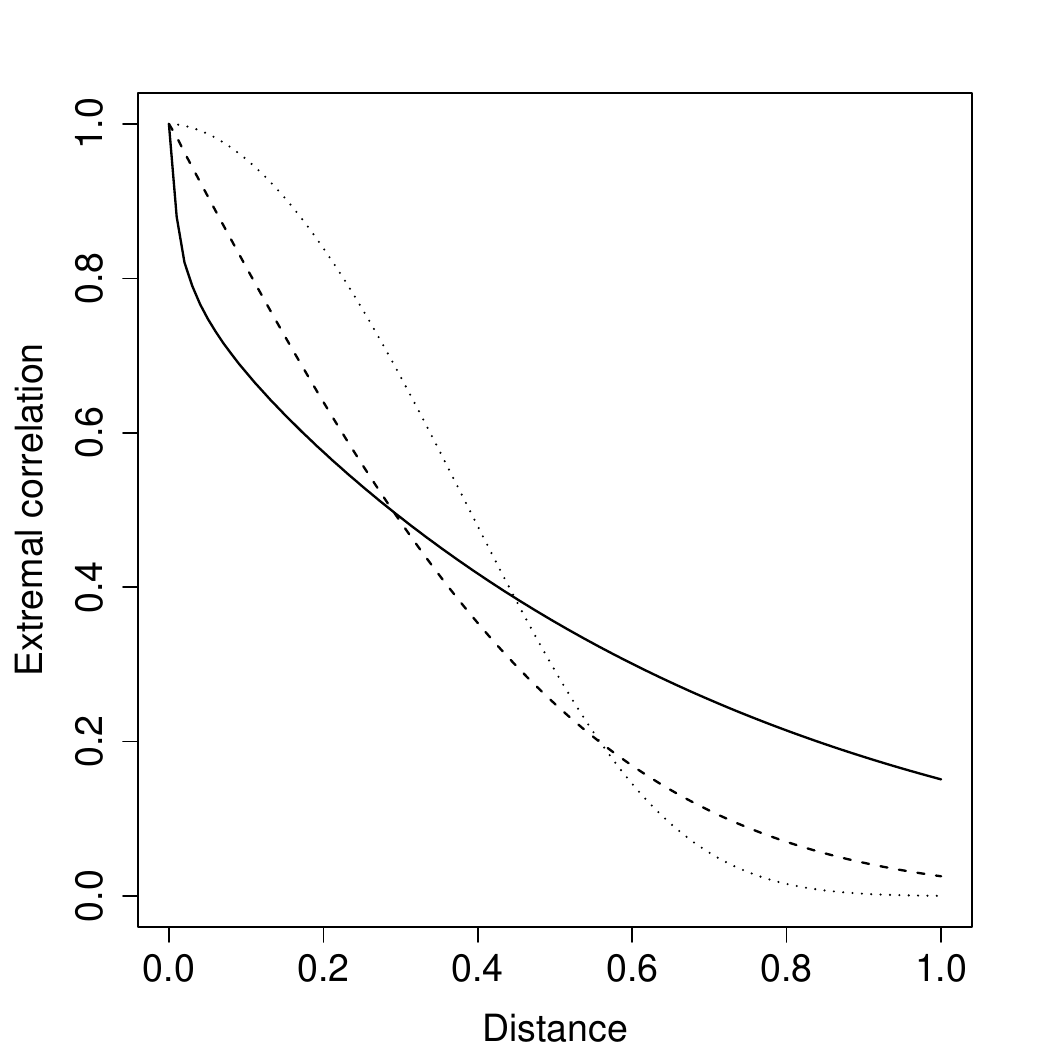}
	\caption{\label{Fig:extremal_correlation} Extremal correlation, $\chi(h)$, as a function of distance, $h$, of the \BR{} processes when $d=2$. Left: Classical fractional model ($\alpha=0$) with $\beta=1.4$ (solid) and $\beta=1.2$ (dashed), and the proper \WM{} model ($\beta=0$) with $\alpha=1.4$ (dotted line). Right: Our new \WM{} \BR{} model with $\alpha=-0.5$, $\beta=1.7$ (solid), $\alpha=\beta=1$ (dashed), and $\alpha=1.5$, $\beta=1.4$ (dotted).}
\end{figure}

Alternatively, a seemingly natural idea to overcome limitations (ii) and (iii) is to use a stationary proper \WM{} field as Gaussian process $W(\mv s)$ with covariance function $C(h)$ in the \BR{} construction \eqref{spec_Gauss}; see, e.g., \cite{hea2024} who use such a model. While this provides flexibility, there are strong limitations concerning~(i): since the variogram $\gamma(h) = 2\{C(0) - C(h)\}$ of a stationary field is bounded, there is a lower bound on the extremal correlation given by 
$\chi(h) \geq  2 - 2\Phi(\sqrt{C(0)/2})$, $h\geq 0$, 
which induces long-range dependence characteristic of non-ergodic fields. 
This is illustrated by the dotted line in the left panel of Figure~\ref{Fig:extremal_correlation}.
Moreover, it turns out that if a sparse approximation (e.g., through finite elements) is used for the stationary Gaussian field, then the corresponding finite-dimensional \HR{} marginals of the \BR{} process is dense in the sense of extremal graphical models; see Proposition~\ref{thm:GR} in the appendix for details.

We now introduce the \WM{} \BR{} process that combines both model classes and solves the issues above in a natural and elegant way.

\begin{definition}
\WM{} \BR{} processes are the class of \BR{} processes with parameters $\beta\in [0,2)$, $\alpha, \tau, \kappa \geq 0$  and spectral function
\begin{align}\label{spectral_WM}
    V(\mv s) = u(\mv s) - u(\mv s_0) - \gamma(\|\mv s - \mv s_0\|)/2, \quad \mv s\in\mathcal D,
\end{align}
where $u$ is the intrinsic \WM{} field
with variogram $\gamma$
from Definition~\ref{def_WM}, and $\mv s_0\in\mathcal D$ is an arbitrary location.
\end{definition}

Note that the field $u(\mv s) - u(\mv s_0)$ is a 
proper Gaussian field with variance function $\sigma^2(\mv s) = \gamma(\|\mv s - \mv s_0\|)$, explaining the link with~\eqref{spec_Gauss}.
The classes of fractional and proper \WM{} Brown--Resnick processes correspond to the special cases with $\alpha = 0$ and $\beta =0$, respectively.
Our new model generalizes these classes by decoupling the local and global behavior of the variogram: the long range behavior is determined by $\beta$, the short range behavior is determined by $\alpha+\beta$, and the transition between the two scales is determined by $\kappa$; recall Proposition \ref{prop:LG}.
The right panel of Figure~\ref{Fig:extremal_correlation} shows the 
new flexibility of this model class in terms of extremal correlation functions
for three parameter combinations. 
The model therefore addresses the desired properties (i) and (ii) above; natural ways to adapt the process to non-stationary and non-Euclidean settings will be discussed in Section~\ref{sec:discussion}.

Our \WM{} \BR{} prcocess is the first spatial extremes model
that allows for sparse computations
and enables inference even for
datasets with many spatial locations, thus satisfying point (iii) above. 
Indeed, relying on the FEM approximation developed in Section~\ref{FEM}, we show that its finite-dimensional \HR{} 
margins are well approximated by
sparse extremal graphical models on a suitable
triangulation $(\mathcal T_\meshwidth)_{\meshwidth\in(0,1)}$ be a triangulation with diameter $\meshwidth$ and $N=N_\meshwidth$ vertices $\mv v_1,\dots, \mv v_N$; efficient estimation is discussed in the next section. Let 
\begin{align}\label{fem_V}
V_\meshwidth(\mv s) = \widehat u_\meshwidth(\mv s) - \widehat u_\meshwidth(\mv s_0) - \gamma_\meshwidth(\|\mv s-\mv s_0\|)/2,    
\end{align}
be the finite-element spectral function, for some $\mv s_0\in\mathcal D$, 
where $\widehat u_\meshwidth$ is the FEM approximation of the intrinsic field $u$ in~\eqref{eqn:FEMp} and $\gamma_\meshwidth$ its variogram.

\begin{theorem}\label{thm:Par_int_link}
    If $m$ and $\widetilde m$ are chosen as in Remark \ref{rem:mtm}, then the max-stable and $r$-Pareto \WM{} \BR{} processes $\eta_\meshwidth$ and $Y^{(r)}_\meshwidth$, respectively, with spectral function~\eqref{fem_V} converge to the corresponding \WM{} \BR{} process with the same parameters, i.e.,
    $\eta_\meshwidth \Rightarrow \eta$ and $Y^{(r)}_\meshwidth \Rightarrow Y^{(r)}$, as $\meshwidth\to 0$, where the convergence is in terms of finite-dimensional distributions.
\end{theorem}

Theorem \ref{thm:Par_int_link} follows from Theorem \ref{cov_fem_approx_rate_general} and the fact that both the max-stable and $r$-Pareto \HR{} distributions are continuous in $\Gamma$. Equation \eqref{eqn:HRCI} implies that the \HR{} distribution inherits is extremal conditional independence structure from the conditional independence structure of its corresponding intrinsic GMRF. For integer $\beta \geq 1$ and $\alpha$, the conditional independence structure of $\widehat u_\delta$ evaluated at the mesh nodes $\mv v_1, \dots, \mv v_N$ is sparse and thus the extremal conditional independence structure of the $r$-Pareto process $Y_\delta^{(r)}$ is also sparse. For example, if $\alpha=\beta=1$ then $Y^{(r)}_\delta$ evaluated at the vertices $\mv v_1, \dots, \mv v_N$ follows a \HR{} distribution that is an extremal graphical model on the graph corresponding to the triangulation $\mathcal T_\meshwidth$.

\subsection{Estimation}\label{sec:LE1}

Suppose we have $n$ observations $\mathbf y_1, \dots, \mathbf y_n \in \mathbb R^k$ of the $r$-Pareto process $Y^{(r)} = \{Y^{(r)}(\mv s): \mv s\in \mathcal D\}$ at locations $\mv s_1,\dots, \mv s_k \in \mathcal D\subseteq \mathbb R^d$. We describe how to use sparse computations to approximate the likelihood function and estimate the parameters of the \WM{} \BR{} model with parameters $\beta\in [0,2)$, $\alpha, \tau, \kappa \geq 0$.

The log-likelihood of the observations $\mathbf y_1, \dots, \mathbf y_n$ of $Y^{(r)}$ may be approximated by the log-likelihood of the finite-element Pareto process $Y^{(r)}_\meshwidth$ given by
\begin{align*}
        L(\tau,\kappa, \alpha, \beta, \sigma^2) &= - n\log C_{r,\widetilde \Theta} +  \sum_{i=1}^n \log \lambda(\mathbf y_i; \widetilde{\Theta})
\end{align*}
where $\widetilde \Theta$ is the precision matrix of the intrinsic GMRF $\widetilde{\mathbf W}$ in~\eqref{eqn:Hich}. In general the normalizing constant $C_{r,\widetilde \Theta}$ from~\eqref{eqn:frd} depends on the model parameters. In the sequel, we consider the two popular risk functions mentioned after~\eqref{def_pareto}, since then the normalizing constant is independent of the model parameters \citep{Eng15}. 
Together with the representation~\eqref{eqn:ExpDN} of the exponent measure density, this yields 
\begin{align}\label{sparse_lik}
        L(\tau,\kappa, \alpha, \beta, \sigma^2) &\equiv \sum_{i=1}^n  \log f_{\widetilde{\mathbf W}}(\mathbf y_i + \widetilde{\Gamma}_{\cdot, j_0} /2),
\end{align}
where $j_0\in\{1,\dots,k\}$ is an arbitrary index, $f_{\widetilde{\mathbf W}}$ is the density  
of the observations of the intrinsic \WM{} field, and $\widetilde \Gamma$ is the variogram in~\eqref{eqn:nugg}; here $\equiv$ means up to an additive constant. Maximizing the \HR{} likelihood is thus equivalent to maximizing the likelihood of the intrinsic GMRF $\widetilde{\mathbf W}$, except for $\widetilde{\Gamma}_{\cdot, j_0}$, the $j_0$-th column of the variogram matrix, which depends on model parameters.
The term $\widetilde{\Gamma}_{\cdot, j_0}$ can be computed efficiently using $\widetilde{\Gamma}_{-j_0, j_0}=\text{diag}\{(\Theta^{(j_0)})^{-1}\}+\sigma^2\mv 1$, where $\Theta^{(j_0)}$ is $\Theta$ with the $j_0$-th row and column deleted and $\sigma^2$ is the nugget. We can then efficiently evaluate \eqref{sparse_lik} using \eqref{eqn:LCom}. In practice, the term $\sigma^2\mv 1$ can be dropped because $f_{\mv{\widetilde w}}(\mv{\widetilde{w}})=f_{\mv{\widetilde{w}}}(\mv w+c\mv1)$ for all $c \in \mathbb{R}$ since it is a first-order intrinsic density. 
The computational savings
for simulation and estimation methods generally mirror the $O(k^{3/2})$ versus $O(k^3)$ computational costs for sparse versus dense proper GMRFs.

\subsection{Conditional simulation and extremal kriging}\label{sec:extkrig}

Given a set of $k'$ unobserved locations $\mv u_1,\dots, \mv u_{k'}\in \mathcal D$, conditional simulation of the $r$-Pareto process can be performed efficiently for the \WM{} \BR{} model.
We let the risk functional 
$r: \mathbb R^{k+k'} \to \mathbb R$
only depend on the variables
in $\mathbb R^k$ at the observed locations 
$\mv s_1,\dots, \mv s_k$, that is, $r(\mv y, \mv y') = \tilde r(\mv y)$
for some risk $\tilde r: \mathbb R^{k} \to \mathbb R$.
For any $\mv y\in\mathbb R^k$ such that $\tilde r(\mv y) > 0$, the conditional density is then 
\begin{equation}\label{cond_dens}
f^{(r)}(\mv{y}' \mid \mv y ; \Theta) = \frac{\lambda(\mv y, \mv y'; \Theta)}{\lambda(\mv y; \tilde \Theta)}, \quad \mathbf y' \in \mathbb R^{k'},
\end{equation}
where $\Theta$ is the joint precision matrix of all 
locations implied by the model, and the precision matrix
of the marginal model corresponding only to observed locations 
can be written as the Schur complement $\tilde \Theta = \Theta_{OO} - \Theta_{OU}\Theta_{UU}^{-1} \Theta_{UO}$ \citep{eng_taeb}; here $O$ and $U$ are the index sets
of observed and unobserved locations, respectively.

\begin{proposition}\label{prop:HR_cond_distr}
    The conditional density $f^{(r)}(\mv{y}' \mid \mv y ; \Theta)$ in~\eqref{cond_dens}
    is the density of a proper multivariate normal distribution with mean and precision matrix given by 
    \begin{align*}
    \mu_{U\mid O} = - \Theta_{U,U}^{-1} \Theta_{U,O}\mathbf y_{O} - \Theta_{U,U}^{-1} \mathbf v_U,\quad  
    \Theta_{U\mid O} &= \Theta_{U,U},
\end{align*}
where the vector $\mv v = \Theta \diag(\Theta^+)/2 + \mv 1/d$ is called the resistance
curvature \citep{DEVRIENDT202468}; here $\Theta^+$ denotes the Moore--Penrose pseudoinverse of $\Theta$. If the set $U$ consists of a single point indexed by 0, then the formula simplifies to
    \begin{align}\label{eqn:ext_krig_form}
    \mu_{\{0\}\mid O} = -\sum_{i\in O}\frac{\Theta_{0,i}}{\Theta_{0,0}}y_i  - \frac{v_0}{\Theta_{0,0}},\quad  
    \Theta_{\{0\}\mid O} &= \Theta_{0,0},
\end{align}
\end{proposition}

This result allows for efficient conditional simulation since the
matrix $\Theta_{U,U}$ is sparse. A similar formula has been obtained in
\cite{dom2012}, which has however a more involved form and does not 
use the \HR{} precision matrix. 
Our formula uses the Fiedler--Bapat identity \citep[Corollary 3.7]{devriendt2022a}
to make the density independent of a conditioning variable $m$ as used in~\eqref{eqn:ExpDN}.

Our formula yields \eqref{eqn:ext_krig_form} as a special case, which is a simple expression for extremal kriging for the $r$-Pareto process. 
The conditional mean $\mu_{\{0\}\mid O}$ above is equivalent to classical kriging formula in~\eqref{krig2} up to the $-v_0/\Theta_{0,0}$ term. Moreover, \eqref{eqn:ext_krig_form} provides insight into the definition of extremal conditional independence introduced in \eqref{eqn:HRCI}: at an unobserved location $0$, extremal conditional independence reduces to (non-extreme) conditional independence. This is because \eqref{eqn:ext_krig_form} implies that the conditional distribution of the field at location $0$ depends only on the observed locations $i$ with $\Theta_{0,i} \neq 0$.
Figure~\ref{kriging_ex_extremes} illustrates extremal kriging for two constellations of observations using the five different
models also considered in Figure~\ref{kriging_ex}, with and without a nugget.
The intrinsic models with $\beta >0$ have a downward trend 
when moving away from the observations. This is highly desirable in spatial
extremes applications, since on the original data scale
locations farther away then do not inherit the extremeness of the observed locations. 
This is a natural assumption if there is no strong long-range dependence.

\begin{figure}[t!]
\centering
\includegraphics[width=0.33\textwidth]{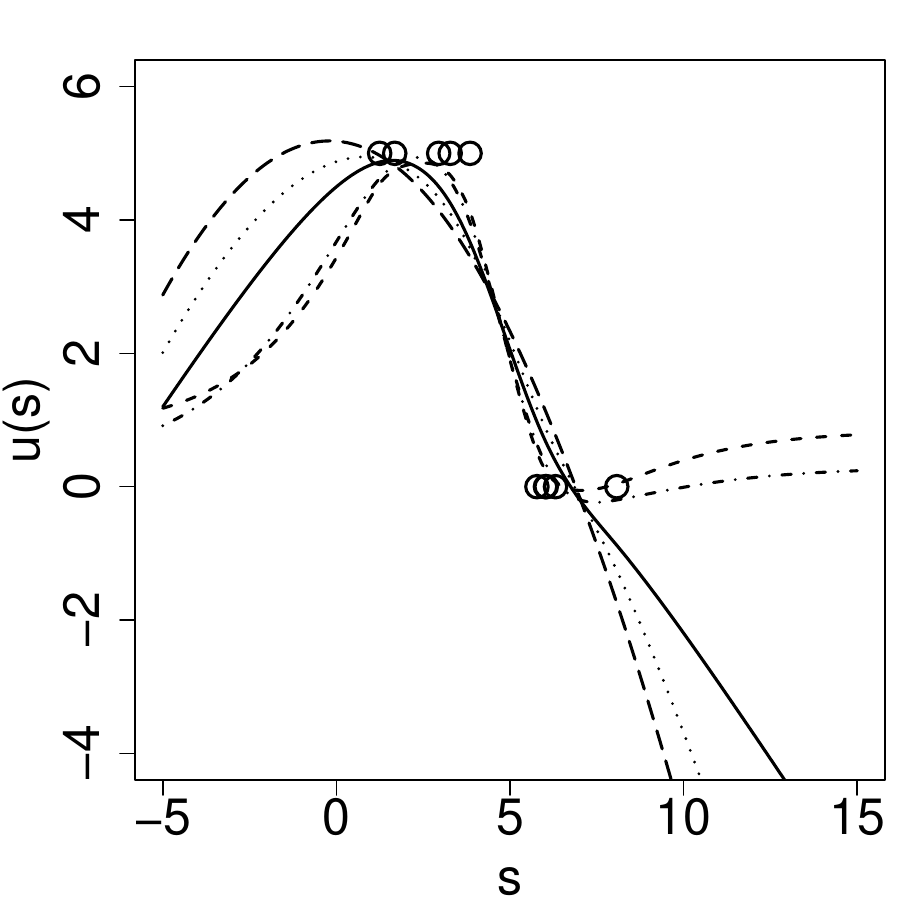}
\includegraphics[width=0.33\textwidth]{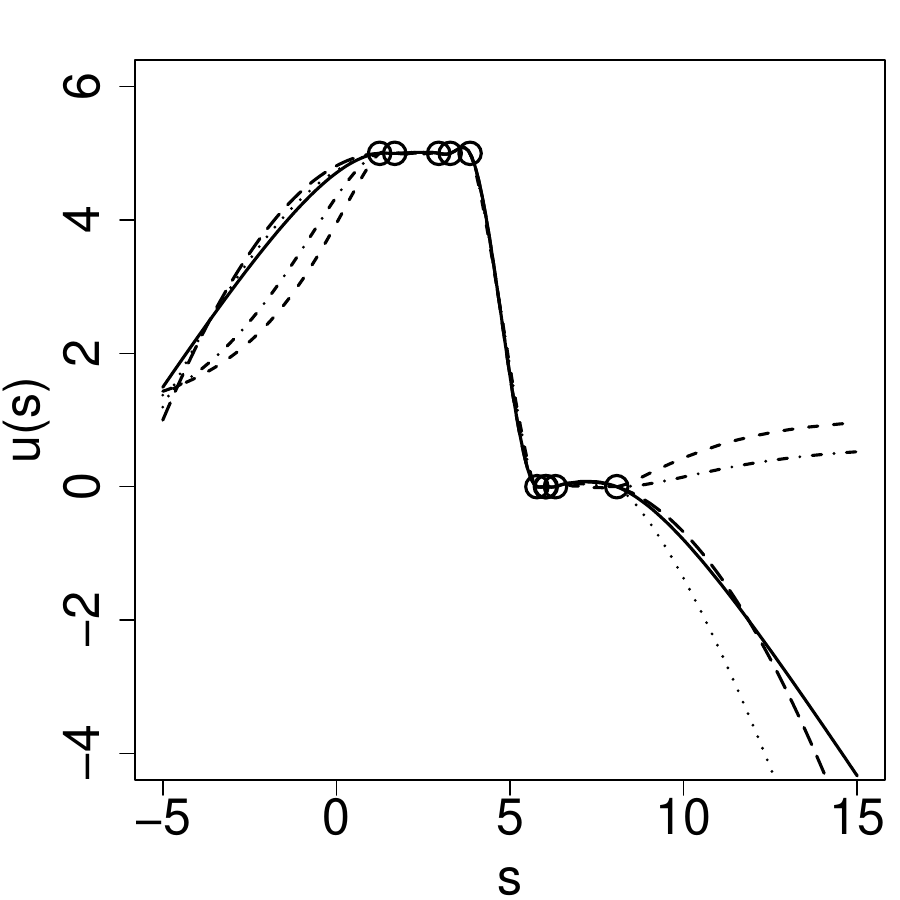}
\includegraphics[width=0.33\textwidth]{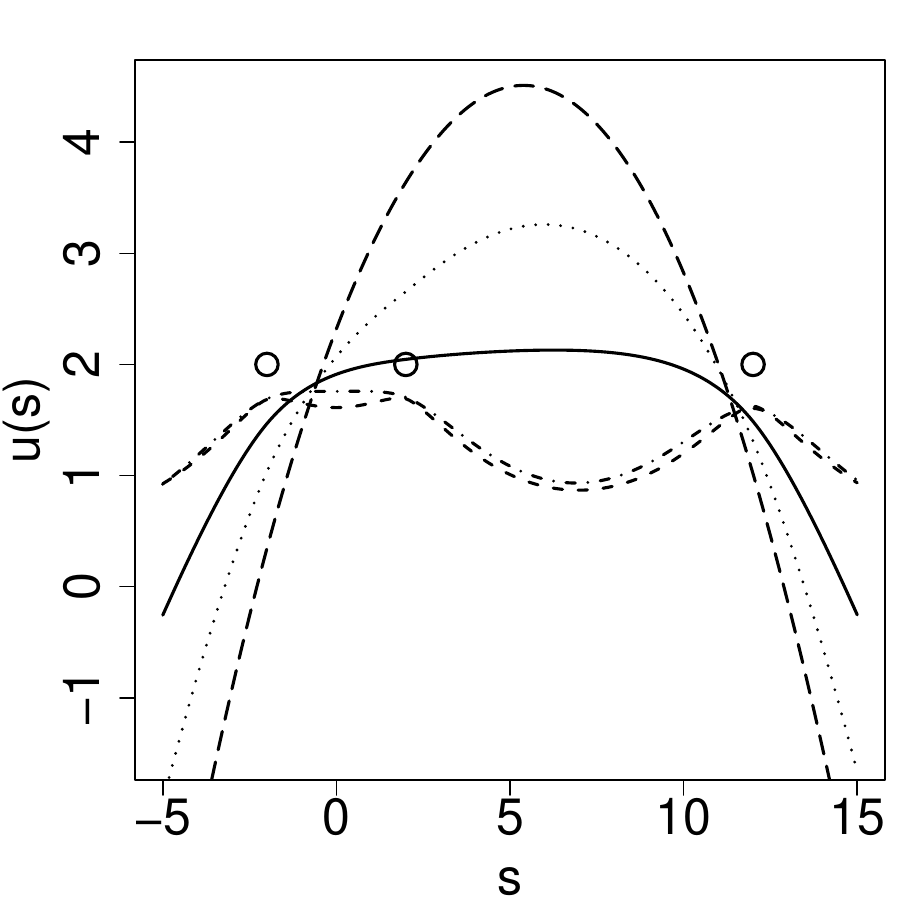}
\includegraphics[width=0.33\textwidth]{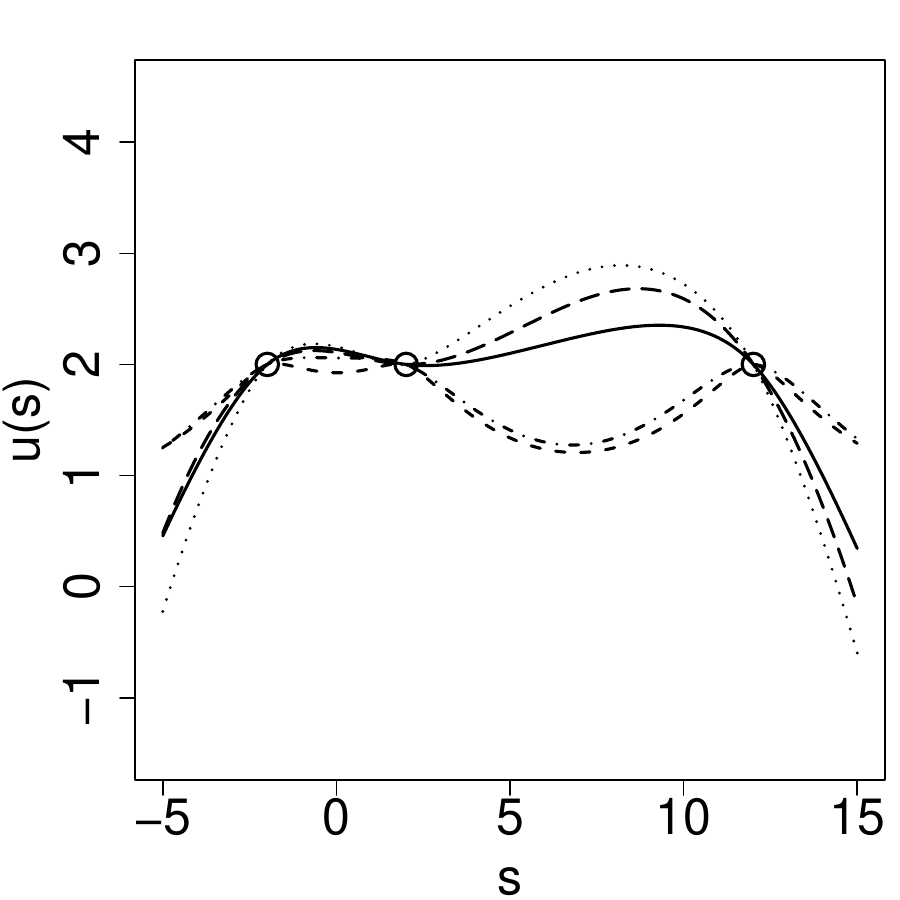}
\caption{\label{kriging_ex_extremes} Extremal kriging estimates computed using the varigoram displayed in Figure~\ref{kriging_ex} with a nugget of size 1 (left) and without a nugget (right).}
\end{figure}

\subsection{Application: Marine heat waves}\label{sec:A1}

We apply our methodology to study the spatial extent of marine heat waves over the Southwest Pacific Ocean, off the eastern coast of Australia; see Figure~\ref{fig:austrlia}. Marine heat waves are prolonged periods of high sea surface temperature (SST) anomalies and they have become increasingly frequent and intense due to climate change \citep{Hughes2017, Holbrook2019}. The study region includes the Coral Sea, which contains the Great Barrier Reef, the world's largest coral reef system. The latter is particularly vulnerable to marine heat waves and has experienced several events in the last century, which have led to mass coral bleaching, reduced biodiversity, and long-term reef degradation \citep{Oliver2018}.
The spatial pattern of coral bleaching has been directly linked to 
the spatial distribution of SST anomalies during marine heat waves; see \citet[][Figure 1]{Hughes2017} for examples of the 1998, 2002 and 2016 marine heat waves. It is thus crucial to build statistical tools to better model and predict these large-scale spatial extreme events \citep[e.g.,][]
{hazra2021estimating}.
Our \WM{} \BR{} model is well suited for this purpose.

\begin{figure}[t!]
	\includegraphics[width=0.32\textwidth]{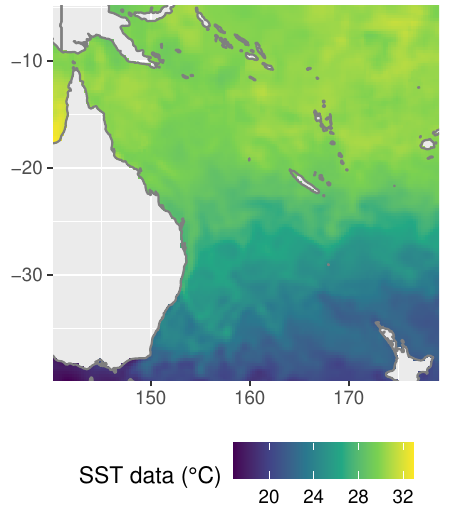}
    \includegraphics[width=0.32\textwidth]{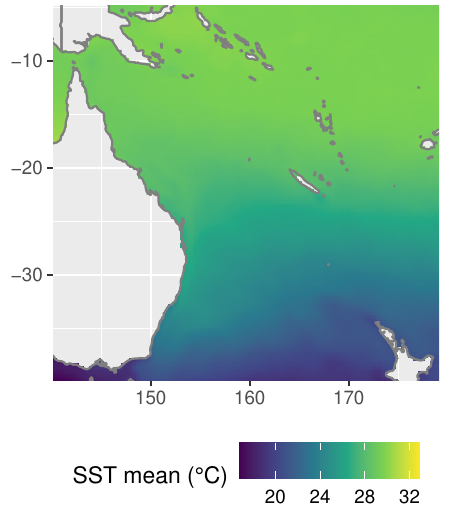}
    \includegraphics[width=0.32\textwidth]{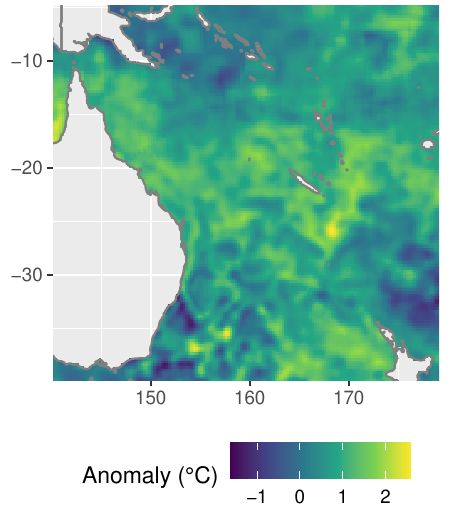}
	\caption{Map of the study region (east of Australia) and plot of sea surface temperature data for March 5, 2015, the day with the largest spatial SST average over the domain. Left: original SST data; Middle: estimated mean; Right: SST anomalies.}
	\label{fig:austrlia}
\end{figure}

We analyze data from the Optimum Interpolation Sea Surface Temperature (OISST) global product\footnote{Data available from \url{https://www.ncei.noaa.gov/products/optimum-interpolation-sst}}, which incorporates observations from different sources (satellites, ships, buoys and Argo floats) into a high-resolution regular grid over the World Ocean. Specifically, we consider daily SST data from January 1, 1982, to December 31, 2019, available on a latitude-longitude grid at $0.25^\circ\times0.25^\circ$ resolution. Our area of interest includes $17530$ grid cells. We use the first 37 years for model fitting and keep data from 2019 to illustrate spatial prediction with our model. 
In order to obtain SST anomalies, we first estimate the spatiotemporal mean temperature by fitting a linear regression model in each grid cell using twelve cyclic B-splines to capture the seasonal component and a linear time trend to capture the effect of climate change; daily anomalies are then obtained by subtracting the estimated mean from the data. Figure~\ref{fig:austrlia} illustrates the original data, the estimated mean, and the resulting anomalies for a particularly extreme day. Since we focus on assessing the spatial extent of marine heat waves, we model the extremal spatial dependence of SST anomalies averaged over non-overlapping 14-day windows, resulting in $964$ spatial observations from 1982--2018. This biweekly aggregation also helps to significantly reduce temporal dependence, which would be complex to model.

After standardizing the aggregated data to the unit Fr\'echet scale via marginal rank transforms, we compute for each 14-day window the spatial average of the data over the entire domain and extract the $50$ spatial fields with the largest spatial average (i.e., about $5\%$ spatial extreme events). Treating these as independent realizations from the same extremal model, we fit our proposed \WM{} \BR{} $r$-Pareto process with parameters $(\tau,\kappa,\alpha,\beta)$, defined through \eqref{spec_Gauss} based on the SPDE~\eqref{SPDE}. We discretize the domain into a fine triangulated mesh with $4834$ nodes and perform maximum (full) likelihood inference using the selected $17530$-dimensional observed vectors. Interestingly, when fitting the general model, $\beta$ is estimated close to one, the value at which the long-range behavior of the variogram transitions from being bounded to unbounded. The value $\beta=1$ is in fact special as it is the only case where the long-range behavior is logarithmic; recall Proposition~\ref{prop:LG}. To isolate the effect of $\beta$ on the fitted model, we fit submodels where $(\tau,\kappa,\alpha)$ are estimated with $\beta=0,0.5,1,1.5,2$ held fixed. For comparison, we also consider the simpler model where $\alpha=\beta=1$, and the popular fractional variogram model where $(\tau,\kappa,\beta)$ are estimated with $\alpha=0$ held fixed.

\begin{figure}[t!]
    \centering
	\includegraphics[width=0.4\textwidth]{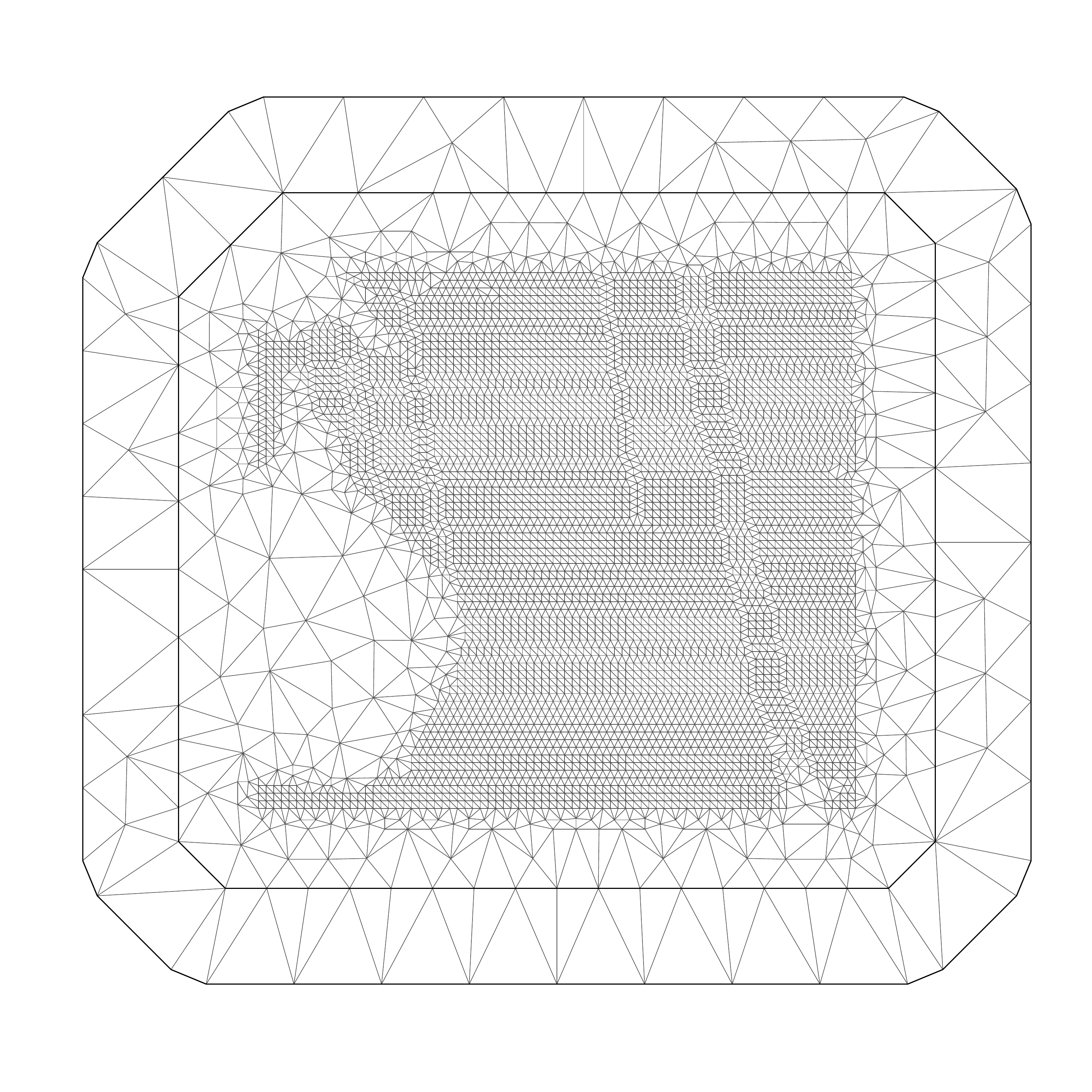}
    \includegraphics[width=0.59\textwidth]{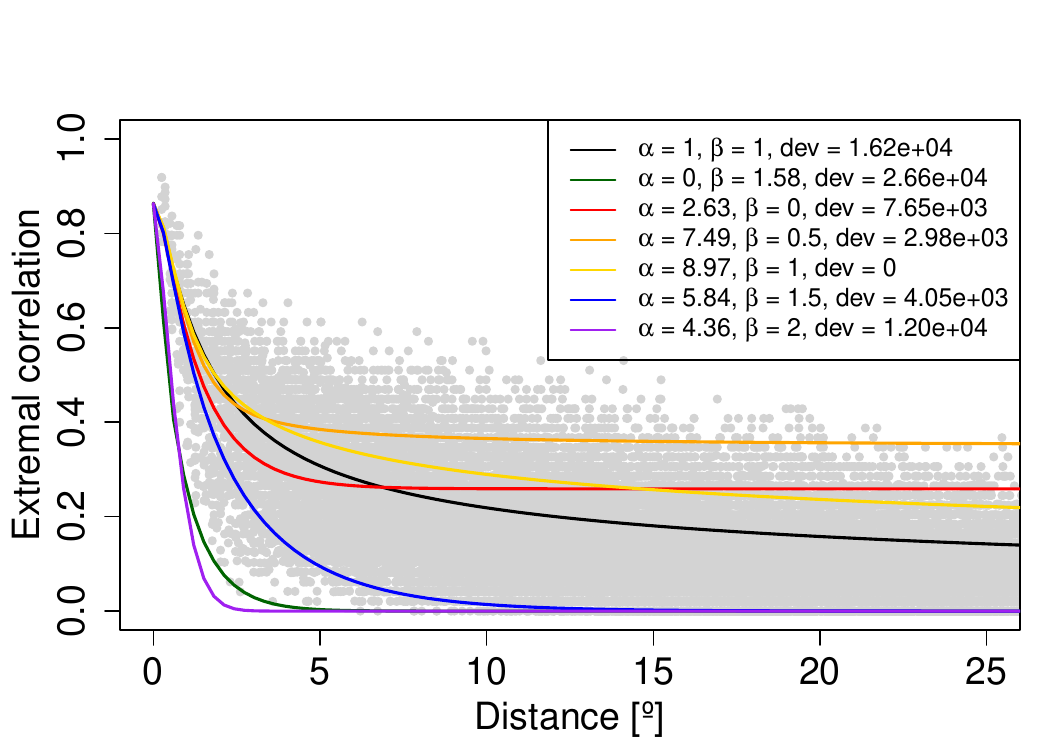}
	\caption{Left: Mesh used for model fitting. Right: Empirical (gray dots) and model-based (colored lines) extremal correlation function, plotted as a function of spatial distance. Empirical estimates are based on a random subset of 300 locations. The fitted curves correspond to the models with $\alpha=\beta=1$ (black); $\alpha=0$ and $\beta$ estimated (dark green); and $\alpha$ estimated with $\beta=0,0.5,1,1.5,2$ (red, orange, gold, blue, purple, respectively). Specific (fixed or estimated) values of $(\alpha,\beta)$ and the deviance with respect to the best model are provided in the legend.}
	\label{fig:fittedextremalcorrelation}
\end{figure}

Figure~\ref{fig:fittedextremalcorrelation} shows the mesh used for model fitting and the estimated extremal correlation functions \eqref{chi_coeff} resulting from the fitted models. The right panel displays the (fixed or estimated) $\alpha$ and $\beta$ parameters, and the deviance (i.e., twice the difference in log-likelihood) of each fitted model with respect to the best one ($\alpha$ estimated with $\beta=1$ fixed). Models with a low value of $\beta$ (i.e., the non-intrinsic models with $\beta=0,0.5$) clearly lack flexibility to capture long-range independence, and this leads to poor fits at both short and long distances. On the other hand, models with a large value of $\beta$ (i.e., intrinsic models with $\beta>1$) also display poor fits as the fitted extremal correlation function decays too fast. Interestingly, the fractional variogram model with $\alpha=0$, often used in spatial extremes analysis, provides the worst fit across all models, owing to its inability to separately control the short- and long-range behavior. Since the data are locally smooth in this data application, this forces the $\beta$ parameter to be quite large (with an estimate of $1.58$), which cannot be easily compensated for by the other parameters. Overall, the best fits are thus obtained with $\beta=1$. While the likelihood is maximized when $\alpha$ is estimated (with $\beta=1$), the extremal correlation function seems to be slightly better captured at long distances with the model where $\alpha=\beta=1$.

To illustrate the ability to perform extremal kriging with our model (recall Section~\ref{sec:extkrig}), we consider the 14-day window in 2019 that has the largest SST anomaly average (i.e., January 1--14). We then mask $25\%$ of the observations at exactly $4382$ locations, and predict these masked values by conditional simulation from each of the fitted models. Given that our model characterizes the joint upper tail, we focus on predicting moderately large values and therefore select the $25\%$ masked locations in areas where the actual true observation exceeds its corresponding marginal median. Extremal kriging is done on the unit Fr\'echet scale with the spatial average as risk functional, and the simulated values are then back-transformed to the original scale of temperature anomalies. Figure~\ref{fig:conditionalsim} displays the true averaged SST anomalies for the selected time window, the mask (i.e., missingness pattern), the corresponding conditioning values, as well as the conditional mean and two conditional simulations from the best model ($\alpha$ estimated and $\beta=1$ fixed).
\begin{figure}[t!]
    \centering
	\includegraphics[width=0.32\textwidth]{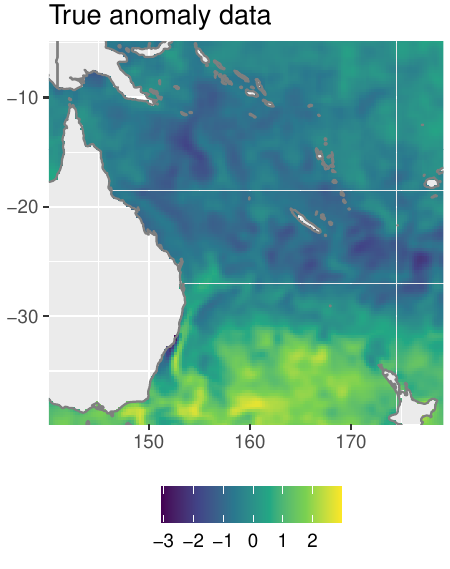}
    \includegraphics[width=0.32\textwidth]{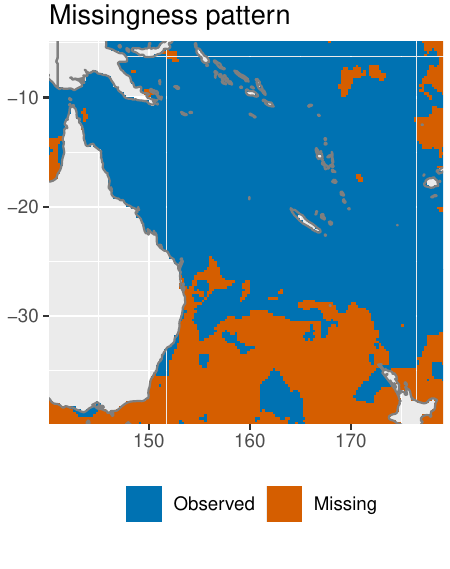}
    \includegraphics[width=0.32\textwidth]{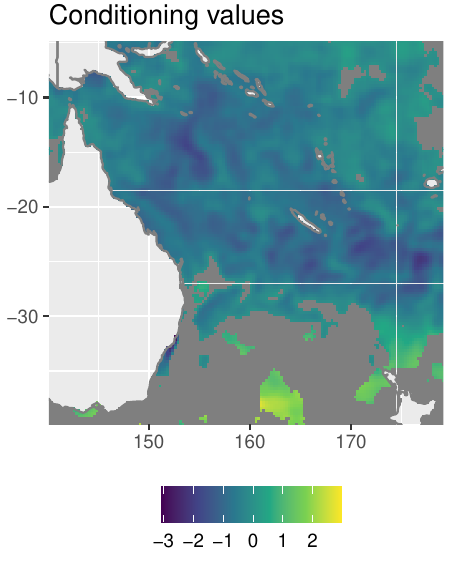}
    \includegraphics[width=0.32\textwidth]{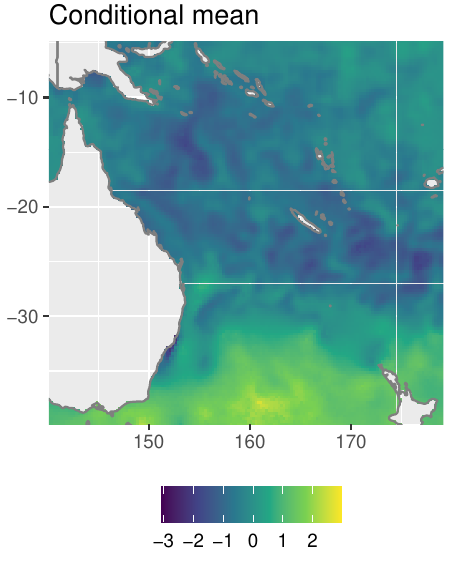}
    \includegraphics[width=0.32\textwidth]{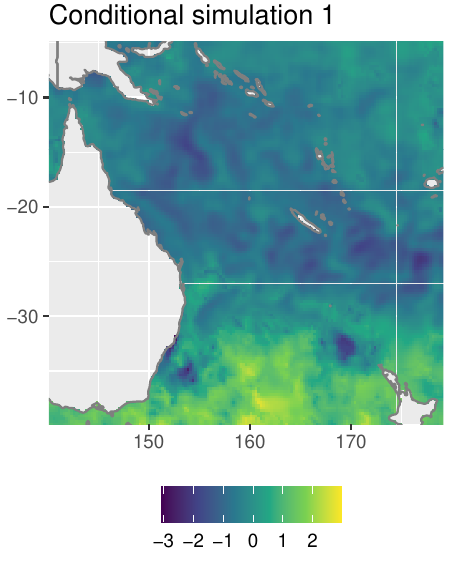}
    \includegraphics[width=0.32\textwidth]{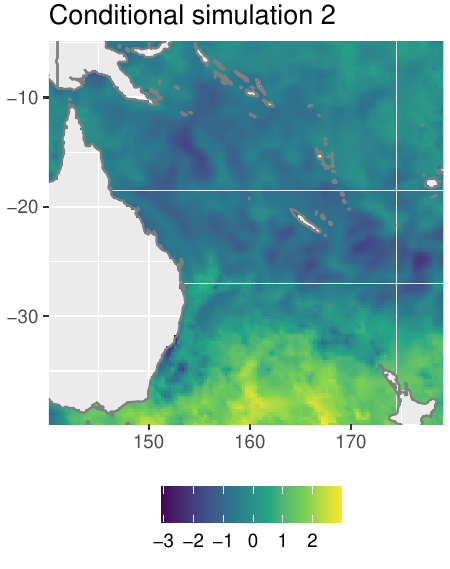}
	\caption{Top left: Averaged SST anomaly data for January 1--14, 2019. Top middle: mask (i.e., missingness pattern). Top right: Corresponding values used for conditional simulation. Bottom: Conditional mean (left) and two conditional simulations (middle and right) from the best fitted model ($\alpha$ estimated and $\beta=1$ fixed), displayed on the original scale of the averaged SST anomalies.}
	\label{fig:conditionalsim}
\end{figure}
Given that we masked a substantial area containing the majority of high values, these conditional simulations are impressively good. The magnitude and spatial pattern of extreme events are emulated well, although the simulations appear somewhat grainier than the true data due to the measurement error (nugget) term. We compare all fitted models on this spatial prediction task by generating $100$ conditional simulations (with the same random seed for each model) and computing the root mean squared prediction error (RMSPE) on the original scale of the data, averaged across all masked locations. Table~\ref{tab:RMSPE.SST} reports the results.
\begin{table}[t!]
\centering
\caption{Root mean squared prediction error (RMSPE) for all fitted models based on 100 conditional simulations in the setting of Figure~\ref{fig:conditionalsim}, along with the (fixed or estimated) $\alpha$ and $\beta$ values. The best model (lowest RMSPE) is highlighted in bold.\label{tab:RMSPE.SST}}
\begin{tabular}{rr|ccccccc}
Model & $\alpha$ & 1 & 0 & 2.63 & 7.49 & 8.97 & 5.84 & 4.36 \\
 & $\beta$ & 1 & 1.58 & 0 & 0.5 & 1 & 1.5 & 2 \\ \hline
\multicolumn{2}{r|}{RMSPE} & \textbf{.560} & .707 & .722 & .827 & .564 & .708 & .656
\end{tabular}
\end{table}
Interestingly, we see that models with $\beta=1$ significantly outperform the other models, and the model with $\alpha=\beta=1$ is the best overall. Conditional means and simulations from the other models are shown for comparison in Appendix~\ref{sec:sst_anom}. Conditional simulations from the two models with $\beta=1$ have a comparable behavior. 

\section{Discussion}\label{sec:discussion}

\textbf{Summary.} Representing intrinsic Gaussian random fields as solutions to SPDEs is a powerful tool for data applications requiring spatial, temporal, or spatio-temporal modeling, and it opens up a wide range of possibilities across many subfields of statistics. We propose the flexible class of intrinsic \WM{} processes and studied their kriging properties with a particular focus on their extrapolation capabilities. Developing the ``SPDE approach'' for intrinsic fields with novel FEM approximation and software, we illustrate the benefits of this new framework in two very different areas of application: longitudinal data analysis with an application to kidney function data, and spatial extremes modeling with an application to marine heatwave assessment. Both applications provide strong evidence that the flexibility of our new model leads to improvements in model fits and predictions compared to existing methods, and that estimation and conditional simulation can be performed at a reasonable computational cost given the size of the problem. Our new sparse extremes approach allowed modeling sea surface temperature data at more than 17000 locations, which is unprecedented in the field. Our applications showcase just a few possibilities of our new framework. We briefly list possible extensions and avenues for future research.

\vspace{5pt}

\noindent\textbf{Extensions of our intrinsic SPDE framework.} One of the main benefits of working with SPDEs is that they can easily be extended to more complex settings. In the classical context of proper Gaussian random fields, a wide range of SPDE extensions have been proposed to handle anisotropic and nonstationary \citep{Lin11,fuglstad2015exploring,bakka2019nonstationary} or (nonseparable) spatiotemporal \citep{lindgren2020diffusion} dependence on Euclidean domains, the sphere \citep{Lin11}, metric graphs \citep{BSW2022}, or other manifolds \citep{mejia2020bayesian}. With careful adjustments, similar extensions could be applied in the context of our intrinsic \WM{} fields and, by extension, of \WM{} \BR{} random fields for extremes.

\vspace{5pt}

\noindent\textbf{Avenues for spatial extremes.}
Unconditional simulation of $r$-Pareto and max-stable processes requires generating multiple samples of the spectral functions \citep{Dom16,Liu19,Oes22}. For Brown--Resnick processes this amounts to the simulation of the Gaussian process in~\eqref{spec_Gauss}. Using our intrinsic \WM{} SPDE model and its sparse FEM approximation, efficient simulation from $r$-Pareto and max-stable processes becomes possible in very high dimensions.

When modeling spatial extremes through $r$-Pareto processes, it is customary to censor observations falling below a selected marginal threshold, to prevent them from biasing estimation of the tail dependence structure. When multivariate censoring is required in high dimensions, or when general risk functionals $r$ are considered, likelihood-based inference is typically slowed down significantly. Our sparse SPDE approach could be naturally combined with the alternative score-matching method of \citet{defondeville2018high} for significant computational savings.

Max-stable and $r$-Pareto processes are characterized by asymptotic tail dependence. To extend our sparse spatial extremes approach to the asymptotic  tail independence case, it would be interesting to adapt our SPDE-based \BR{} intensity to the domain-scaled regular variation framework of \citet{strokorb2025domain}, or by letting its spatial dependence range decrease with the value of the risk functional $r$, similar to the approach of \citet{zhong2022modeling} for max-infinitely divisible processes. This would extend the theoretical framework of \citet{engelke2024infinite} to the spatial context.
\vspace{5pt}

\noindent\textbf{Avenues for the modeling of areal data.}
In ecology and geoenvironmental science \citep{lombardo2020spacetime}, socio-economics and political science \citep{horan2024multilevel}, as well as public health (disease mapping) and epidemiology \citep{Konstantinoudis2022regional}, data are often observed or aggregated at the county/state/regional level, or over grid cells or alternative mapping units, and thus require specialized models for areal data. In such cases, popular spatial processes (often used as priors) include the intrinsic conditional autoregressive \citep[iCAR;][]{besag1974spatial,besag1975statistical}, the Besag--York--Molli\'e \citep[BYM;][]{besag1991bayesian} and the Leroux \citep{leroux2000estimation} models. These models often lack flexibility in capturing spatial dependence and are not invariant under a change of support. In particular, they are not consistent as the spatial discretization becomes finer, which yields model interpretation issues. To fix these problems, our flexible intrinsic \WM{} model can be incorporated as a latent process, linked to observed data only through spatial aggregates, similar to \citet{Gotway2002combining} in the case of proper Gaussian fields. This significantly broadens the class of available intrinsic areal models for a wide range of spatial applications.

\vspace{5pt}

\noindent\textbf{Avenues for simulation-based inference.}
Our SPDE-based modeling framework not only facilitates likelihood-based inference, but it also accelerates model simulation. Therefore, it has direct benefits for simulation-based inference methods and amortized learning approaches \citep[e.g.,][]{zam2025}, often used when the likelihood function is unavailable or computationally intractable. An increasingly popular approach is neural Bayes estimation \citep[e.g.,][]{Len23, sainsbury2024likelihood, richards2024neural, roe2025}, where an estimator is represented as a neural network that is trained from thousands or millions of parameter-data pairs simulated from the model. Using sparse models (such as the one we propose) in this context can thus significantly reduce the training time, making these estimators applicable in higher-dimensional settings without domain partitioning \citep{hector2024blackbox}.

\vspace{5pt}

\noindent\textbf{Avenues for inverse problems.} 
The SPDE approach is becoming more popular for defining priors in inverse problems \citep[e.g.,][]{alonso2022}. Because intrinsic fields often are used as priors in Bayesian models, a natural avenue for future research is to use the intrinsic Whittle--Mat\'ern fields as priors in Bayesian inverse problems. This would allow prior restrictions on the smoothness of the latent fields without restricting their mean levels. Interesting theoretical problems include deriving posterior contraction rates as in \citet{alonso2022b} or studying theoretical properties of Gaussian process regression as in \citet{teckentrup2020}.

\section*{Acknowledgments}
This research was supported by King Abdullah University of Science and Technology (KAUST) Office of Sponsored Research (OSR) under Award No.~OSR-CRG2020-4394.
Sebastian Engelke was also supported by the Swiss National Science Foundation under Grant 186858.

\begingroup
\singlespacing

\bibliographystyle{chicago}
\bibliography{Intrinsic_references}

\endgroup

\appendix

\section{Mathematical details on intrinsic GMRFs}\label{App:Intrinsic_intro}

Definition~\ref{def:FOI} can be interpreted in a limiting sense.
Similar to \citet[][Equation~3.19]{Rue05}, consider the centered (proper) GMRF $\mathbf{W}^{[\gamma]}$ with precision matrix 
$\Theta^{[\gamma]}:=\Theta + \gamma \mathbf{1}\mathbf{1}^\top$.
The matrix $\Theta$ has an orthonormal basis $\{\mathbf{e}_i\}_{1 \leq i \leq d}$ whose eigenvalues $\{\lambda_{i} \}_{1 \leq i \leq d}$ can be arranged in a non-decreasing order, so that $\lambda_1=0$ and $\mathbf{e}_1 =\mathbf{1}$. 
The spectral decomposition of $\Theta^{[\gamma]}$ is the same as that of $\Theta$ except with $\lambda_1=\gamma$.
Consequently, for any $\mathbf{w} \in \mathbb{R}^d$,
\begin{equation}\label{eqn:Gt0}
f(\mathbf{w}) = \lim_{\gamma \to 0} \frac{ f^{[\gamma]}(\mathbf{w})}{\sqrt{2 \pi \gamma}},
\end{equation}
where $f^{[\gamma]}(\cdot)$ is the density of $\mathbf{W}^{[\gamma]}$ and $f(\cdot)$ is the density in Definition~\ref{def:FOI}.
For a probabilistic interpretation of \eqref{eqn:Gt0}, note that if we let $\xi_i$ be independent $N(0,1)$ random variables, then
\begin{equation} \label{eqn:WgD}
\mathbf{W}^{[\gamma]} \stackrel{d}{=} \frac{\mathbf{1}}{\sqrt{\gamma}} \xi_1  + \frac{\mathbf{e}_2}{\sqrt{\lambda_2}} \xi_2 + \dots + \frac{\mathbf{e}_d}{\sqrt{\lambda_d}} \xi_d.
\end{equation}
As $\gamma \to 0$, the variance in the direction $\mathbf{e}_1 = \mathbf{1}$ increases to infinity. This is reflected in the fact that $f(\mathbf{w}) = f(\mathbf{w}+ c\mathbf{1})$ for all $c \in \mathbb{R}$, i.e., that $\mathbf{W}$ is first-order intrinsic.
However, as $\gamma \to 0$, the joint distribution of $( W^{[\gamma]}_i - W^{[\gamma]}_j)_{ij}$ does not change, because \eqref{eqn:WgD} implies that these differences are independent of $\gamma$  (since $\xi_1/\sqrt\gamma$ cancels when the difference is taken).
Consequently, $\mathbf{W}$ can be understood as a random variable with infinite variance in the direction $\mathbf{e}_1 = \mathbf{1}$, which is characterized by the joint distribution of its differences $(W_i - W_j)_{ij}$.

The first-order intrinsic random variable $\mathbf{W}$ can also be understood through conditioning. If $\mathbf{h} \in \mathbb{R}^d$ with $\mathbf{h}^\top \mathbf{1} \neq 0$, then 
\begin{equation}\label{eqn:Wh}
\mathbf{W}^{[\mathbf{h}]}:= (\mathbf{W} \mid\, \mathbf{h}^\top \mathbf{W}=0)\stackrel{d}{=} \mathbf{W}- \frac{\mathbf{h}^\top \mathbf{W}}{\mathbf{h}^\top \mathbf{1}} \mathbf{1},
\end{equation}
and $\mathbf{W}^{[\mathbf{h}]}$ is a proper random variable.
Similar to above, we can establish \eqref{eqn:Wh} by showing $\mathbf{W}^{[\mathbf{h}]} \stackrel{d}{=} \lim_{\gamma \to 0} (\mathbf{W}^{[\gamma]} \mid\mathbf{h}^\top \mathbf{W}^{[\gamma]} =0)$.
Equation \eqref{eqn:Wh}, implies that the variogram of $\mathbf{W}^{[\mathbf{h}]}$ does not depend on $\mathbf{h}$, that is,
\begin{equation} \label{eqn:varbh}
{\rm Var}(W^{[\mathbf{h}]}_i - W^{[\mathbf{h}]}_j) = {\rm Var}(W_i - W_j)=\Gamma_{ij}.
\end{equation}
Indeed, since $\mathbf{W}$ is characterized by the joint distribution of its differences $(W_i - W_j)_{ij}$ and \eqref{eqn:Wh} implies that the differences $(W_i-W_j)_{ij}$ remain unchanged after conditioning,
we can view $\mathbf{W}^{[\mathbf{h}]}$ as the restriction of $\mathbf{W}$ to the space $\{ \mathbf{w} \in \mathbb{R}^d : \mathbf{h}^\top \mathbf{w} =0 \}$.
There are two choices of $\mathbf{h}$ that are noteworthy:
\begin{enumerate}
\item[\textit{(a)}] If $\mathbf{h}=\mathbf{e}_1=\mathbf{1}$ then 
\begin{equation}\label{eqn:covbh}
{\rm Cov} (\mathbf{W}^{[\mathbf{h}]}) = \sum_{j=2}^d \lambda_j^{-1} \mathbf{e}_j \mathbf{e}_j^\top,
\end{equation}
where $(\lambda_i)_{1 \leq i\leq d}$ and $(\mathbf{e}_i)_{1 \leq i\leq d}$ are the eigenvalues and vectors of $\Theta$.
Any other choice of $\mathbf{h}$ will lead to a covariance matrix whose eigenvectors are not shared with $\Theta$. Note that this corresponds to $\lim_{\gamma \to \infty} {\rm Cov}( \mathbf{W}^{[\gamma]})$ for $\mathbf{W}^{[\gamma]}$ in \eqref{eqn:WgD}.

\item[\textit{(b)}] Let $\boldsymbol{\delta}_k$ be a vector with $k$-th entry $1$ and all other entries 0. If $\mathbf{h}= \boldsymbol{\delta}_k$ then 
\begin{equation}\label{eqn:prebh}
{\rm Prec}(\mathbf{W}^{[\mathbf{h}]}_{-k}) = \Theta^{(k)} := (\Theta_{ij})_{i,j \neq k}.
\end{equation}
Thus, the conditional dependence graph of $\mathbf{W}^{[\mathbf{\boldsymbol{\delta}}_k]}_{-k}$ is the same as that of $\mathbf{W}$, i.e., $(\mathcal{V}, \mathcal{E})$, except that node $k$ and its adjacent edges are deleted.
\end{enumerate}
    Point \textit{(a)} can be helpful in establishing theoretical properties of the first-order intrinsic GMRF because it enables us to work with a proper GMRF with the same variogram and and spectral decomposition as the intrinsic GMRF. Point \textit{(b)} can be helpful for inference as it allows us to work with a proper GMRF with the same variogram as the intrsic GMRF while maintaining sparsity of the precision matrix---note that the precision matrix that arises from the proper field in \textit{(a)} is generally dense. Indeed, if we let $f_d(\cdot; \Theta)$ denote the density of a $d$ dimensional GMRF with precision $\Theta$, then the following result ensures that likelihood inference with the proper GMRF in \textit{(b)} and the first-order intrinsic GMRF in Definition~\ref{def:FOI} is equivalent.
\begin{lemma}\label{lem:Eqv}
If $\Theta$ is a first-order intrinsic precision matrix then for any $\mv{w}\in\mathbb{R}^d$ and $k \in \{1, \dots, d\}$,
$f_{d}(\mathbf{w}; \Theta) = d^{1/2}f_{d-1}((\mathbf{w}-\mathbf{1}w_k)_{-k}; \Theta^{(k)})$.
\end{lemma}
\begin{proof}
Without loss of generality assume $k=1$. The result follows from the fact that $|\Theta|^*=d|\Theta^{(1)}|$ (see \cite{rottger2023total}) and 
\begin{align*}
\mv{w}^{\top} \Theta \mv{w} &= (\mv{w}-w_1 \mathbf{1})^\top \Theta (\mv{w}-w_1 \mathbf{1}) \\
&=(\mv{w}-w_1 \mathbf{1})^\top \begin{pmatrix} \sum_{i,j} \Theta^{(1)}_{ij} & -(\sum_{i} \Theta^{(1)}_{ij})_j^\top \\ -(\sum_{i} \Theta^{(1)}_{ij})_j & \Theta^{(1)} \end{pmatrix} (\mv{w}-w_1 \mathbf{1})  \\
&=(\mv{w}-w_1 \mathbf{1})^\top_{-1} \Theta^{(1)} (\mv{w}-w_1 \mathbf{1})_{-1}.
\end{align*}
\end{proof}

\section{Details on the intrinsic Whittle--Mat\'ern fields}\label{app:fractional}
The operator $\widetilde{L} = - \widetilde{\Delta}$, is clearly densely defined, positive definite, and has a compact inverse on $H$. Because, $e_1 \propto 1$, we have that $\int e_i(s) {\rm d}s = 0$ for all $i>1$, since $e_i \indep e_1$. 
Clearly, $\{e_j\}_{j=2}^{\infty}$ and $\{\lambda_j\}_{j=2}^{\infty}$ are the eigenvectors and eigenvalues of $\widetilde{L}$, and  $\{e_j\}_{j=2}^{\infty}$  form an orthonormal basis of $H$. Due to the Weyl asymptotics of the Neumann Laplacian, there exist $c,C>0$ such that for every $j\in\mathbb{N}$,
$c j^{2/d} \leq \lambda_j \leq C j^{2/d}.$
For $\beta>0$, we define the fractional power of $\widetilde{L}$ in the spectral sense. That is, let 
$$
\mathcal{D}(\widetilde{L}^\beta) = \{\phi\in H : \sum_{j>1} \lambda_j^{2\beta} (\phi, e_j)_H^2 < \infty\}
$$
denote the domain of the fractional operator $\widetilde{L}^{\beta} : \mathcal{D}(\widetilde{L}^{\beta} ) \to H$, whose action is  
$$
\widetilde{L}^{\beta}  \phi = \sum_{j>1}  \lambda_j^{\beta} (\phi, e_j)_H e_j.
$$
We introduce $\dot{H}_{\widetilde{L}}^{2\beta} = \cD(\widetilde{L}^\beta) \subset H$, which is a Hilbert space with respect to the inner product 
$(\phi,\psi)_{2\beta} = (\widetilde{L}^\beta \phi, \widetilde{L}^\beta \psi)_H$ and induced norm $\|\phi\|_{2\beta}^2 = \|\widetilde{L}^\beta \phi\|_H^2$. We let $\dot{H}_{\widetilde{L}}^{-2\beta}$ denote the dual space of $\dot{H}_{\widetilde{L}}^{2\beta}$ and note that $\dot{H}_{\widetilde{L}}^0 \cong H$.

Now, for $\kappa\geq 0$, define the operator $\hat{L} : \cD(\hat{L})\subset H \rightarrow H$ as
$\hat{L} = (-\widetilde{\Delta})^{\alpha/\beta}(\kappa^2 - \widetilde{\Delta})$. That is, 
$$
\hat{L}\phi = \sum_{j>1}\lambda_{j}^{\alpha/\beta}(\kappa^2 + \lambda_j) (\phi, e_j) e_j,
$$
for $\phi\in \cD(\hat{L})$.
Thus, this operator has eigenvalues $\{\hat{\lambda}_j\}_{j\in \bbN}$, with $\hat{\lambda}_j = \lambda_{j+1}^{\alpha/\beta}(\kappa^2 + \lambda_{j+1})$ and eigenvalues $\{\hat{e}_j\}_{j\in\bbN}$ with $\hat{e}_j = e_{j+1}$, and 
$$
c j^{\eta} \leq \hat{\lambda}_j \leq C j^{\eta}, \quad \eta = \frac{2}{d}(\frac{\alpha}{\beta} + 1).
$$
We define the fractional power of $\hat{L}$ in the spectral sense again, and introduce the corresponding space $\dot{H}_{\hat{L}}^{2\beta} = \cD(\hat{L}^\beta) \subset H$ as for $\widetilde{L}$.
Because $\hat{L}^{\beta/2} = (-\widetilde{\Delta})^{\beta/2}(\kappa^2 - \widetilde{\Delta})^{\alpha/2}$, the fractional equation \eqref{eq:fractional_intrinsic} is formally defined as 
\begin{equation}\label{eq:fractional_intrinsic1}
\hat{L}^{\beta/2} u = \widetilde{\mathcal{W}},
\end{equation}

\begin{proof}[Proof of Proposition~\ref{prop:uniqueness}]
For any $\alpha\in\bbR$ and $\beta>0$ such that $\alpha+\beta>0$, $\hat{L}$ is densely defined, positive definite and has a compact inverse on $H$. This follows directly from the fact that $-\widetilde{\Delta}$ has these properties on $H$. Therefore, 
by Lemma 2.1 of \cite{BKK2020}, the operator $\hat{L}^{\beta/2}$ can be extended to an insometric isomorphism $\hat{L}^{\beta/2} : \dot{H}_{\hat{L}}^s \rightarrow \dot{H}_{\hat{L}}^{s-\beta}$ for any $s\in\bbR$. Furthermore, by Proposition 2.3 of \cite{BKK2020}, $\widetilde{\mathcal{W}}\in L_2(\Omega, \dot{H}_{\hat{L}}^{-\frac1{\eta}-\epsilon})$ for any $\epsilon>0$. Thus, \eqref{eq:fractional_intrinsic1} has a unique solution $u\in  L_2(\Omega, \dot{H}_{\hat{L}}^{-\frac1\eta-\eps +\beta})$, and $u\in L_2(\Omega, H)$ if 
$
-\frac1{\eta} + \beta > 0 \Leftrightarrow \alpha + \beta > \frac{d}{2}.
$
\end{proof}

Since the operators $(-\widetilde{\Delta})$ and $(\kappa^2 - \widetilde{\Delta})$ commute, we may rewrite \eqref{eq:fractional_intrinsic} as
\begin{align*}
	(-\widetilde{\Delta})^{\beta/2} u &= \widetilde{x} \\
	(\kappa^2-\widetilde{\Delta})^{\alpha/2} \widetilde{x} &= \widetilde{\mathcal{W}},
\end{align*}
which will be of importance to the finite element analysis in Appendix~\ref{app:fem}. 

We may also rewrite the equation in therms of a standard Whittle--Mat\'ern field. To see this, 
let $L : \cD(L) \subset L_2(\cD) \rightarrow L_2(\cD)$ be defined as $L = \kappa^2 - \Delta$, where $\Delta$ is the Neumann Laplacian, and define fractional powers of $L$ in the spectral sense as above. Let $\dot{H}_L^{2\beta} = \cD(L^\beta)$ be the domain of the operator $L^\beta$ and introduce 
$\widetilde{\Pi} : L_2(D) \rightarrow H$ as the orthogonal projection of a vector in $L_2(D)$ onto $H$. 
Then, for any $\tau\in\mathbb{R}$, and for any $g\in \dot{H}_L^{2\tau}$, we have that 
$\widetilde{\Pi}L^{\tau}g = (\kappa^2-\widetilde{\Delta})^{\tau}\widetilde{\Pi}g$, i.e., the two operators commute. Indeed,
\begin{align*}
	\widetilde{\Pi}(\kappa^2 - \Delta)^\tau g &= \widetilde{\Pi} \sum_{i=1}^{\infty} (\kappa^2 + \lambda_j^2)^\tau (g, e_j)_{L_2(\cD)} e_j = \sum_{i=2}^{\infty} (\kappa^2 + \lambda_j^2)^\tau (g, e_j)_{L_2(\cD)} e_j \\
	&= 
	\sum_{i=1}^{\infty} (\kappa^2 + \lambda_j^2)^\tau (g, e_j)_{L_2(\cD)} \widetilde{\Pi}e_j = 
	(\kappa^2 - \Delta)^\tau \widetilde{\Pi} g.
\end{align*}

It is also clear that $\widetilde{\mathcal{W}}$ can be represented as $\widetilde{\Pi}\mathcal{W}$, which means that  we may also write  \eqref{eq:fractional_intrinsic} as 
\begin{align*}
	\widetilde{L}^{\beta/2} u &= \widetilde{\Pi}x \\
	L^{\alpha/2} x &= \mathcal{W},
\end{align*}
where $x$ is a standard Whittle--Mat\'ern field and $\mathcal{W}$ is Gaussian white noise on $L_2(\cD)$.

Thus, on $H$, we have that 
$(\kappa^2 - \Delta)^\tau = (\kappa^2 - \widetilde{\Delta})^\tau$, and since $\widetilde{\Pi}$ and $(\kappa^2 - \Delta)^\tau$ commute, we may also write the covariance operator of $u$ as $\widetilde{L}^{-\beta}L^{-\alpha}$. 
Finally, note that both $L$ and $\widetilde{L}$ are $H^2(\cD)$-regular, which means that if $f\in H$, then the solutions to $\widetilde{L} u = f$ and $L u = f$ are in $H_{\cN}^2(\cD)\cap H$.

\section{Variogram}\label{App:Variogram}

\begin{proof}[Proof of Proposition \ref{lem:Vard2}]
Let us start with the case $d=1$. In this case, the domain is $[-L/2, L/2]$ and the operator has the same eigenvalues as that on $[0,L]$ but eigenfunctions $\{\cos(j\pi s/L + j\pi/2)\}_{j=1}^\infty$. The expression for $\gamma_L$ is thus
$$
\gamma_L(s,t) = \frac1{\tau^{2}}\sum_{j=1}^{\infty} \frac1{\left(\frac{j\pi}{L}\right)^{2\beta}\left(\kappa^2 + \left(\frac{j\pi}{L}\right)^{2}\right)^{\alpha}}\left(\cos\left(\frac{j\pi s}{L} + \frac{j\pi}{2}\right) - \cos\left(\frac{j\pi t}{L} + \frac{j\pi}{2}\right)\right)^2.
$$
Introduce $w_j = j\pi/L$ and $\Delta_w = w_{j+1} - w_j = \pi/L$. Then 
$$
\gamma_L(s,t) = \frac{2}{\tau^{2}\pi} \Delta_w \sum_{j=1}^{\infty} w_j^{-2\beta}(\kappa^2 + w_j^2)^{-\alpha}(\cos(w_j s + j\pi/2) - \cos(w_j t + j\pi/2))^2.
$$
By elementary trigonometric identities, the square of the difference of the two cosines in this expression equals $1 - \cos(w_j(s-t)) + f_j(s,t)$, with 
\begin{align*}
f_j(s,t) = \frac{\cos(w_j s + j\pi) + \cos(w_j t + j\pi)}{2} - \cos(w_j(s+t) + j\pi).
\end{align*}
Therefore, $\gamma_L(s,t) = S_1(L) + S_2(L)$, where
$$
S_1(L) = \frac{2}{\tau^{2}\pi} \Delta_w \sum_{j=1}^{\infty} w_j^{-2\beta}(\kappa^2 + w_j^2)^{-\alpha}(1 - \cos(w_j (s - t))).
$$
This is a Riemann sum which for fixed $s$ and $t$ converges to 
$$
\frac{2}{\tau^{2}\pi} \int_0^{\infty} w^{-2\beta}(\kappa^2 + w^2)^{-\alpha}(1 - \cos(w (s - t))) dw
$$
as $L \rightarrow \infty$. Now, 
\begin{align*}
S_2(L) &= \frac{2}{\tau^{2}\pi} \Delta_w \sum_{j=1}^{\infty} w_j^{-2\beta}(\kappa^2 + w_j^2)^{-\alpha}f_j(s,t) 
= \frac{2}{\tau} \sum_{j=1}^\infty \frac{ L^{2(\alpha + \beta) - 1}}{(j\pi)^{2\beta}((\kappa L)^2 + (j\pi)^2)^{\alpha}}f_j(s,t).
\end{align*}
As $L\rightarrow\infty$, we have that $2 L^{2(\alpha + \beta) - 1} \rightarrow 0$ because $\alpha + \beta > 1/2$, further, 
$
(j\pi)^{-2\beta}((\kappa L)^2 + (j\pi)^2)^{-\alpha} \rightarrow 0
$
for any $j\in\mathbb{N}$
and 
$$
f_j(s,t) \rightarrow 
\frac{\cos(j\pi) + \cos(j\pi)}{2} - \cos(j\pi) = 0. 
$$
That is, $S_2(L) \rightarrow 0$ as $L \rightarrow \infty$ and we thus have showed that for fixed $s$ and $t$, 
\begin{align*}
\gamma_L(s,t) &\rightarrow \gamma(s,t) := \frac{2}{\tau^{2}\pi} \int_0^{\infty} w^{-2\beta}(\kappa^2 + w^2)^{-\alpha}(1 - \cos(w (s - t))) dw \\ 
&= 
\frac{2}{\tau^{2}(2\pi)} \int_{-\infty}^{\infty} w^{-2\beta}(\kappa^2 + w^2)^{-\alpha}(1 - \cos(w (s - t))) dw.
\end{align*}
Similar arguments show that for any $d\in\mathbb{N}$, we have that $\gamma_L(\mv{s},\mv{t}) \rightarrow  \gamma(\|\mv{s}-\mv{t}\|)$ with 
\begin{equation}\label{eq:vario_proof}
\gamma(\mv{h}) := \frac{2}{\tau^2(2\pi)^d} \int (1 - \cos(\mv{\omega}^\top \mv{h})) \frac{1}{\lVert \mv{\omega} \rVert^{2\beta}(\kappa^2 + \lVert \mv{\omega}\rVert^2)^{\alpha}} {\rm d}\mv{\omega}.
\end{equation}
For $d=1$, this directly proves the claim. 
For $d\geq 2$, we note that the variogram is isotropic, and it is enough to calculate the integral for $\mv{h} = (h,0,\ldots,0)$. Further, 
recall that a modified Bessel function of order $m=0, 1/2, 1, \ldots$ is defined as 
$$
J_m(z) = \frac{(z/2)^m}{\sqrt{\pi}\Gamma(m + 1/2)}\int_0^{\pi} \cos(z\cos(\theta)) \sin^{2m}(\theta) {\rm d}\theta.
$$
For $d=2$, switching to polar coordinates, $\mv{w} = r(\cos(\theta), \sin(\theta))$, with $\theta\in[-\pi,\pi]$, 
\begin{align*}
\gamma (\mv{h}) 
&= \frac{2}{\tau^2(2\pi)^2} \int_0^{\infty}\int_{-\pi}^\pi \frac{1 - \cos(h r \cos(\theta))}{r^{2\beta}(\kappa^2 + r^{2})^{\alpha}}  r {\rm d}\theta {\rm dr} \\
&= \frac{4}{\tau^{2}(2\pi)^2} \int_0^{\infty}\int_{0}^\pi 1 - \cos(h r \cos(\theta)) {\rm d}\theta \frac{r}{r^{2\beta}(\kappa^2 + r^{2})^{\alpha}} {\rm dr}\\ 
&= \frac{1}{\tau^{2}\pi} \int_0^{\infty} (1-J_0(hr)) \frac{r}{r^{2\beta}(\kappa^2 + r^{2})^{\alpha}} {\rm dr}.
\end{align*}

To calculate the variogram for $d\geq 3$, we use polar coordinates $\mv{w} = r (l_1,\ldots, l_d)$, with $r = \|\mv{w}\|$ and $\sum_{i=1}^d l_i^2 = 1$ and now let $l_1 = \cos\theta_{d-1}$ with  $0 \leq \theta_{d-1} \leq \pi$.
Again calculating the integral for $\mv{h} = (h,0,\ldots,0)$ and switching to polar coordinates yields
\begin{align*}
\gamma (\mv{h}) &= 2\int (1 - \cos(h r \cos(\theta_{d-1}))) \frac{f(\mv{\omega})}{\lVert \mv{\omega} \rVert^{2\beta}} {\rm d}\mv{\omega} 
= \frac{2}{\tau^{2}(2\pi)^d} \int \frac{1 - \cos(h r \cos(\theta_{d-1}))}{\lVert \mv{\omega} \rVert^{2\beta}(\kappa^2 + \lVert \mv{\omega} \rVert^{2})^{\alpha}}  {\rm d}\mv{\omega} \\
&= \frac{2}{\tau^{2}(2\pi)^d} \int_0^\infty \int (1 - \cos(h r \cos(\theta_{d-1}))) {\rm d}\sigma(\mv{\theta}) \frac{r^{d-1}}{r^{2\beta}(\kappa^2 + r^{2})^{\alpha}}  {\rm d}r ,
\end{align*}
where $\sigma(\mv{\theta})$ is the surface measure of the unit sphere in dimension $d$. For each $\theta_{d-1}$, $(l_2,\ldots, l_d)$ is a location on a sphere with radius $\sin \theta_{d-1}$. Integrating the constant function $1 - e^{- i h r \cos(\theta_{d-1})}$ over this sphere, which has area $2\pi^{(d-1)/2}\sin^{d-2}(\theta_{d-1})/\Gamma((d-1)/2)$ yields 
\begin{align*}
\gamma (\mv{h}) &= 
 \frac{4\pi^{(d-1)/2}}{\tau^{2}(2\pi)^d\Gamma((d-1)/2)}  \int_0^\infty \int_0^\pi (1 - \cos(h r \cos(\theta))) \sin^{d-2}(\theta) {\rm d}\theta \frac{r^{d-1}}{r^{2\beta}(\kappa^2 + r^{2})^{\alpha}}  {\rm d}r .
\end{align*}
Now, using that
$$
\int_0^\pi \sin^{d-2}(\theta) {\rm d}\theta = \sqrt{\pi}\frac{\Gamma\left(\frac{d-1}{2}\right)}{\Gamma\left(\frac{d}{2}\right)},
$$
and the expression for the Bessel function, we have 
\begin{align*}
\int_0^\pi (1 - \cos(h r \cos(\theta))) \sin^{d-2}(\theta) {\rm d}\theta &= 
\sqrt{\pi}\frac{\Gamma\left(\frac{d-1}{2}\right)}{\Gamma\left(\frac{d}{2}\right)} - 
\frac{\Gamma(\frac{d-2}{2} + \frac1{2})\sqrt{\pi}}{(hr/2)^{(d-2)/2}} J_{(d-2)/2}(hr)\\
&= \frac{\sqrt{\pi}\Gamma\left(\frac{d-1}{2}\right) }{\Gamma(d/2)}\left(1 - \Lambda_d(hr)\right).
\end{align*}
Inserting this expression into the expression for $\gamma(\mv{h})$ yields
\begin{align*}
\gamma (\mv{h}) &= 
 \frac{1}{\tau^{2}2^{d-2}\pi^{d/2}\Gamma(d/2)}  \int_0^\infty (1 - \Lambda_d(h r))  \frac{r^{d-1}}{r^{2\beta}(\kappa^2 + r^{2})^{\alpha}}  {\rm d}r .
\end{align*}
For the case $d=2$, $\Lambda_2(z) = J_0(z)$, and this expression equals that we derived above.
\end{proof}

\begin{proof}[Proof of Proposition \ref{prop:LG}]
We establish the result for $d=2$. When $d\neq 2$ the result follows from similar, albeit more involved arguments. First, observe that
the leading order term of $1-\Lambda_2(z)\equiv1-J_0(z)$ when $z$ is small is $z^2$, and there exists $c_1, C_1 >0$ such that 
$c_1z^2 \leq 1 - \Lambda_2(z) \leq C_1z^2$, for all  $z \in [0,1)$.
Additionally, there exists $c_2, C_2>0$ such that 
$c_2 \leq 1-\Lambda_2(z) \leq C_2$, for all  $z \geq 1$,
which follows from the observation that the unique maximum of $|J_\nu(z)|$ over $(\nu,z) \in [0,\infty)^2$ is $J_0(0)=1$ (note we take $z \geq 1$).

We first let $\lVert h \rVert \to 0$.
Applying both these inequalities in the second step we have
\begin{align*}
\gamma(h) &= \frac{1}{2\pi}\int^{\infty}_0 \frac{1-J_0(r \lVert h\rVert)}{r^{2 \beta -1} (\kappa^2 + r^2)^\alpha} {\rm d}r \\
&\leq\frac{C_1}{2\pi}\int^{\lVert h \rVert^{-1}}_0 \frac{ r^2 \lVert h \rVert^2}{r^{2 \beta -1} (\kappa^2 + r^2)^\alpha} {\rm d}r + \frac{C_2}{2\pi}\int^\infty_{\lVert h \rVert^{-1}} \frac{1}{r^{2 \beta -1} (\kappa^2 + r^2)^\alpha} {\rm d}r \\
&\leq \frac{C_1\lVert h \rVert^2}{2 \pi \kappa^\alpha} \int_0^1 r^{3-2\beta} {\rm d}r +\frac{C_1  \lVert h \rVert^2}{2\pi}\int^{\lVert h \rVert^{-1}}_1r^{3-2 \beta -2\alpha } {\rm d}r+ \frac{C_2}{2\pi}\int^\infty_{\lVert h \rVert^{-1}} {r^{1-2 \beta -2\alpha }} {\rm d}r \\
&\leq C^* \left[ \lVert h \rVert^{2} + \lVert h\rVert^{2 \alpha + 2\beta -2} \right],
\end{align*}
where $C^*$ is a suitably large constant. The upper bound then follows by observing that the first term dominates when $\alpha+\beta>2$ and the second term dominates when $\alpha+\beta \leq 2$. The lower bound follows by similar arguments leading to the result.

Now consider the case where $\lVert h \rVert \to \infty$. Following a similar line of reasoning as above, we obtain
\begin{align*}
\gamma(h) &\leq\frac{C_1}{2\pi}\int^{\lVert h \rVert^{-1}}_0 \frac{ r^2 \lVert h \rVert^2}{r^{2 \beta -1} (\kappa^2 + r^2)^\alpha} {\rm d}r +\frac{C_2}{2\pi}\int^{\infty}_{\lVert h \rVert^{-1}} \frac{ 1}{r^{2 \beta -1} (\kappa^2 + r^2)^\alpha} {\rm d}r \\
&\leq \frac{C_1\lVert h \rVert^2}{2 \pi \kappa^\alpha} \int_0^{\lVert h \rVert^{-1}} r^{3-2\beta} {\rm d}r +\frac{C_2}{2\pi (\kappa^2+1)^\alpha}\int^1_{\lVert h \rVert^{-1}}r^{1-2 \beta } {\rm d}r+ \frac{C_2}{2\pi}\int^\infty_{1} {r^{1-2 \beta -2\alpha }} {\rm d}r \\
&\leq C^*\left[ \lVert h\rVert^{2\beta-2} + \lVert h \rVert^{2\beta-2} + 1 \right],
\end{align*}
when $\beta\neq 1$, for $C^*$ sufficiently large. When $\beta=1$, the second $\lVert h \rVert^{2\beta-2}$ term in the above equation becomes $\log \lVert h \rVert$. Considering the leading order terms in each case then leads to the upper bound. The lower bound again follows from similar arguments leading to the result.
\end{proof}

\section{Finite element details}\label{app:fem}
In this section, we provide the details for the FEM and rational approximation method. We begin with an introduction to the method and then provide the proofs.

\subsection{Introduction to the method}
Let $V_{\meshwidth}$ be a finite element space that is spanned by the set of continuous piecewise linear basis functions $\{\varphi_j\}_{j=1}^{n_{\meshwidth}}$, with $n_{\meshwidth}\in\mathbb{N}$, defined with respect to the triangulation $\mathcal{T}_{\meshwidth}$ of $\overline{\mathcal{D}}$. Further, define $\widetilde{V}_{\meshwidth} = V_{\meshwidth} \cap H$, which consists of all functions
$$
f(s) = \sum_{i=1}^{N_{\meshwidth}} c_i \varphi_i(s), \quad \mbox{with} \quad \sum_{i=1}^{N_{\meshwidth}} h_i c_i = 0,
$$
where $h_i = \int \phi_j(s) {\rm d}s$. 
We then have that $V_{\meshwidth} \subset V := \dot{H}_L^1$ and $\widetilde{V}_{\meshwidth} \subset \widetilde{V} := \dot{H}_{\widetilde{L}}^1$. 

Let $\Pi_{\meshwidth} : L_2(\cD) \rightarrow V_{\meshwidth}$ and $\widetilde{\Pi}_{\meshwidth} : L_2(\cD) \rightarrow \widetilde{V}_{\meshwidth}$
denote the $L_2(\cD)$-orthogonal projections onto $V_{\meshwidth}$ and $\widetilde{V}_{\meshwidth}$ respectively. We then have the following result which will be of importance later. In the statement, and throughout the paper, we use the notation $\|A\|_{\cL(E)}$ for the operator norm of an operator $A : E \rightarrow E$, for some  Hilbert space $E$. We will later also use the notation $\|B\|_{\cL(E,F)}$ for the operator norm of an operator $B : E \rightarrow F$, where $F$ is some other Hilbert space.

\begin{lemma}\label{lemma:projection}
The orthogonal projections $\Pi_{\meshwidth}$ and $\widetilde{\Pi}_{\meshwidth}$ are $H^1(\cD)$-stable, i.e., there exists a constant $C$ such that $\|\widetilde{\Pi}_{\meshwidth}\|_{\cL(H^1(\cD))} \leq C$ and $\|\Pi_{\meshwidth}\|_{\cL(H^1(\cD))} \leq C$. Further, we have that $\widetilde{\Pi}_{\meshwidth} = \widetilde{\Pi}\Pi_{\meshwidth} = \Pi_{\meshwidth}\widetilde{\Pi}$.
\end{lemma}

The operators $L$ and $-\widetilde{\Delta}$ induce the following continuous and coercive bilinear forms on $V$ and $\widetilde{V}$ respectively:
\begin{align}\label{bilinear_form}
	a(v, u) &= (\nabla u, \nabla v )_{L_2(\mathcal{D})} + \kappa^2 (u, v )_{L_2(\mathcal{D})},  &&u,v \in V, \\
	\widetilde{a}(v, u) &= (\nabla u, \nabla v )_{L_2(\mathcal{D})}, &&u,v \in \widetilde{V}.
\end{align}
Let $L_{\meshwidth}: V_{\meshwidth}\to V_{\meshwidth}$ and $\widetilde{L}_{\meshwidth} : \widetilde{V}_{\meshwidth} \to \widetilde{V}_{\meshwidth}$ be defined in terms of the bilinear forms $a$ and $\widetilde{a}$, as their restrictions to $V_{\meshwidth}\times V_{\meshwidth}$ and $\widetilde{V}_{\meshwidth}\times \widetilde{V}_{\meshwidth}$, respectively:
\begin{align*}
	(L_{\meshwidth} \phi_{\meshwidth}, \psi_{\meshwidth})_{L_2(\mathcal{D})} &= a(\phi_{\meshwidth}, \psi_{\meshwidth}) ,  &&\phi_{\meshwidth}, \psi_{\meshwidth} \in V_{\meshwidth}.\\
	(\widetilde{L}_{\meshwidth} \phi_{\meshwidth}, \psi_{\meshwidth})_{L_2(\mathcal{D})} &= \widetilde{a}(\phi_{\meshwidth}, \psi_{\meshwidth}) , &&\phi_{\meshwidth}, \psi_{\meshwidth} \in \widetilde{V}_{\meshwidth}.
\end{align*}

We now refer to the following SPDE on $\widetilde{V}_{\meshwidth}$ as the discrete model of \eqref{eq:fractional_intrinsic}:
\begin{align*}
	\widetilde{L}_{\meshwidth}^{\beta/2} u_{\meshwidth} &= \widetilde{\Pi}_{\meshwidth}  x_{\meshwidth} \\
	L_{\meshwidth}^{\alpha/2} x_{\meshwidth} &= \mathcal{W}_{\meshwidth},
\end{align*}
where $\mathcal{W}_{\meshwidth} = \sum_{j=1}^{\infty} \xi_j \Pi_{\meshwidth} e_{j}$ is a $V_{\meshwidth}$ valued approximation of $\mathcal{W}$. 

The covariance operator of $u_{\meshwidth}$ is given by 
$\widetilde{L}_{\meshwidth}^{-\alpha}L_{\meshwidth}^{-\beta}$, and
$$
\varrho_{\meshwidth}^{\alpha,\beta}(x,y) = \sum_{j=1}^{n_{\meshwidth}} \widetilde{\lambda}_{j,\meshwidth}^{-\beta} \lambda_{j,\meshwidth}^{-\alpha}e_{j,\meshwidth}(x) e_{j,\meshwidth}(y),\quad\hbox{for a.e. $(x,y)\in\mathcal{D}\times \mathcal{D}$},
$$ 
is the corresponding covariance function. Here, $\{e_{j,\meshwidth}\}_{j=1}^{n_{\meshwidth}}$ are the eigenvectors of $L_{\meshwidth}$, which are orthonormal in $L_2(\mathcal{D})$ and form a basis of $V_{\meshwidth}$. Order these so that the corresponding eigenvalues $0 < \lambda_{1,\meshwidth} \leq \lambda_{2,\meshwidth} \leq \cdots\leq \lambda_{n_{\meshwidth}, \meshwidth},$ are in a non-decreasing order. Then $e_{1,\meshwidth}\propto 1$ and $\{\widetilde{e}_{j,\meshwidth}\}_{j=1}^{n_{\meshwidth}-1}$, with $\widetilde{e}_{j,\meshwidth} = e_{j+1,\meshwidth}$ are therefore the eigenvectors of $\widetilde{L}_{\meshwidth}$, with eigenvalues $\widetilde{\lambda}_j = \lambda_{j+1}-\kappa^2$. 

We now have the following result regarding the FEM error of the approximation. 
\begin{proposition}\label{cov_fem_approx_rate}
	Let $\alpha,\beta\geq 0$  be such that $\alpha + \beta > d/4$. Then, for each $\varepsilon>0$, 
	\begin{equation}\label{eq:cov_bound}
		\|\varrho^{\alpha,\beta} - \varrho_{\meshwidth}^{\alpha,\beta}\|_{L_2(\mathcal{D}\times\mathcal{D})} \lesssim_{\varepsilon,\alpha, \beta, \kappa,\mathcal{D}} h^{\min\{2\alpha + 2\beta-d/2 -\varepsilon,2\}}.
	\end{equation}
\end{proposition}

To handle the fractional powers of the operators, we can either perform rational approximations of the differential operators or use a rational approximation of the covariance operator of $u_{\meshwidth}$. The latter is better if we only care about the approximation of the covariance function or variogram, whereas the former can be used if we also care about convergence of the process itself. 

We here focus on the second option and approximate the covariance operator of the random field. To that end, 
we first split $L_{\meshwidth}^{-\alpha} = L_{\meshwidth}^{-\{\alpha\}} L_{\meshwidth}^{-\lfloor\alpha\rfloor}$, where $\{x\}=x-\lfloor x\rfloor$ is the fractional part of $x$. Then, we approximate $L_{\meshwidth}^{-\{\alpha\}}$ with a rational approximation of order $m$. This yields an approximation
	$L_{\meshwidth}^{-\alpha} \approx	L_{\meshwidth,m}^{\alpha} := L_{\meshwidth}^{-\lfloor \alpha\rfloor} p(L_{\meshwidth}^{-1})q(L_{\meshwidth}^{-1})^{-1}$.
Here, $p(L_{\meshwidth}^{-1}) = \sum_{i=0}^{m} a_i L_{\meshwidth}^{m-i}$ and $q(L_{\meshwidth}^{-1}) = \sum_{j=0}^m b_j L_{\meshwidth}^{m-j}$ are polynomials obtained from a rational approximation of order $m$ of the real-valued function $f(x) = x^{\{\alpha\}}$, i.e.,
$$
x^{\{\alpha\}} \approx \frac{\sum_{i=0}^{m} a_i x^{i}}{\sum_{i=0}^m b_i x^{i}}.
$$
Specifically, to obtain $\{a_i\}_{i = 0}^{m}$ and  $\{b_i\}_{i = 0}^{m}$, we approximate the function $f(x) = x^{\{\alpha\}}$ on the interval $[\lambda_{n_{\meshwidth}, \meshwidth}^{-1},\lambda_{1,\meshwidth}^{-1}]$, which covers the spectrum of $L^{-1}_{\meshwidth}$. The coefficients are computed as the best rational approximation in the $L_\infty$-norm, which, for example, can be obtained via the second Remez algorithm \citep{Remez1934}. 

Next, we perform the same type of approximation for $\widetilde{L}_{\meshwidth}^{-\beta}$, where we first split $\widetilde{L}_{\meshwidth}^{-\beta} = \widetilde{L}_{\meshwidth}^{-\{\beta\}} \widetilde{L}_{\meshwidth}^{-\lfloor\beta\rfloor}$, and then approximate $\widetilde{L}_{\meshwidth}^{-\{\beta\}}$ with a rational approximation of order $\widetilde{m}$: 
\begin{equation} \label{cov_operator_appro3}
	\widetilde{L}_{\meshwidth}^{-\beta} \approx	\widetilde{L}_{\meshwidth,\widetilde{m}}^{\beta} := \widetilde{L}_{\meshwidth}^{-\lfloor \beta\rfloor} \widetilde{p}(\widetilde{L}_{\meshwidth}^{-1})\widetilde{q}(\widetilde{L}_{\meshwidth}^{-1})^{-1}.
\end{equation}
The combined approximation is thus given by 
$\widetilde{L}_{\meshwidth,\widetilde{m}}^{\beta}L_{\meshwidth,m}^{\alpha}$.

The covariance function corresponding to this covariance operator is 
$$
\varrho_{\meshwidth,m,\widetilde{m}}^{\alpha,\beta}(x,y) = \sum_{j=1}^{n_{\meshwidth}} \frac{\widetilde{\lambda}_{j,\meshwidth}^{-\lfloor \beta \rfloor}\widetilde{p}(\widetilde{\lambda}_{j,\meshwidth}^{-1})}{\widetilde{q}(\widetilde{\lambda}_{j,\meshwidth}^{-1})}  \frac{\lambda_{j,\meshwidth}^{-\lfloor \alpha \rfloor}p(\lambda_{j,\meshwidth}^{-1})}{q(\lambda_{j,\meshwidth}^{-1})}  e_{j,\meshwidth}(x) e_{j,\meshwidth}(y),\quad\hbox{for a.e. $(x,y)\in\mathcal{D}$}.
$$
There are three  sources of errors when we consider $\varrho_{\meshwidth,m,\widetilde{m}}^{\alpha, \beta}$ as an approximation of the true covariance function $\varrho$: the FEM approximation and the two rational approximations. We now show that we have precise control over these errors via the FEM mesh width and the two orders $m$ and $\widetilde{m}$ if the rational approximations. 

\begin{theorem}\label{thm:coverror}
	Let $\alpha\geq 0$ and $\beta\geq 0$ be such that $\alpha + \beta > d/2$. Then, for every $\varepsilon>0$ and for sufficiently small $h$, we have:
	\begin{equation} \label{eq:cov_error_bound}
		\begin{split}
			\|\varrho^{\alpha,\beta} - \varrho_{\meshwidth,m,\widetilde{m}}^{\alpha,\beta}\|_{L_2(\cD \times \cD)}  \lesssim_{\varepsilon,\beta, \alpha \kappa,\mathcal{D}} & h^{\min\{2\alpha + 2\beta-d/2 -\varepsilon,2\}} \\
			&+ h^{-d/2}(1_{\beta \notin \mathbb{N}}  e^{-2\pi\sqrt{\{\beta\} \widetilde{m}}} + 1_{\alpha \notin \mathbb{N}} e^{-2\pi\sqrt{\{\alpha\} m}}).
			\end{split}
	\end{equation}
\end{theorem}

\subsection{Practical implementation}
Let $r_{m,\alpha}$ and $r_{\widetilde{m},\beta}$ denote the rational functions used to approximate $L_{\meshwidth}^\alpha$ and $\widetilde{L}_{\meshwidth}^\beta$ respectively. In practice, we normalize 
By $L_{\meshwidth}$ and $\widetilde{L}_{\meshwidth}$ so that the rational approximations can be performed on $[0,1]$. Then, by \citet[][Proposition 3]{bolin2025linear}, 
\begin{align}\label{rational_app}
r_{m,\alpha}(x) = k + \sum_{i=1}^m \frac{c_i}{x-p_i}, \quad \widetilde{r}_{\widetilde{m},\beta}(x) = \widetilde{k} + \sum_{i=1}^{\widetilde{m}} \frac{\widetilde{c}_i}{x-\widetilde{p}_i} 
\end{align}
where $k,\widetilde{k}, c_i, \widetilde{c}_i>0$ and $p_i, \widetilde{p}_i<0$.
This results in the approximation 
\begin{align*}
\widetilde{L}_\meshwidth^{-\beta}L_\meshwidth^{-\alpha} &\approx 
\widetilde{L}_\meshwidth^{-\lfloor \beta \rfloor}L_\meshwidth^{-\lfloor \alpha \rfloor}\left(\widetilde{k} I_\meshwidth + \sum_{i=1}^{\widetilde{m}} \widetilde{c}_i (\widetilde{L}_\meshwidth - \widetilde{p}_i I_\meshwidth)^{-1}\right)
\left(k I_\meshwidth + \sum_{i=1}^m c_i (L_\meshwidth - p_i I_\meshwidth)^{-1}\right)\notag\\
&=
\sum_{i=1}^{M} Q_{\meshwidth,i}^{-1},
\end{align*}
where $M = 1 + m + \widetilde{m} + m\widetilde{m}$ and each $Q_{\meshwidth,i}$ is of the form
$\widetilde{L}_\meshwidth^{\lfloor \beta \rfloor}L_\meshwidth^{\lfloor \alpha \rfloor}$,
$\widetilde{L}_\meshwidth^{\lfloor \beta \rfloor}L_\meshwidth^{\lfloor \alpha \rfloor}(\widetilde{L}_\meshwidth - c I_\meshwidth)$,
$\widetilde{L}_\meshwidth^{\lfloor \beta \rfloor}L_\meshwidth^{\lfloor \alpha \rfloor}(L_\meshwidth - c I_\meshwidth)$, or 
$\widetilde{L}_\meshwidth^{\lfloor \beta \rfloor}L_\meshwidth^{\lfloor \alpha \rfloor}(\widetilde{L}_\meshwidth - \widetilde{c} I_\meshwidth)(L_\meshwidth - c I_\meshwidth)$,
where $c,\widetilde{c}<0$ are constants, and each precision operator is multiplied with some constant scaling parameter determined by the coefficients of the rational approximation. Each of these precision operators is a non-fractional operator which can be represented by a matrix ${Q}_i$ in the same way as for the non-fractional case above. For example, $\widetilde{L}_\meshwidth L_\meshwidth (L_\meshwidth - c I_\meshwidth)$ is represented by the matrix ${G}{C}^{-1}(\kappa^2{C} + {G}){C}^{-1}((\kappa^2 - c){C} + {G})$. For all these matrices, we use the standard mass lumping technique to obtain sparse matrices, see \cite{Lin11}.

If $\lfloor\beta\rfloor \geq 1$, all these matrices are first order intrinsic precision matrices. However, when $\beta < 1$, they are of full rank and we must explicitly make the projection onto the space of zero mean functions to obtain an intrinsic field. For the likelihood calculations, this is simply done by a rank-1 adjustment which does not increase the cost.

\subsection{Proofs}
\begin{proof}[Proof of Lemma~\ref{lemma:projection}]
Let $\{e_{i,\meshwidth}\}_{i=1}^{N_{\meshwidth}}$ denote an orthonormal basis of $V_{\meshwidth}$ with $e_{1,\meshwidth} \propto 1$. Since $\widetilde{V}_{\meshwidth} = V_{\meshwidth} \cap H$, we have that $\{e_{i,\meshwidth}\}_{i=2}^{N_{\meshwidth}}$ is an orthonormal basis of $\widetilde{V}_{\meshwidth}$. Then, 
$$
\widetilde{\Pi}_{\meshwidth} g = \sum_{i=2}^{N_{\meshwidth}} (g, e_{i,\meshwidth}) e_{i,\meshwidth} = \widetilde{\Pi}\sum_{i=1}^{N_{\meshwidth}} (g, e_{i,\meshwidth}) e_{i,\meshwidth} = \widetilde{\Pi}\Pi_{\meshwidth} g,
$$
for $g\in L_2(\cD)$.
Further, 
$$
\Pi_{\meshwidth} \widetilde{\Pi} g = \Pi_{\meshwidth} \left(g - \oint g(s) {\rm d}s\right) = \sum_{i=1}^{N_{\meshwidth}} \left(g-\oint g(s) {\rm d}s, e_{i,\meshwidth}\right) e_{i,\meshwidth} =  \sum_{i=2}^{N_{\meshwidth}} (g, e_{i,\meshwidth}) e_{i,\meshwidth},
$$
since $(g-\oint g(s) {\rm d}s, e_{1,\meshwidth}) = 0$ and $(g-\oint g(s) {\rm d}s, e_{i,\meshwidth}) = (g,e_{i,\meshwidth})$ for $i>1$. 

The $H^1(\cD)$-stability of $\Pi_{\meshwidth}$ follows from the fact that we are using a quasi-uniform triangulation \cite[see][]{Bramble2002}. For the $H^1(\cD)$-stability of $\widetilde{\Pi}_{\meshwidth}$, note that H{\"o}lder's inequality gives that 
$
\|\widetilde{\Pi}f\|_{H^1(\cD)} \leq (1+|\cD|)\|f\|_{H^1(\cD)},
$
and we thus have  
$$
\|\widetilde{\Pi}_{\meshwidth}\|_{\cL(H^1(\cD))} \leq \|\widetilde{\Pi}\|_{\cL(H^1(\cD))} \|\Pi_{\meshwidth}\|_{\cL(H^1(\cD))} \leq (1+|\cD|)\|\Pi_{\meshwidth}\|_{\cL(H^1(\cD))}.
$$
\end{proof}

Let $\dot{H}_L^\beta$ denote the domain of the operator $(\kappa^2 - \Delta)^{\beta/2}$, where $\Delta$ is the Neumann Laplacian on $L_2(\cD)$. The following basic result is key for the following arguments:
\begin{lemma}\label{lemma:embedding}
	For any $\theta\geq 0$, $\dot{H}_{\widetilde{L}}^{\theta} \hookrightarrow \dot{H}_{L}^{\theta}$. For any $\theta\in\bbR$ there exists constants $0<c_\theta<C_\theta<\infty$ such that   
	 $$
	 c_\theta \|g \|_{\dot{H}_{\widetilde{L}}^{\theta}} \leq \|g\|_{\dot{H}_{L}^{\theta}} \leq C_\theta \|g \|_{\dot{H}_{\widetilde{L}}^{\theta}} \quad \forall g\in \dot{H}_{\widetilde{L}}^{\theta}.
	 $$
\end{lemma}
\begin{proof}
	Recall that $\{e_i\}_{i\in\bbN}$ denotes the eigenfunctions of the Neumann Laplacian on $\cD$, with corresponding eigenvalues $\{\lambda_i\}_{i\in\bbN}$. Assume that $g\in \dot{H}_{\widetilde{L}}^{\theta}$. Then $g = \sum_{i>1} \lambda_i^{\theta/2}(g,e_i)e_i$ with $\sum_{i>1}\lambda_i^{\theta}(g,e_i)^2 < \infty$. In particular, $(g,e_1)=0$ and thus, 
	$$
	\|g\|_{\dot{H}_{L}^{\theta}} = \sum_{i\in\bbN}(\kappa^2 +\lambda_i)^{\theta}(g,e_i)^2 = \sum_{i>1}(\kappa^2 +\lambda_i)^{\theta}(g,e_i)^2 > \sum_{i>1} \lambda_i^{\theta}(g,e_i)^2 = 
	\|g\|_{\dot{H}_{\widetilde{L}}^{\theta}},
	$$
    for $\theta\geq 0$,
	i.e., $c_\theta = 1$. For $\theta < 0$, we have that $C_\theta = 1$:
	$$
\|g\|_{\dot{H}_{L}^{\theta}} = \sum_{i\in\bbN}(\kappa^2 +\lambda_i)^{\theta}(g,e_i)^2 = \sum_{i>1}(\kappa^2 +\lambda_i)^{\theta}(g,e_i)^2 < \sum_{i>1} \lambda_i^{\theta}(g,e_i)^2 = 
\|g\|_{\dot{H}_{\widetilde{L}}^{\theta}} < \infty.
$$	
	Further, due to the Weyl asymptotics of the eigenvalues, we have for $\theta>0$,
	$$
	\|g\|_{\dot{H}_{L}^{\theta}}  = \sum_{i>1}(\kappa^2\lambda_i^{-1} +1)^{\theta}\lambda_i^{\theta}(g,e_i)^2 < 
	\sum_{i>1}(\kappa^2 c i^{-2/d} +1)^{\theta}\lambda_i^{\theta}(g,e_i)^2
	\leq 
	C_\theta \|g\|_{\dot{H}_{\widetilde{L}}^{\theta}} < \infty,
	$$
	with $C_\theta = 	(\kappa^2 c 2^{-2/d} +1)^{\theta}$. Thus, $g\in \|g\|_{\dot{H}_{L}^{\theta}}$ and 
	$\|g\|_{\dot{H}_{L}^{\theta}} < C_\theta \|g\|_{\dot{H}_{\widetilde{L}}^{\theta}}$. For $\theta < 0$, we similarly get 
		$$
	\|g\|_{\dot{H}_{L}^{\theta}}  = \sum_{i>1}(\kappa^2\lambda_i^{-1} +1)^{\theta}\lambda_i^{\theta}(g,e_i)^2 >
	\sum_{i>1}(\kappa^2 C i^{-2/d} +1)^{\theta}\lambda_i^{\theta}(g,e_i)^2
	\geq 
	c_\theta \|g\|_{\dot{H}_{\widetilde{L}}^{\theta}},
	$$
	with $c_\theta = 1$. 
\end{proof}

We also have the following basic facts about the restricted $L_2(\cD)$ space $H$ which will be used throughout the proofs. In the statement, and throughout the paper, we use the notation $\|A\|_{\cL_2(E,F)}$ to denote the Hilbert--Schmidt norm of an operator $A : E \rightarrow F$ mapping between two Hilbert spaces $E$ and $F$. We also use the shorthand notation $\|A\|_{\cL_2(E)}$ if $E = F$. 
\begin{lemma}\label{lemma:proj}
	Let $\widetilde{\Pi} : L_2(\cD) \rightarrow H$ denote the orthogonal projection onto $H$ and let 
	$\{e_i\}_{i\in\bbN}$ denote an orthogonal basis of $L_2(\cD)$ with $e_1 \propto 1$, so that 
	$\{e_i\}_{i>1}$ is an orthogonal basis of $H$. Then,
	\begin{enumerate}[label= \roman{enumi}.]
	\item The projection $\widetilde{\Pi}$ satisfies $\|\widetilde{\Pi}\|_{\mathcal{L}(L_2(\cD))} = 1$. 
	\item For a linear operator $A : H \rightarrow H$, we have that 
	$\|A\|_{\mathcal{L}_2(H)} = \|A\widetilde{\Pi}\|_{\mathcal{L}_2(L_2(\cD))}$. 
	\item If $B : L_2(\cD) \rightarrow L_2(\cD)$ is a bounded operator that diagonalizes with respect to $\{e_i\}_{i\in\bbN}$,  then $B g \in H$ for any $g\in H$ and $\widetilde{\Pi} B f = B \widetilde{\Pi} f$ for any $f\in L_2(\cD)$. Thus, 
	$$
	 \|AB\|_{\mathcal{L}_2(H)} = \|A\widetilde{\Pi}B\|_{\mathcal{L}_2(L_2(\cD))} =  \|AB\widetilde{\Pi}\|_{\mathcal{L}_2(L_2(\cD))}.
	$$ 
	\item For any $\gamma>0, \theta\in\bbR$, and any linear operator $C : L_2(\cD) \rightarrow L_2(\cD)$, there exists a constant $c\in(0,\infty)$ such that 
	$
	\|\widetilde{\Pi}C\|_{\mathcal{L}(\dot{H}_{\widetilde{L}}^\gamma, \dot{H}_{\widetilde{L}}^\theta)} \leq c \|C\|_{\mathcal{L}(\dot{H}_{L}^\gamma, \dot{H}_{L}^\theta)}.
	$
	\end{enumerate}
\end{lemma}
\begin{proof}
	Recall that $\{e_i\}_{i\in\bbN}$ denotes the eigenfunctions of the Neumann Laplacian on $\cD$, which form an orthogonal basis of $L_2(\cD)$. Also recall that $\{e_i\}_{i>1}$ form an orthogonal basis of $H$. Now, (i) follows directly from the fact that $\widetilde{\Pi}$ is an orthogonal projection. Note that 
$$
\|A\|_{\cL_2(H)}^2 = \sum_{i>1}\|A e_i\|_{H}^2 = \sum_{i>1}\|A e_i\|_{L_2(\cD)}^2  
= \sum_{i\in\bbN}\|A \widetilde{\Pi} e_i\|_{L_2(\cD)}^2 = \|A\widetilde{\Pi}\|_{\cL_2(L_2(\cD))}^2,
$$
since $\widetilde{\Pi}e_1 = 0$ and $\widetilde{\Pi}e_i = e_i$ for $i>1$. Thus, (ii) holds. For (iii) note that if $g\in H$, then $g = \sum_{i>1} g_i e_i$ with $\sum_{i>1} g_i^2 < \infty$ and 
$$
B g =  \sum_{i>1} g_i B e_i = \sum_{i>1} g_i \lambda_{B,i} e_i \in H
$$
since $B g \in L_2(\cD)$, and hence $\|B g\|_{L_2(\cD)} = \sum_{i>1} g_i^2 \lambda_{B,i}^2 < \infty$. For $f = \sum_{i\in\bbN} f_i e_i \in L_2(\cD)$, we have that 
$$
B \widetilde{\Pi} f = B \sum_{i>1} f_i e_i =  \sum_{i>1} f_i \lambda_{B,i} e_i = \widetilde{\Pi} \sum_{i\in \bbN} f_i \lambda_{B,i} e_i = \widetilde{\Pi} B f.
$$
Finally, for (iv), we have by Lemma~\ref{lemma:embedding} that there exists a $c\in(0,\infty)$ such that
\begin{align*}
\|\widetilde{\Pi}C\|_{\mathcal{L}(\dot{H}_{\widetilde{L}}^\gamma, \dot{H}_{\widetilde{L}}^\theta)} &= \sup_{u \in \dot{H}_{\widetilde{L}}^\gamma : u \neq 0} \frac{\|\widetilde{\Pi}Cu\|_{\dot{H}_{\widetilde{L}}^\theta}}{\|u\|_{\dot{H}_{\widetilde{L}}^\gamma}} \leq 
c \sup_{u \in \dot{H}_{\widetilde{L}}^\gamma : u \neq 0} \frac{\|\widetilde{\Pi}Cu\|_{\dot{H}_{L}^\theta}}{\|u\|_{\dot{H}_{L}^\gamma}} \\
&\leq 
c \sup_{u \in \dot{H}_{\widetilde{L}}^\gamma : u \neq 0} \frac{\|Cu\|_{\dot{H}_{L}^\theta}}{\|u\|_{\dot{H}_{L}^\gamma}} 
\leq 
c \sup_{u \in \dot{H}_{L}^\gamma : u \neq 0} \frac{\|Cu\|_{\dot{H}_{L}^\theta}}{\|u\|_{\dot{H}_{L}^\gamma}} = 
c\|C\|_{\mathcal{L}(\dot{H}_{L}^\gamma, \dot{H}_{L}^\theta)}.
\end{align*}
\end{proof}

To bound the finite element error, we also need the following three lemmata. 

\begin{lemma}\label{lemma:interp}
There exists a linear operator $\cI_{\meshwidth} : H^2(\cD) \rightarrow V_{\meshwidth}$ such that for all $1/2 < \theta \leq 2$, $\cI_{\meshwidth} : H^{\theta} \rightarrow V_{\meshwidth}$ is a continuous extension, and 
$$
\|v - \cI_{\meshwidth} v\|_{\dot{H}_{\widetilde{L}}^\sigma(\cD)} \lesssim h^{\theta - \sigma}\|v\|_{\dot{H}_{\widetilde{L}}^\theta(\cD)}  \quad \forall v\in \dot{H}_{\widetilde{L}}^\theta(\cD),
$$
for all $0\leq \sigma \leq \min\{1,\theta\}$. 
\end{lemma}
\begin{proof}
Due to the quasi-uniformity of the mesh and the $H^2$-regularity of the operators, we have by \citet[Remark 5]{BSZ2023} that there exists an operator $\cI_{\meshwidth} : H^\theta(\cD) \rightarrow V_{\meshwidth}$ such that for all $1\leq \theta < 2$ and $v\in H^\theta(\cD)$, 
$
\|v - \cI_{\meshwidth} v\|_{H^\theta(\cD)} \lesssim h^{\theta}\|v\|_{H^\theta(\cD)}.
$
By \cite[Proposition 4]{BSZ2023}, we then have that 
$$
\|v - \cI_{\meshwidth} v\|_{\dot{H}_{L}^\theta(\cD)} \lesssim h^{\theta}\|v\|_{\dot{H}_{L}^\theta(\cD)}  \quad \forall v\in \dot{H}_{L}^\theta(\cD),
$$
which by a standard interpolation argument yields that 
$$
\|v - \cI_{\meshwidth} v\|_{\dot{H}_{L}^\sigma(\cD)} \lesssim h^{\theta-\sigma}\|v\|_{\dot{H}_{L}^\theta(\cD)}  \quad \forall v\in \dot{H}_{L}^\theta(\cD).
$$
Applying Lemma~\ref{lemma:embedding} now gives the desired result. 
\end{proof}

\begin{lemma}\label{lemma:bound}
	For $0\leq \tau \leq 1/2$, we have that 
	$
	\|\widetilde{L}^\tau \widetilde{L}_{\meshwidth}^{-\tau}\widetilde{\Pi}_{\meshwidth}\|_{\cL(H)} \lesssim 1.
	$
	Further, for $0\leq \gamma \leq 1$, 
	\begin{equation}\label{eq:lem_bound2}
	\|\widetilde{L}^{-\gamma} \widetilde{L}_{\meshwidth}^{\gamma}\widetilde{\Pi}_{\meshwidth}\|_{\cL(H)} \lesssim 1
	\end{equation}
\end{lemma}
\begin{proof}
	To prove \eqref{eq:lem_bound2}, note that $\|\widetilde{L}^{-\gamma} \widetilde{L}_{\meshwidth}^{\gamma}\widetilde{\Pi}_{\meshwidth}\|_{\cL(H)} = \| \widetilde{L}_{\meshwidth}^{\gamma}\widetilde{\Pi}_{\meshwidth}\widetilde{L}^{-\gamma}\|_{\cL(H)}$ since both $\widetilde{L}^{-\gamma}$ and $\widetilde{L}_{\meshwidth}^{-\gamma}\widetilde{\Pi}_{\meshwidth}$ are self adjoint and since $\|A\|_{\cL(H)} = \|A^*\|_{\cL(H)}$ for any operator $A$, where $A^*$ denotes the adjoint of $A$. It is therefore sufficient to prove that $ \| \widetilde{L}_{\meshwidth}^{\gamma}\widetilde{\Pi}_{\meshwidth}\widetilde{L}^{-\gamma}\|_{\cL(H)} \lesssim 1$. The proof now follows from the same arguments as those in the proof of \cite[][Lemma 7]{Cox2020}.
\end{proof}

\begin{lemma}\label{lemma:fem}
For $-1 \leq \gamma \leq 2$, $\gamma \neq 3/2$,  $\tau>0$, and $0\leq \sigma \leq 1$ such that $2\tau + \gamma - \sigma >0$ and $2\tau - \sigma \geq 0$, we have that 
$$
\|\widetilde{L}^{-\tau} - \widetilde{L}_{\meshwidth}^{-\tau}\widetilde{\Pi}_{\meshwidth}\|_{\mathcal{L}(\dot{H}_{\widetilde{L}}^{\gamma},\dot{H}_{\widetilde{L}}^{\sigma})} \lesssim h^{\min\{2\tau + \gamma -\sigma - \eps,2-\sigma, 2+\gamma, 2\}}.
$$
\end{lemma}
\begin{proof}
	By Lemmata~\ref{lemma:interp} and \ref{lemma:bound}, the result follows from the same argument as in 
	the proof of \cite[Theorem 1]{Cox2020}. We omit the details for brevity. 
\end{proof}

We are now ready to prove the rate of the FEM approximation.
\begin{proof}[Proof of Proposition \ref{cov_fem_approx_rate}]
	First, note that since $\varrho^{\alpha,\beta},\varrho_{\meshwidth}^{\alpha,\beta} \in H\times H$, which is equipped with the $L_2(\cD\times \cD)$ norm, we have that 
	$$
	\|\varrho^{\alpha,\beta} - \varrho_{\meshwidth}^{\alpha,\beta}\|_{L_2(\mathcal{D}\times\mathcal{D})} = \|\varrho^{\alpha,\beta} - \varrho_{\meshwidth}^{\alpha,\beta}\|_{L_2(H\times H)}.
	$$
	We will split the proof in two cases depending on the values of $\alpha$ and $\beta$. 
	
		{\bf Case 1 ($\beta <1/2 + d/4$ or $\alpha> 1/2$):}
	Let us begin by introducing the covariance function $\hat{\varrho}^{\alpha,\beta}$ corresponding to the covariance operator $\widetilde{L}_{\meshwidth}^{-\beta}\widetilde{\Pi}_{\meshwidth}L^{-\alpha}$ and use the triangle inequality to partition the error as 
	$$
	\|\varrho^{\alpha,\beta} - \varrho_{\meshwidth}^{\alpha,\beta}\|_{L_2(H\times H)} \leq 
	\|\varrho^{\alpha,\beta} - \hat{\varrho}^{\alpha,\beta}\|_{L_2(H\times H)} + \|\hat{\varrho}^{\alpha,\beta} - \varrho_{\meshwidth}^{\alpha,\beta}\|_{L_2(H\times H)} := \text{(I)} + \text{(II)}
	$$
	We now bound these two terms separately, starting with (I). For that, fix $\eps>0$ and 
	let $\gamma = \min\{2\alpha-d/2-\eps,2\}$. Then, by Lemma~\ref{lemma:fem}
	\begin{align*}
		\text{(I)} &= \|\widetilde{L}^{-\beta}L^{-\alpha} - \widetilde{L}_{\meshwidth}^{-\beta}\widetilde{\Pi}_{\meshwidth}L^{-\alpha}\|_{\mathcal{L}_2(H)} \leq
		\|(\widetilde{L}^{-\beta} - \widetilde{L}_{\meshwidth}^{-\beta}\widetilde{\Pi}_{\meshwidth})\widetilde{L}^{-\gamma/2}\|_{\mathcal{L}(H)}
		\|\widetilde{L}^{\gamma/2}L^{-\alpha}\|_{\mathcal{L}_2(H)}\\
		&\lesssim h^{\min\{2\alpha + 2\beta - d/2 -\eps,2\}}\|\widetilde{L}^{\gamma/2}L^{-\alpha}\|_{\mathcal{L}_2(H)}.
	\end{align*}
	Due to the Weyl asymptotics of the eigenvalues, we have that 
	$$
	\|\widetilde{L}^{\gamma/2}L^{-\alpha}\|_{\mathcal{L}_2(H)}^2
	= \sum_{i\in\bbN}\|\widetilde{L}^{\gamma/2}L^{-\alpha} e_j\|_{H}^2 = 
	\sum_{i>1}\lambda_j^{\gamma - 2\alpha} \leq C \sum_{i\in\bbN}j^{2(\gamma - 2\alpha)/d} < \infty,
	$$
	since $2(\gamma - 2\alpha)/d < -1$.
	For the second term, choose $\gamma = \min\{2\beta - d/2 - \eps,2 - 2\alpha\}$. Then 
	\begin{align*}
		\text{(II)} &= \|\widetilde{L}_{\meshwidth}^{-\beta}\widetilde{\Pi}_{\meshwidth}L^{-\alpha}-\widetilde{L}_{\meshwidth}^{-\beta}\widetilde{\Pi}_{\meshwidth}L_{\meshwidth}^{-\alpha}\Pi_{\meshwidth}\|_{\mathcal{L}_2(H)} 
		= 
		\|\widetilde{L}_{\meshwidth}^{-\beta}\widetilde{\Pi}_{\meshwidth}(L^{-\alpha}-L_{\meshwidth}^{-\alpha}\Pi_{\meshwidth})\|_{\mathcal{L}_2(H)} 
		\\
		&\leq 
		\|\widetilde{L}_{\meshwidth}^{-\beta}\widetilde{\Pi}_{\meshwidth}\|_{\mathcal{L}_2(H,\dot{H}_{\widetilde{L}}^{\gamma})} 
		\|L^{-\alpha}-L_{\meshwidth}^{-\alpha}\Pi_{\meshwidth}\|_{\mathcal{L}(\dot{H}_{\widetilde{L}}^{\gamma}, H)}.
	\end{align*}
	By Lemma~\ref{lemma:proj} and \cite[Proposition 1]{BSZ2023}, we have 
	\begin{align*}
	\|L^{-\alpha}-L_{\meshwidth}^{-\alpha}\Pi_{\meshwidth}\|_{\mathcal{L}(\dot{H}_{\widetilde{L}}^{\gamma}, H)} &\lesssim \|L^{-\alpha}-L_{\meshwidth}^{-\alpha}\Pi_{\meshwidth}\|_{\mathcal{L}(\dot{H}_{L}^{\gamma}, L_2(\cD))} \\
	&\lesssim h^{\min\{2\alpha + \gamma - \eps, 2\}} = h^{\min\{2\alpha + 2\beta - d/2 - \eps, 2\}}.
	\end{align*}
	Thus, it remains to show that $\|\widetilde{L}_{\meshwidth}^{-\beta}\widetilde{\Pi}_{\meshwidth}\|_{\mathcal{L}_2(H,\dot{H}_{\widetilde{L}}^{\gamma})} = 
	\|\widetilde{L}^{\gamma/2}\widetilde{L}_{\meshwidth}^{-\beta}\widetilde{\Pi}_{\meshwidth}\|_{\mathcal{L}_2(H)}$ is bounded. We split this into four different cases:
	
	Case 1.A: $\beta < d/4$. In this case, $\gamma < 0$ and we define $\beta^* = \gamma/2+d/4-\beta-\eps \in (0,1)$. Then
	$$
	\|\widetilde{L}^{\gamma/2}\widetilde{L}_{\meshwidth}^{-\beta}\widetilde{\Pi}_{\meshwidth}\|_{\mathcal{L}_2(H)} \leq 
	\|\widetilde{L}^{\gamma/2+\beta^*}\|_{\mathcal{L}(H)}
	\|\widetilde{L}^{-\beta^*}\widetilde{L}_{\meshwidth}^{\beta^*}\widetilde{\Pi}_{\meshwidth}\|_{\mathcal{L}(H)}
	\|\widetilde{L}_{\meshwidth}^{-d/4}\widetilde{\Pi}_{\meshwidth}\|_{\mathcal{L}_2(H)}.
	$$	
	Here, the first term on the left hand side is bounded since $\gamma/2 + \beta^* < 0$, the third term is bounded because of the Weyl asymptotics of the eigenvalues, and the second term is bounded by Lemma~\ref{lemma:bound} since $\beta^* \in (0, 1)$.
	
	Case 1.B: $d/4 \leq \beta < 1/2 + d/4$ and $\alpha < 1$. In this case, $\gamma \in (0,1)$ and we bound
	$$
	\|\widetilde{L}^{\gamma/2}\widetilde{L}_{\meshwidth}^{-\beta}\widetilde{\Pi}_{\meshwidth}\|_{\mathcal{L}_2(H)} \leq 
	\|\widetilde{L}^{\gamma/2}\widetilde{L}_{\meshwidth}^{-\gamma/2}\widetilde{\Pi}_{\meshwidth}\|_{\mathcal{L}(H)}
	\|\widetilde{L}_{\meshwidth}^{-(\beta - \gamma/2)}\widetilde{\Pi}_{\meshwidth}\|_{\mathcal{L}_2(H)}.
	$$	
	Here, the first term is bounded by Lemma~\ref{lemma:bound} since $\gamma/2\in(0,1/2)$ and the second  term is bounded because of the Weyl asymptotics of the eigenvalues and the fact that $\beta - \gamma/2 > d/4$. 
	
	Case 1.C: $\beta > d/4$ and $\alpha > 1$. In this case, $\gamma < 0$ and we bound
	$$
	\|\widetilde{L}^{\gamma/2}\widetilde{L}_{\meshwidth}^{-\beta}\widetilde{\Pi}_{\meshwidth}\|_{\mathcal{L}_2(H)} \leq 
	\|\widetilde{L}^{\gamma/2}\|_{\mathcal{L}(H)}\|\widetilde{L}_{\meshwidth}^{-(\beta-d/4-\eps)}\widetilde{\Pi}_{\meshwidth}\|_{\mathcal{L}(H)}
	\|\widetilde{L}_{\meshwidth}^{-d/4-\eps}\widetilde{\Pi}_{\meshwidth}\|_{\mathcal{L}_2(H)}.
$$		
	Here, the first two terms are bounded because $\gamma <0 $ and $\beta - d/4 - \eps > 0$, whereas the second term again is bounded because of the Weyl asymptotics of the eigenvalues.
	
	Case 1.D: $\beta > 1/2 + d/4$ and $\alpha < 1$. Here, we have $\gamma \in (0,1)$ and bound 
		$$
	\|\widetilde{L}^{\gamma/2}\widetilde{L}_{\meshwidth}^{-\beta}\widetilde{\Pi}_{\meshwidth}\|_{\mathcal{L}_2(H)} \leq 
	\|\widetilde{L}^{\gamma/2}\widetilde{L}_{\meshwidth}^{-\gamma/2}\widetilde{\Pi}_{\meshwidth}\|_{\mathcal{L}(H)}
	\|\widetilde{L}_{\meshwidth}^{\gamma/2 - \beta + d/4 - \eps }\widetilde{\Pi}_{\meshwidth}\|_{\mathcal{L}(H)}\|\widetilde{L}_{\meshwidth}^{-d/4+\eps}\widetilde{\Pi}_{\meshwidth}\|_{\mathcal{L}_2(H)}.
	$$		
	Here, the first term is bounded by Lemma~\ref{lemma:bound} since $\gamma < 1$, the second is bounded because $\gamma/2 - \beta + d/4 - \eps < 0$, and the third is bounded  because of the Weyl asymptotics of the eigenvalues.

	{\bf Case 2 ($\alpha \leq 1/2$ and $\beta \geq 1/2 + d/4$):}
	The proof for this case is essentially the same as that of the first case, except that we split the error differently. We introduce the covariance function $\tilde{\varrho}^{\alpha,\beta}$ corresponding to the covariance operator $\widetilde{L}^{-\beta}L_{\meshwidth}^{-\alpha}\Pi_{\meshwidth}$ and use the triangle inequality to partition the error as 
	$$
	\|\varrho^{\alpha,\beta} - \varrho_{\meshwidth}^{\alpha,\beta}\|_{L_2(H\times H)} \leq 
	\|\varrho^{\alpha,\beta} - \widetilde{\varrho}^{\alpha,\beta}\|_{L_2(H\times H)} + \|\widetilde{\varrho}^{\alpha,\beta} - \varrho_{\meshwidth}^{\alpha,\beta}\|_{L_2(H\times H)} := \text{(I')} + \text{(II')}
	$$
	Fix $\eps >0$ and let $\gamma = \min\{2\beta-d/2-\eps,2\}$ Then 
		\begin{align*}
		\text{(I')} &= \|\widetilde{L}^{-\beta}L^{-\alpha} - \widetilde{L}^{-\beta}L_{\meshwidth}^{-\alpha}\Pi_{\meshwidth}\|_{\mathcal{L}_2(H)} = 
		\|\widetilde{L}^{-\beta}(L^{-\alpha} - L_{\meshwidth}^{-\alpha}\Pi_{\meshwidth})\|_{\mathcal{L}_2(H)} \\
		&\leq
		\|(L^{-\alpha} - L_{\meshwidth}^{-\alpha}\Pi_{\meshwidth})L^{-\gamma/2}\|_{\mathcal{L}(H)}
		\|L^{\gamma/2}\widetilde{L}^{-\beta}\|_{\mathcal{L}_2(H)}\\
		&\lesssim h^{\min\{2\alpha + 2\beta - d/2 -\eps,2\}}\|\widetilde{L}^{\gamma/2}L^{-\beta}\|_{\mathcal{L}_2(H)},
	\end{align*}
where $\|\widetilde{L}^{\gamma/2}L^{-\beta}\|_{\mathcal{L}_2(H)}$ is bounded due to the Weyl asymptotics of the eigenvalues. 
For the second term, let  $\gamma = \min\{2\alpha - d/2 - \eps,2 - 2\beta\}$. Then
\begin{align*}
	(II') &= \|\widetilde{L}^{-\beta}L_{\meshwidth}^{-\alpha}\Pi_{\meshwidth} - \widetilde{L}_{\meshwidth}^{-\beta}L_{\meshwidth}^{-\alpha}\Pi_{\meshwidth}\|_{\mathcal{L}_2(H)} = 
	\|(\widetilde{L}^{-\beta} - \widetilde{L}_{\meshwidth}^{-\beta})L_{\meshwidth}^{-\alpha}\Pi_{\meshwidth}\|_{\mathcal{L}_2(H)} \\ 
	&\leq 
	\|(\widetilde{L}^{-\beta} - \widetilde{L}_{\meshwidth}^{-\beta})\widetilde{L}^{-\gamma/2}\|_{\mathcal{L}(H)}
	\|\widetilde{L}^{\gamma/2}L_{\meshwidth}^{-\alpha}\Pi_{\meshwidth}\|_{\mathcal{L}_2(H)} \\
	&\lesssim h^{\min\{2\beta + 2\alpha - d/2 - \eps, 2\}}\|\widetilde{L}^{\gamma/2}L_{\meshwidth}^{-\alpha}\Pi_{\meshwidth}\|_{\mathcal{L}_2(H)}.
\end{align*}
Thus, the final step is to show that $\|\widetilde{L}^{\gamma/2}L_{\meshwidth}^{-\alpha}\Pi_{\meshwidth}\|_{\mathcal{L}_2(H)}$ is bounded. First, note that we by Lemma~\ref{lemma:embedding} have that 
$\|\widetilde{L}^{\gamma/2}L_{\meshwidth}^{-\alpha}\Pi_{\meshwidth}\|_{\mathcal{L}_2(H)} \lesssim \|L^{\gamma/2}L_{\meshwidth}^{-\alpha}\Pi_{\meshwidth}\|_{\mathcal{L}_2(H)}$. We now bound 
$$
\|L^{\gamma/2}L_{\meshwidth}^{-\alpha}\Pi_{\meshwidth}\|_{\mathcal{L}_2(H)} \leq 
\|L^{\gamma/2-\alpha}\|_{\mathcal{L}_2(H)}
\|L^{\alpha}L_{\meshwidth}^{-\alpha}\Pi_{\meshwidth}\|_{\mathcal{L}(H)}.
$$
Here, the first term is bounded because of the Weyl asymptotics of the eigenvalues and the fact that $\gamma/2 - \alpha < d/4$, and the second term is bounded by \cite[Lemma 7]{Cox2020} since $\alpha \leq1/2$.
\end{proof}

Finally, we can prove the main theorem. 

\begin{proof}[Proof of Theorem \ref{thm:coverror}]
	By the triangle inequality, we have that 
	\begin{align}
	\|\varrho_{\meshwidth,m}^\beta - \varrho^\beta\|_{L_2(\mathcal{D}\times\mathcal{D})} &\leq 	\|\varrho^{\alpha,\beta} - \varrho_{\meshwidth}^{\alpha,\beta}\|_{L_2(\mathcal{D}\times\mathcal{D})}  + 
		\|\varrho_{\meshwidth}^{\alpha,\beta} - \varrho_{\meshwidth,m,\widetilde{m}}^{\alpha,\beta}\|_{L_2(\cD \times \cD)} \\
		&= \text{(I)} + \|\varrho_{\meshwidth}^{\alpha,\beta} - \varrho_{\meshwidth,m,\widetilde{m}}^{\alpha,\beta}\|_{L_2(\cD \times \cD)}.
	\end{align}
	The term (I) here is bounded by Proposition~\ref{cov_fem_approx_rate}, as 
	\begin{equation}\label{eq:term1}
		(I) \lesssim h^{\min\{2\alpha + 2\beta-d/2 -\varepsilon,2\}}.
	\end{equation}
	Thus, we only need to handle the second. For that, we have that 
	\begin{align*}
		\|\varrho_{\meshwidth}^{\alpha,\beta} - \varrho_{\meshwidth,m,\widetilde{m}}^{\alpha,\beta}\|_{L_2(\cD \times \cD)}^2 &= \|\widetilde{L}_{\meshwidth}^{-\beta} \widetilde{\Pi}_{\meshwidth} L_{\meshwidth}^{-\alpha}\Pi_{\meshwidth} - \widetilde{L}_{\meshwidth,\widetilde{m}}^{\beta} \widetilde{\Pi}_{\meshwidth} L_{\meshwidth,m}^{\alpha}\Pi_{\meshwidth}\|_{\mathcal{L}_2(L_2(\mathcal{D}))}^2 \\
		&\leq 
		\|\widetilde{L}_{\meshwidth}^{-\beta}\widetilde{\Pi}_{\meshwidth} L_{\meshwidth}^{-\alpha}\Pi_{\meshwidth} - \widetilde{L}_{\meshwidth}^{-\beta}\widetilde{\Pi}_{\meshwidth} L_{\meshwidth,m}^{\alpha}\Pi_{\meshwidth}\|_{\mathcal{L}_2(L_2(\mathcal{D}))}^2 \\
		&\quad+
		\| \widetilde{L}_{\meshwidth}^{-\beta} \widetilde{\Pi}_{\meshwidth} L_{\meshwidth,m}^{\alpha}\Pi_{\meshwidth} - \widetilde{L}_{\meshwidth,\widetilde{m}}^{\beta}\widetilde{\Pi}_{\meshwidth} L_{\meshwidth,m}^{\alpha}\Pi_{\meshwidth}\|_{\mathcal{L}_2(L_2(\mathcal{D}))}^2\\
		&:= \text{(II)} + \text{(III)}.
	\end{align*}
	We now bound these two terms separately.
	To simplify the analysis, we without loss of generality normalize $L$ and $\widetilde{L}$ so that $\lambda_1, \widetilde{\lambda}_1 \geq 1$. This is done by writing the covariance operator as $c_1^{\beta} c_2^{\alpha} (\widetilde{L}/c_1)^{-\beta}(L/c_2)^{-\alpha}$ for some constants $c_1,c_2>0$ and then using $\widetilde{L}/c_1$ as $\widetilde{L}$ and $L/c_2$ as $L$.
	We let $\{\lambda_{\meshwidth,m}\}_{m=1}^{n_{\meshwidth}}$ denote the eigenvalues of $L_{\meshwidth}$, and let 
	$\{\widetilde{\lambda}_{\meshwidth,m}\}_{m=1}^{n_{\meshwidth}-1}$ be the eigenvalues of   $\widetilde{L}_{\meshwidth}$, where $\widetilde{\lambda}_{\meshwidth,m} = \lambda_{\meshwidth,m+1}-\kappa^2 > 0$.
	By \cite[Proposition B.2, item 2]{BSZ2023}, we have that $J_{\meshwidth} \subset J$, where $J_{\meshwidth} = [\lambda_{n_{\meshwidth},\meshwidth}^{-1}, \lambda_{1,\meshwidth}^{-1}]$ and $J = [0,\lambda_1^{-1}]$, since $\lambda_1$ is the smallest eigenvalue of $L$. 
	Thus, $J_{\meshwidth} \subset J \subset [0,1]$.
	By the normalization of the operators, \cite[Proposition B.2, item 2]{BSZ2023}, and the fact that $\widetilde{\lambda}_{\meshwidth,m} = \lambda_{\meshwidth,m}-\kappa^2 $, we also have that $\widetilde{J}_{\meshwidth} = [\widetilde{\lambda}_{n_{\meshwidth},\meshwidth}^{-1}, \widetilde{\lambda}_{1,\meshwidth}^{-1}] \subset \widetilde{J} = [0,\widetilde{\lambda}_1^{-1}] \subset [0,1]$.   
	
	Now, let $f_\alpha(x) = x^{\alpha}$ and $\hat{f}_\alpha(x) = x^{\{\alpha\}}$, so that $f_\alpha(x) = x^{\lfloor \alpha\rfloor}\hat{f}_\alpha(x)$. Let $\hat{r}_{\alpha,J}(x) = \frac{p_\alpha(x)}{q_\alpha(x)}$ be the $L_{\infty}$-best approximation of $\hat{f}_\alpha(x)$ on an interval $J \subset [0,1]$, and define the function $r_{\alpha,J}(x) = x^{\lfloor \alpha\rfloor}\hat{r}_{\alpha,J}(x)$. Then, we have the following bound for (II):
	\begin{align*}
		\text{(II)}
		&=
		\sum_{j = 1}^{n_{\meshwidth}} \|\widetilde{L}_{\meshwidth}^{-\beta}\widetilde{\Pi}_{\meshwidth} L_{\meshwidth}^{-\alpha} e_{j,\meshwidth} -\widetilde{L}_{\meshwidth}^{-\beta}\widetilde{\Pi}_{\meshwidth} L_{\meshwidth,m}^{\alpha} e_{j,\meshwidth} \|_{L_2(\mathcal{D})}^2 
		=
		\sum_{j = 2}^{n_{\meshwidth}} (\widetilde{\lambda}_{j-1,\meshwidth}^{-\beta}(\lambda_{j,\meshwidth}^{-\alpha}-r_{\alpha,J_{\meshwidth}}(\lambda_{j,\meshwidth}^{-1})))^2
		\nonumber \\
		&\leq
		n_{\meshwidth} \widetilde{\lambda}_{1,\meshwidth}^{-\alpha}\max\limits_{2 \leq j \leq n_{\meshwidth}} \lvert \lambda_{j,\meshwidth}^{-\beta}-r_{\alpha,J_{\meshwidth}}(\lambda_{j,\meshwidth}^{-1}) \rvert^2 
		\leq
		n_{\meshwidth} \max\limits_{2 \leq j \leq n_{\meshwidth}} \lvert \lambda_{j,\meshwidth}^{-\beta}-r_{\alpha,J_{\meshwidth}}(\lambda_{j,\meshwidth}^{-1}) \rvert^2.
	\end{align*}
	By  \cite[Theorem 1]{Stahl2003}, and the fact that $x^{\lfloor \alpha \rfloor} \leq 1$ on $J_{\meshwidth}$, we have:
	\begin{align*}
		\max\limits_{1 \leq j \leq n_{\meshwidth}} \lvert \lambda_{j,\meshwidth}^{-\alpha}-r_{\alpha,J_{\meshwidth}}(\lambda_{j,\meshwidth}^{-1}) \rvert
		&\leq 
		\sup\limits_{x \in J_{\meshwidth}} \lvert f_\alpha(x)-r_{\alpha,J_{\meshwidth}}(x)\rvert \\
		&\leq 
		\sup\limits_{x \in [0,1]} \lvert \hat{f}(x)-\hat{r}_{\alpha,J_{\meshwidth}}(x)\rvert
		\lesssim 
		e^{-2\pi\sqrt{\{\alpha\}  m}}.
	\end{align*}
	Thus, in total we get the bound  $(II) \lesssim n_{\meshwidth}^{1/2}e^{-2\pi\sqrt{\{\alpha\}  m}}$ and by \cite[Proposition B.2, item 3]{BSZ2023}, we obtain
	$n_{\meshwidth}^{1/2}e^{-2\pi\sqrt{\{\alpha\}  m}} \lesssim h^{-d/2}e^{-2\pi\sqrt{\{\alpha\}  m}}$. 	This source of error only occurs if we need the rational approximation, i.e., if $\alpha \notin \mathbb{N}$, so in total we have the bound 
	\begin{equation}\label{eq:term2}
	\text{(II)} \lesssim 1_{\alpha \notin \mathbb{N}}h^{-d/2}e^{-2\pi\sqrt{\{\alpha\}  m}}
	\end{equation}
	
	The term (III) is handled similarly. Again using  \cite[Proposition B.1, item 3]{BSZ2023} and \cite[Theorem~1]{Stahl2003},
	we have the bound
		\begin{align}\label{eq:term3}
		\text{(III)}
		&=
		\sum_{j = 1}^{n_{\meshwidth}} \|
		\widetilde{L}_{\meshwidth}^{-\beta}\widetilde{\Pi}_{\meshwidth} L_{\meshwidth,m}^{\alpha}e_{j,\meshwidth} - \widetilde{L}_{\meshwidth,\widetilde{m}}^{\beta}\widetilde{\Pi}_{\meshwidth} L_{\meshwidth,m}^{\alpha}e_{j,\meshwidth} \|_{L_2(\mathcal{D})}^2 \nonumber \\
		&=
		\sum_{j = 2}^{n_{\meshwidth}} (r_{\alpha,J_{\meshwidth}}(\lambda_{j,\meshwidth})(\widetilde{\lambda}_{j,\meshwidth}^{-\beta}-r_{\beta,\widetilde{J}_{\meshwidth}}(\widetilde{\lambda}_{j,\meshwidth}^{-1})))^2
		\leq
		2 n_{\meshwidth} \max\limits_{1 \leq j \leq n_{\meshwidth}} \lvert \widetilde{\lambda}_{j,\meshwidth}^{-\beta}-r_{\meshwidth}(\widetilde{\lambda}_{j,\meshwidth}^{-1}) \rvert^2 \nonumber \\
		& \lesssim 1_{\beta \notin \mathbb{N}}h^{-d/2}e^{-2\pi\sqrt{\{\beta\}  m}},
	\end{align}
	where we in the third step used that 
	$$
	r_{\alpha,J_{\meshwidth}}(\lambda_{j,\meshwidth}) \leq \sup\limits_{x \in [0,1]} \lvert r_{\alpha,J_{\meshwidth}}(x)\rvert 
	\lesssim 1+e^{-2\pi\sqrt{\{\alpha\}  m}} \leq 2.
	$$
	Now, combining \eqref{eq:term1}, \eqref{eq:term2} and \eqref{eq:term3} gives the desired bound.
\end{proof}

Theorem~\ref{cov_fem_approx_rate_general} is now obtained as a simple corollary of this result:

\begin{proof}[Proof of Theorem~\ref{cov_fem_approx_rate_general}]
	By the definitions of the variograms in terms of the corresponding covariance functions, we obtain that 
	$$
	\|\gamma^{\alpha,\beta} - \gamma_{\meshwidth,m,\widetilde{m}}^{\alpha,\beta}\|_{L_2(\cD \times \cD)}^2 \leq (2|\cD| + 4)\|\varrho^{\alpha,\beta} - \varrho_{\meshwidth,m,\widetilde{m}}^{\alpha,\beta}\|_{L_2(\cD) \times L_2(\cD)}^2.
	$$
	The result then follows by Theorem~\ref{thm:coverror}.
\end{proof}

\section{A link between proper and intrinsic GMRFs}\label{sec:Prop}

Any proper GMRF with (positive definite) precision $Q$ has a well defined variogram matrix $\Gamma$, and this induces a well defined first-order intrinsic precision matrix $\Theta$ (also referred to as a \HR{} precision matrix). 
The precision $Q$ may be the sparse approximation of a proper stationary field, such as a Mat\'ern field.
If we recall that (1) extremal independence between sites $i$ and $j$ is equivalent to $\Gamma_{ij}=\infty$ and (2) the variogram of a stationary field is uniformly bounded, then we see that two locations that are infinitely far from each other will not be independent in the corresponding H\"{u}sler--Reiss (or Brown--Resnick) process. This limits the applicability of these classes of extreme value processes in practice.

Nonetheless, the next theorem describes the mapping from $Q$ to its corresponding H\"{u}sler--Reiss precision matrix $\Theta$.
\begin{proposition}\label{thm:GR}
If $Q$ is the precision of a proper GMRF then the corresponding H\"{u}sler--Reiss precision matrix $\Theta$ is
\[
\Theta = Q- \frac{Q \mathbf{1} \mathbf{1}^{\text{T}}Q}{\mathbf{1}^{\text{T}}Q \mathbf{1}},
\]
which can be expressed as 
\begin{equation}\label{eqn:CR}
\Theta_{ij} = Q_{ij} - \frac{\left(\sum_{k=1}^d Q_{ik} \right)\left(\sum_{k=1}^d Q_{jk} \right)}{ \sum_{i,j=1}^{d} Q_{ij}}, \quad \text{ for all } i,j \in V .
\end{equation}
\end{proposition}
\begin{proof}
Applying \cite[Proposition 3.3]{hentschel2024statistical} in the first step, and the Woodbury matrix identity in the second step, we have 
\begin{align*}
\Theta &= \lim_{\sigma^2 \to \infty} (Q^{-1} + \sigma^2 \mathbf{1}\mathbf{1}^\top)^{-1} 
=Q - \lim_{\sigma^2 \to \infty} \left[ \sigma^2 Q \boldsymbol{1}(1 + \sigma^2 \boldsymbol{1}^\top Q \boldsymbol{1})^{-1} \boldsymbol{1}^\top Q \right] \\
&=Q - \lim_{\sigma^2 \to \infty}  \left[ \frac{\sigma^2}{1 + \sigma^2 \boldsymbol{1}^\top Q \boldsymbol{1}} \right]Q \boldsymbol{1} \boldsymbol{1}^\top Q 
= Q - \frac{ Q \boldsymbol{1} \boldsymbol{1}^\top Q}{\boldsymbol{1}^\top Q \boldsymbol{1}},
\end{align*}
which can be equivalently represented as in \eqref{eqn:CR}.
\end{proof}

A consequence of Proposition~\ref{thm:GR} is if $Q$ is a sparse positive definite precision matrix that corresponds to an approximation of a proper Gaussian field (such as a proper Mat\'ern field), then its corresponding first-order intrinsic precision matrix $\Theta$ is generally dense. However, it also implies that $\Theta$ is a rank-1 update of $Q$, so sparse matrix methods may still be applied for efficient computation.

\section{Application: Kriging}\label{App:kriging}

The ordinary kriging estimate is given by \cite[Section 3.2]{cressie2015statistics}, 
\begin{equation}\label{lam_krig}
\hat u(\mv s_0) =\sum_{i=1}^k \lambda_i u(\mv s_i), \quad \quad \boldsymbol{\lambda} = \left( \boldsymbol{\gamma} + \boldsymbol{1} \frac{(1-\boldsymbol{1}^\top \Gamma^{-1} \boldsymbol{\gamma})}{\boldsymbol{1}^\top \Gamma^{-1} \boldsymbol{1}} \right)^\top \Gamma^{-1},
\end{equation}
where $\Gamma=(\Gamma_{ij})_{i,j=1,\dots,k}$ with $\Gamma_{ij}={\rm Var}(u(\mv s_i)-u(\mv s_j))$, and $\boldsymbol{\gamma}=(\gamma_i)_{i=1,\dots,k}$ with $\gamma_i={\rm Var}(u(\mv s_i)-u(\mv s_0))$.

The extrapolation behavior  when kriging with these fields can be understood through the following asymptotic result on the difference between estimated global mean, and the kriging estimate at a single location as it moves further away from observation locations. Because an intrinsic field does not have an overall mean, as a proxy, we use the conditional average value over large spherical region centered at the origin:
\[
\hat{\bar u}=\lim_{R \to \infty} {\rm E}(\bar u(R) \mid u(\mv s_1), \dots, u(\mv s_k)), \quad \text{where} \quad \bar u(R) =  \frac{1}{V_d(R)} \int_{0}^R \int_{\mathbb{S}^{d-1}} u(L \mv v) {\rm d}{\mv v} {\rm d}L,
\]
where $\mathbb{S}^{d-1} = \{ \mv v \in \mathbb{R}^d : \lVert \mv v \rVert=1\}$ is the surface of the dimension $d$ unit sphere and $V_d(L)=\frac{\pi^{d/2}}{\Gamma(d/2+1)}L^d$ is the volume of a dimension $d$ radius $L$ sphere.
In Lemma \ref{lem:ovmean} we show that under rather general conditions 
\begin{equation}\label{eq:ubarhat}
\hat{\bar u} = \frac{\mv {1}^\top \Gamma^{-1} \mv{u}}{\mv 1^\top \Gamma^{-1} \mv 1}.
\end{equation} 
We can now compare $\hat{\bar u}$ to the estimate $\hat u(L \mv v)$ for large $L$ (i.e., extrapolation), where $\mv v$ is a length 1 direction vector.

\begin{lemma}\label{lem:ovmean}
Let $\mathcal{D}=\mathbb{R}^d$, $\mv s_1, \dots, \mv s_d \in \mathcal{D}$, $\mv{u}=(u(\mv s_1), \dots, u(\mv s_k))$ and $\Gamma=(\gamma(\mv s_i, \mv s_j))_{i,j=1,\dots,k}$. If $\gamma(h)=\ell(h)h^{b}$ with $0 <b<2$ where $\ell(\cdot)$ is a slowly varying function, then \eqref{eq:ubarhat} holds.
\end{lemma}
\begin{proof}
By the formulas in \cite[Section 4.2]{cressie2015statistics}, for any finite $R$,
\[
\bar u(R) = \left( \mv \gamma(R) + \mv 1 \frac{(1-\mv 1^\top \Gamma^{-1} \mv \gamma(R)}{\mv 1^\top \Gamma^{-1}\mv 1} \right)^\top \Gamma^{-1} \mv u,
\]
where $\mv \gamma(R)$ is a vector whose entries are given below. Let $\mathbb{V}^d(\mv a, R) = \{ \mv v \in \mathbb{R}^d : \lVert \mv v - \mv a \rVert \leq R \}$ denote the sphere of radius $R$ centered at $\mv a$. We then have 
\begin{align*}
\gamma_i(R) &= \frac{1}{V_d(R)}\int_{\mathbb{V}^d(\mv 0, R)} \gamma(\lVert L \mv v - \mv s_i \rVert) {\rm d}\mv v  \\
&=\underbrace{\frac{1}{V_d(R)}\int_{\mathbb{V}^d(\mv s_i, R)} \gamma(\lVert L \mv v - \mv s_i \rVert) {\rm d}\mv v }_{(i)} \\
&\; + \underbrace{\frac{1}{V_d(R)}\left(\int_{\mathbb{V}^d(\mv 0, R) \backslash \mathbb{V}^d(\mv s_i, R)} \gamma(\lVert L \mv v - \mv s_i \rVert) {\rm d}\mv v -\int_{\mathbb{V}^d(\mv s_i, R)\backslash \mathbb{V}^d(\mv 0, R) } \gamma(\lVert L \mv v - \mv s_i \rVert) {\rm d}\mv v \right)}_{(ii)}.
\end{align*}
For \emph{(i)}, we have
\begin{align*}
(i) &= \frac{1}{V_d(R)}\int_{\mathbb{V}^d(\mv 0, R)} \gamma(\lVert L \mv v  \rVert) {\rm d}\mv v 
=\frac{1}{V_d(R)}\int_0^R \int_{\mathbb{S}^{d-1}} \gamma(\lVert L \mv v  \rVert) {\rm d}\mv v {\rm d} L \\
&=\frac{\Gamma(d/2+1)}{\pi^{d/2} R^d} \int^R_0 \frac{2 \pi^{d/2} L^{d-1}}{\Gamma(d/2)}\ell(L)L^{\beta}  {\rm d}L 
=\frac{(d/2+1) 2^{d/2} \ell(R) R^\beta}{\beta+d}.
\end{align*}
Because, when $h$ is large, $\gamma(h)$ an increasing function, we note that $(ii) \geq 0$ for $R$ sufficiently large. We now obtain the corresponding lower bound
\begin{align*}
(ii) &\leq \frac{1}{V_d(R)}| \mathbb{V}^d(\mv s_i, R)\backslash \mathbb{V}^d(\mv 0, R)| \gamma(R + \lVert \mv s_i\lVert) - | \mathbb{V}^d(\mv 0, R)\backslash \mathbb{V}^d(\mv s_i, R)|\gamma(R  - \mv s_i\lVert) \\
&= \frac{1}{V_d(R)}| \mathbb{V}^d(\mv s_i, R)\backslash \mathbb{V}^d(\mv 0, R)| (\ell(R+\lVert \mv s_i \rVert) (R+\lVert \mv s_i \rVert)^\beta - (\ell(R-\lVert \mv s_i \rVert) (R-\lVert \mv s_i \rVert)^\beta) \\
&= \frac{1}{V_d(R)} O(R^{d-1}) \times O(R^{\beta-1}) = O(R^{\beta-2}).
\end{align*}
Combining the above we then have 
\[
\mv \gamma(R) = \left( \frac{(d/2+1) 2^{d/2} \ell(R) R^\beta}{\beta+d} + O(R^{\beta-2}) \right) \mv 1, 
\]
and thus
\begin{align*}
\bar u (R)&= \left(  \frac{(d/2+1) 2^{d/2} \ell(R) R^\beta}{\beta+d}   \mv 1 + \mv 1\frac{1-\mv 1^\top \Gamma^{-1}  \frac{(d/2+1) 2^{d/2} \ell(R) R^\beta}{\beta+d}   \mv 1}{\mv 1^\top \Gamma^{-1} \mv 1}\right)^\top \Gamma^{-1} \mv u + O(R^{\beta-2}) \\
&=  \frac{\mv 1^\top \Gamma^{-1} \mv u}{\mv 1^\top \Gamma^{-1} \mv 1} + O(R^{\beta-2}),
\end{align*}
where the result then follows by noting that $\beta<2$.
\end{proof}

\noindent
\textit{Proof of Proposition \ref{prop:extrap}}
Let $\mv v \in \mathbb{R}^d$ with $\lVert \mv v \rVert=1$, and suppose $\ell(\cdot)$ is slowly varying function. First note that, for any $\mv a \in \mathbb{R}^d$ 
$\lVert L \mv v - \mv a \rVert = L - \mv a^\top \mv v + O(L^{-1})$.
Second, from the Karamatta representation
$\ell(L+c) = \ell(L)(1 + o(L^{-1}))$.
Combining these two facts we see that
\begin{equation}\label{eq:gamma_approximation}
\mv \gamma = \ell(L)L^b \mv 1 + \ell(L) L^{b-1} \mv s^{\| \mv v} + o(L^{b-1}),
\end{equation}
where $\mv s^{\| \mv v} = (s_i^{\| \mv v})_{i=1,\dots, k}$ with $s_i^{\| \mv v} = \mv v^\top \mv s_i$.
We then have
\begin{align*}
   \hat u(L \mv v) - \frac{\mv 1^\top \Gamma^{-1} \mv u}{\mv 1^\top \Gamma^{-1} \mv 1}  &= \left( \mv \gamma + \mv 1 \frac{(1- \mv 1^\top \Gamma^{-1} \mv \gamma}{\mv 1^\top \Gamma^{-1} \mv 1}  \right)^\top \Gamma^{-1} \mv u - \frac{\mv 1^\top \Gamma^{-1} \mv u}{\mv 1^\top \Gamma^{-1} \mv 1} \\
   &=\Bigg( \ell(L)L^b \mv 1 + \ell(L) L^{b-1} \mv s^{\| \mv v} \\
   &\quad + \mv 1 \frac{1- \mv 1^\top \Gamma^{-1} (\ell(L)L^b \mv 1 + \ell(L) L^{b-1} \mv s^{\| \mv v} -1)}{\mv 1^\top \Gamma^{-1} \mv 1} \Bigg)^\top \Gamma^{-1} \mv u + o(L^{b-1})\\
   &=\left( \ell(L) L^{b-1} \mv s^{\| \mv v} - \mv 1 \frac{ \mv 1^\top \Gamma^{-1} (\ell(L) L^{b-1} \mv s^{\| \mv v})}{\mv 1^\top \Gamma^{-1} \mv 1} \right)^\top \Gamma^{-1} \mv u + o(L^{b-1}) \\
   &= \ell(L) L^{b-1}(\mv s^{\| \mv v})^\top \left( I -\frac{\mv 1 \mv 1^\top \Gamma^{-1}}{\mv 1^\top \Gamma^{-1} \mv 1} \right)^\top\Gamma^{-1} \mv u + o(L^{b-1}) \\
   &= c(\mv s , \mv u, \mv v) \ell(L) L^{b-1} + o(L^{b-1}),
\end{align*}
where 
\[
c(\mv s , \mv u, \mv v) = (\mv s^{\| \mv v})^\top \left( I -\frac{\mv 1 \mv 1^\top \Gamma^{-1}}{\mv 1^\top \Gamma^{-1} \mv 1} \right)^\top\Gamma^{-1} \mv u .
\]
We now turn to the result on the extrapolation variance. There are two cases to consider: (1) when $\ell(\cdot)$ is a general slowly varying function and $b\leq 1$, and (2) when $\ell(\cdot)$ is a constant (i.e., $\ell(\cdot)\equiv \ell$) and $0<\beta<2$. We start with (1). Applying \cite[Equation (3.2.16)]{cressie2015statistics} in the first step, and collecting the leading order terms which arise from \eqref{eq:gamma_approximation} in the second step, we have
\begin{align*}
\sigma^2_k(L \mv v) &= \mv \gamma^\top \Gamma^{-1} \mv \gamma - \frac{(\mv 1^\top \Gamma^{-1}\mv \gamma - 1)^2}{\mv 1^\top \Gamma^{-1} \mv 1} \\
&= \ell(L)^2 [ L^{2b}\mv 1^\top \Gamma^{-1} \mv 1 + 2 L^{2b-1} \mv 1^\top \Gamma^{-1} \mv s] + o(L^{2b-1}) \\
&\qquad - \ell(L)^2[ \ell(L)^2 L^{2b} \mv 1^\top \Gamma^{-1} \mv 1 + \ell(L)^2 2 L^{2b-1} \mv 1^\top  \Gamma^{-1} \mv s - \ell(L) L^b ] + o(L^{2b-1}) \\
&= \ell(L) L^b + o(L^{2b-1}),
\end{align*}
where we note that when $b \leq 1$ the first term in the final equation dominates the second. In case (2), due to the more restrictive assumption on $\ell(\cdot)$ the $o(L^{b-1})$ term in \eqref{eq:gamma_approximation} can be replaced by $O(L^{b-2})$. The result then follows from similar arguments which use the fact that $\beta<2$ rather than $\beta \leq 1$.
\qed

\section{Proof of Proposition~\ref{prop:HR_cond_distr}}

\begin{proof}
Choosing any $m\in O$, by the definition of the exponent measure density in~\eqref{eqn:ExpDN} and Lemma~\ref{lem:Eqv} we obtain 
\begin{align*}
    f^{(r)}(\mv{y}' \mid \mv y ; \Theta) &= \frac{(k+k')^{-1/2} f_d((\mathbf{y}, \mv y') + \Gamma_{m \cdot} / 2 ; \Theta)}{k^{-1/2} f_d(\mathbf{y} + \Gamma_{m, O} / 2 ; \tilde \Theta)}\\
    &=  \frac{ f_{d-1}((\mathbf{y}_{\setminus m}, \mv y') - \mathbf{1} y_m + \Gamma_{m \cdot} / 2 ; \Theta^{(m)})}{ f_{d-1}(\mathbf{y}_{\setminus m} - \mathbf{1} y_m + \Gamma_{m, O} / 2 ; \tilde \Theta^{(m)})}
     \sim N_{d-1}(\mu_{U\mid O}, \Theta_{U\mid O}),
\end{align*} 
where the precision is $\Theta_{U\mid O} = \Theta_{U,U}$ and the conditional mean can be computed as 
\begin{align*}
    \mu_{U\mid O} &= y_m \mathbf{1}_U  - \Gamma_{m,U} / 2 +  \Sigma^{(m)}_{U,O\setminus\{m\}}(\Sigma^{(m)}_{O\setminus\{m\},O\setminus\{m\}})^{-1}(y_{O\setminus\{m\}} - y_m \mathbf{1}_{O\setminus\{m\}}  + \Gamma_{m,O\setminus\{m\}} / 2)\\
    & =y_m \mathbf{1}_U  - \Gamma_{m,U} / 2 -  (\Theta^{(m)}_{U,U})^{-1} \Theta^{(m)}_{U,O\setminus\{m\}}(y_{O\setminus\{m\}} - y_m \mathbf{1}_{O\setminus\{m\}}  + \Gamma_{m,O\setminus\{m\}} / 2)\\
    & =y_m \mathbf{1}_U  - \Gamma_{m,U} / 2 -  \Theta_{U,U}^{-1} \Theta_{U,O\setminus\{m\}}(y_{O\setminus\{m\}} - y_m \mathbf{1}_{O\setminus\{m\}}  + \Gamma_{m,O\setminus\{m\}} / 2)\\
    & =y_m (\mathbf{1}_U + \Theta_{U,U}^{-1} \Theta_{U,O\setminus\{m\}} \mathbf{1}_{O\setminus\{m\}}) -  \Theta_{U,U}^{-1} \Theta_{U,O\setminus\{m\}}y_{O\setminus\{m\}} -  \\
    &\Theta_{U,U}^{-1}(\Theta_{U,O} \Gamma_{m,O} / 2 + \Theta_{U,U} \Gamma_{m,U} / 2 ),
    \end{align*}
    where we used the classical formulas for conditional
distributions of multivariate normals, and the block inversion formula 
for the precision matrix.
We now note that since $\Theta$ has zero row sums, we have
\[\Theta_{U,U}^{-1} \Theta_{U,O\setminus\{m\}} \mathbf{1}_{O\setminus\{m\}} = - \Theta_{U,U}^{-1} (\Theta_{U,U} \mathbf{1}_U + \Theta_{U,m}) = - \mathbf{1}_U - \Theta_{U,U}^{-1}\Theta_{U,m}. \]
This implies that 
\begin{align*}
    \mu_{U\mid O} & =- y_m \Theta_{U,U}^{-1}\Theta_{U,m} -  \Theta_{U,U}^{-1} \Theta_{U,O\setminus\{m\}}y_{O\setminus\{m\}} - \Theta_{U,U}^{-1}\Theta_{U,\cdot} \Gamma_{m,\cdot} / 2 \\    
    & = - \Theta_{U,U}^{-1} \Theta_{U,O}y_{O} - \Theta_{U,U}^{-1}\Theta_{U,\cdot} \Gamma_{m,\cdot} / 2       
     = - \Theta_{U,U}^{-1} \Theta_{U,O}y_{O} - \frac{\Theta_{U,U}^{-1} (\Theta \Gamma)_{U,m}}{2} \\    
    & = - \Theta_{U,U}^{-1} \Theta_{U,O}y_{O} - \frac{\Theta_{U,U}^{-1} (\Theta \Gamma)_{U,m}}{2}     
     = - \Theta_{U,U}^{-1} \Theta_{U,O}y_{O} - \Theta_{U,U}^{-1} v_U.
\end{align*}
The forms are derived by classical formulas for conditional
distributions of multivariate normals, and the block inversion formula 
for the precision matrix. The last equation relies on an implication of the Fiedler--Bapat
identity \citep[Corollary 3.7]{devriendt2022a}
$\Theta \Gamma = - 2 I + 2 \mv v \mv 1^\top$
where the vector $\mv v = \Theta \diag(\Theta^+)/2 + \mv 1/d$ is called the resistance
curvature \citep{DEVRIENDT202468}; here $\Theta^+$ denotes the Moore--Penrose pseudoinverse of $\Theta$. This yields that $\Theta \Gamma$ has constant row entries (except for the diagonal) and explains why $\mu_{U\mid O}$ is independent of the choice of $m$.
\end{proof}

\section{Additional figures for the marine heat wave application}
\label{sec:sst_anom}

\begin{figure}[t!]
    \centering
    \includegraphics[width=0.32\textwidth]{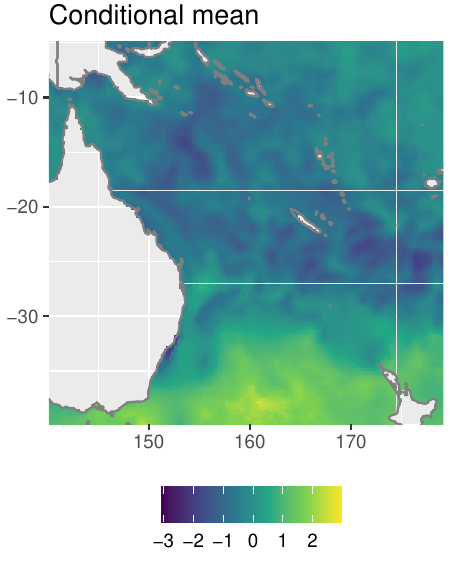}
    \includegraphics[width=0.32\textwidth]{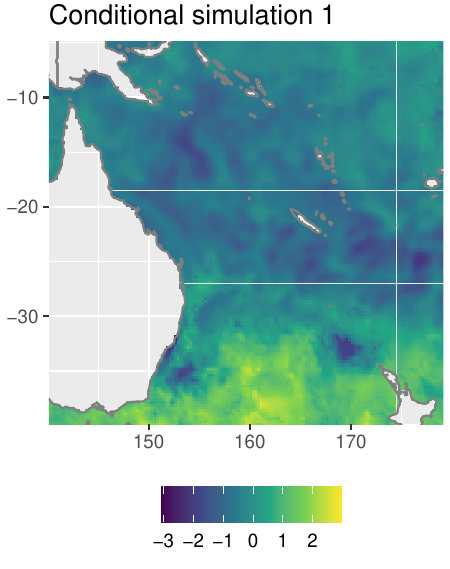}
    \includegraphics[width=0.32\textwidth]{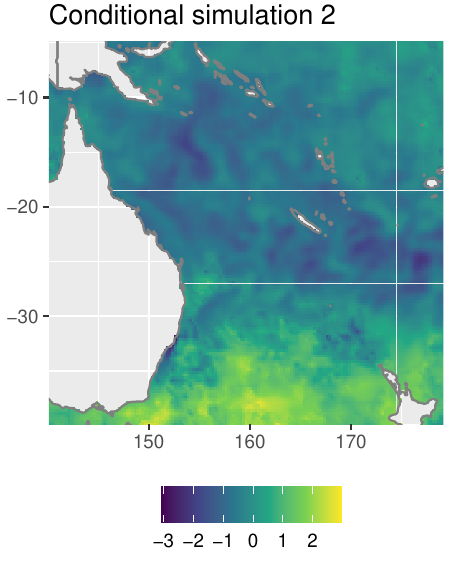}
    \includegraphics[width=0.32\textwidth]{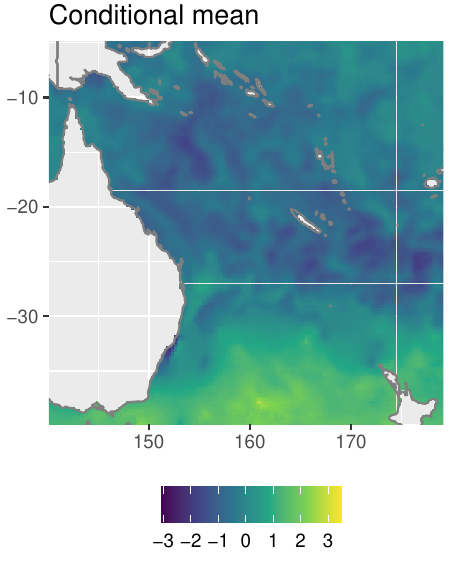}
    \includegraphics[width=0.32\textwidth]{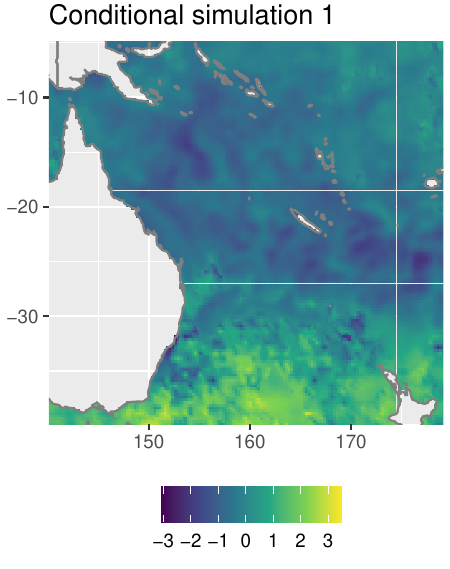}
    \includegraphics[width=0.32\textwidth]{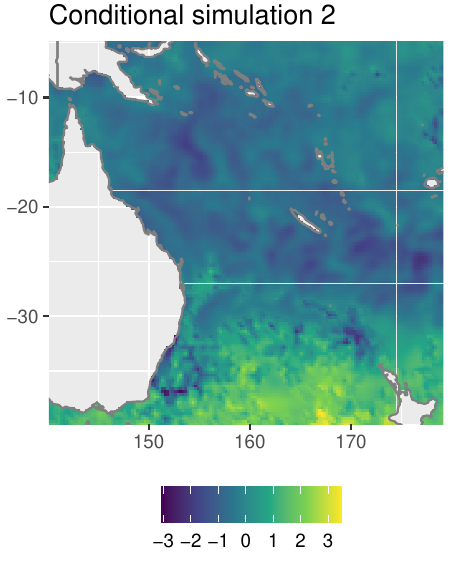}
    \includegraphics[width=0.32\textwidth]{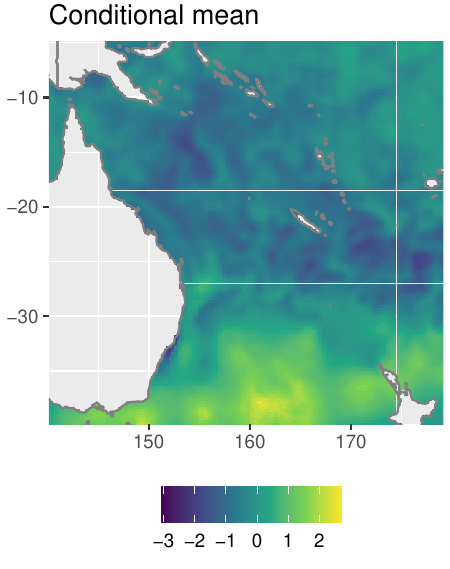}
    \includegraphics[width=0.32\textwidth]{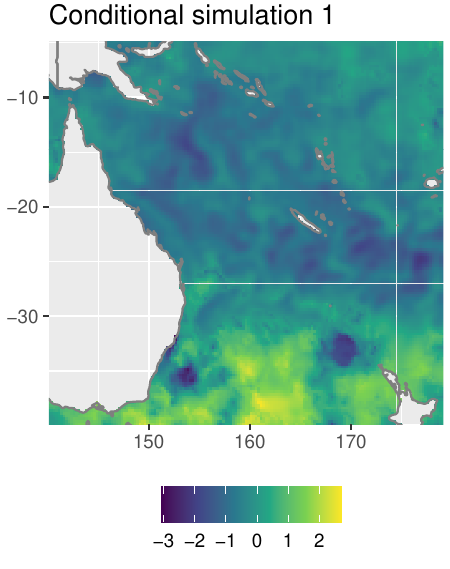}
    \includegraphics[width=0.32\textwidth]{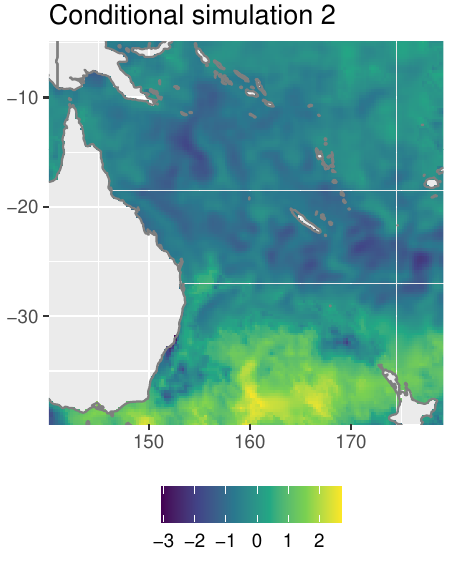}
	\caption{From left to right: Conditional mean and two conditional simulations. From top to bottom: model with $\alpha=\beta=1$, model with $\alpha=0$ and $\beta$ estimated, and model with $\alpha$ estimated and $\beta=0$.}
	\label{fig:conditionalsim2}
\end{figure}

\begin{figure}[t!]
    \centering
    \includegraphics[width=0.32\textwidth]{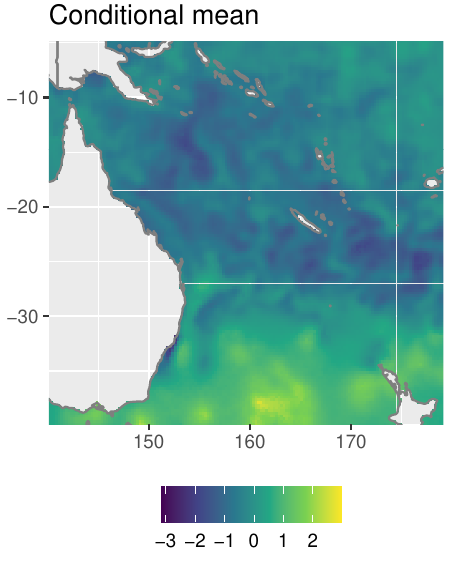}
    \includegraphics[width=0.32\textwidth]{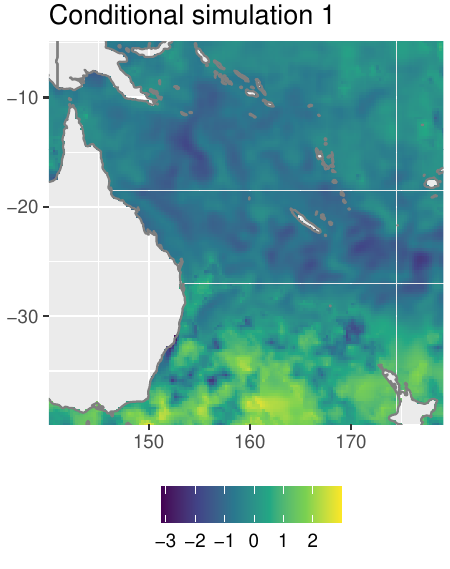}
    \includegraphics[width=0.32\textwidth]{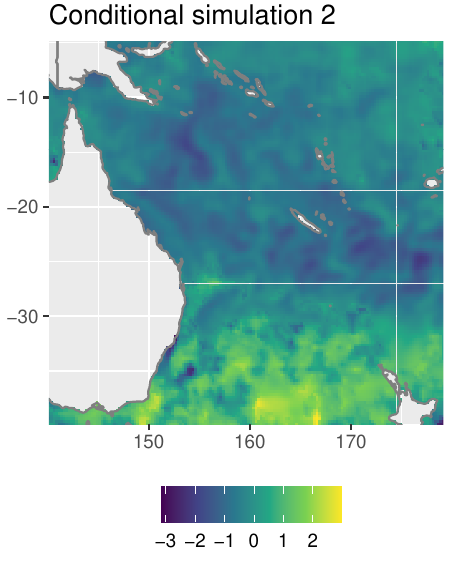}
    \includegraphics[width=0.32\textwidth]{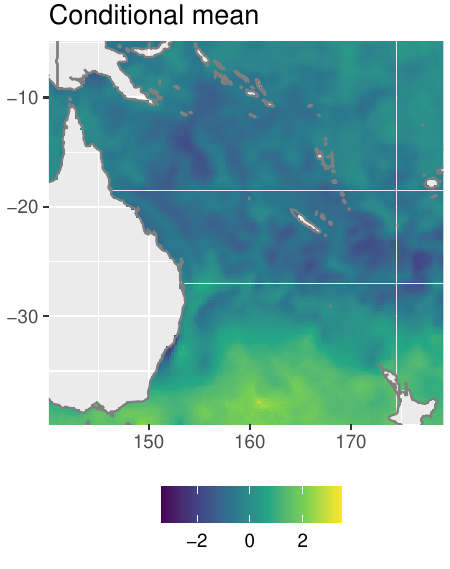}
    \includegraphics[width=0.32\textwidth]{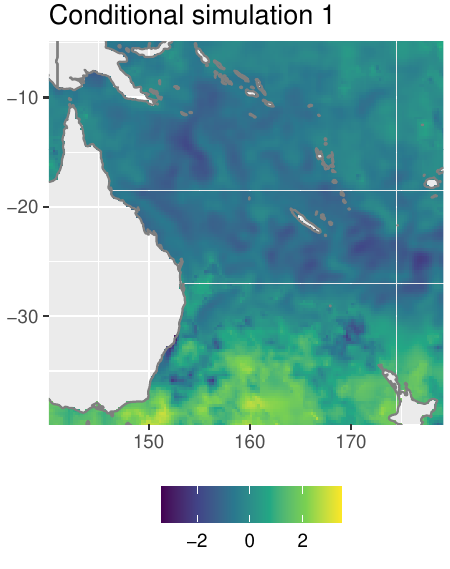}
    \includegraphics[width=0.32\textwidth]{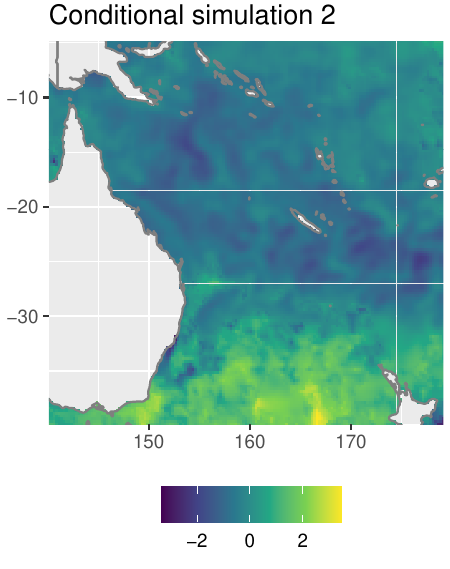}
    \includegraphics[width=0.32\textwidth]{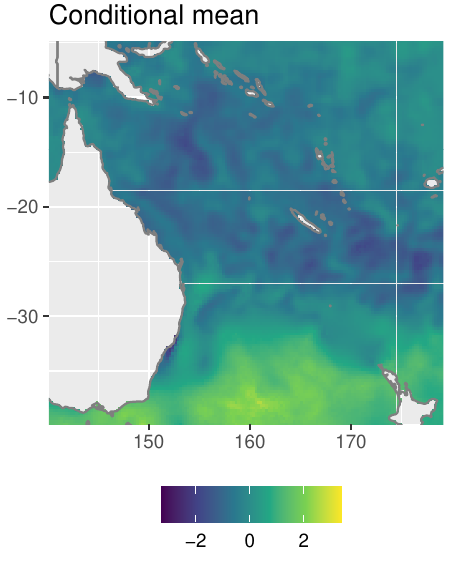}
    \includegraphics[width=0.32\textwidth]{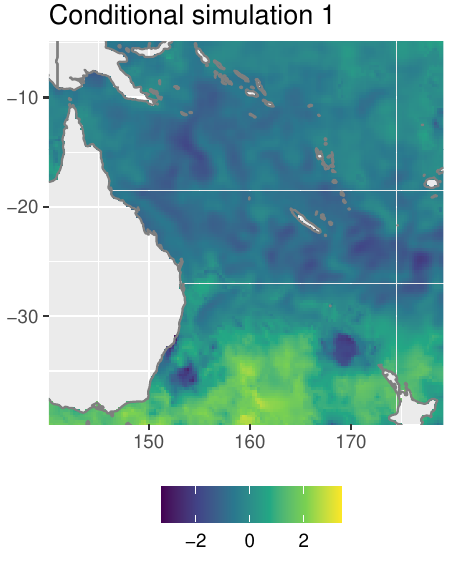}
    \includegraphics[width=0.32\textwidth]{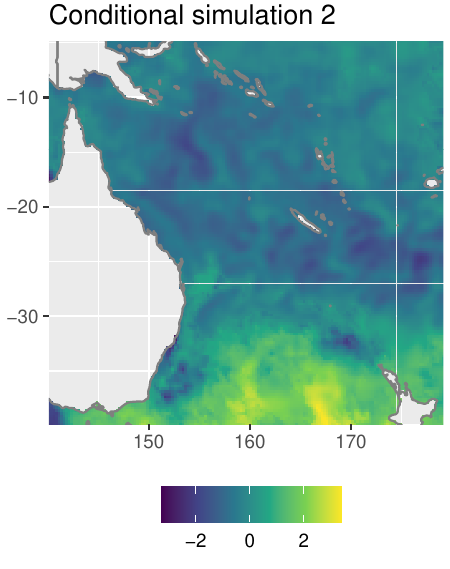}
	\caption{From left to right: Conditional mean and two conditional simulations. From top to bottom: models with $\alpha$ estimated and $\beta=0.5,1.5,2$.}
	\label{fig:conditionalsim3}
\end{figure}

\end{document}